\PassOptionsToPackage{names,dvipsnames}{xcolor}
\pdfoutput=1

\newif\ifdraft\draftfalse % draft = comments
\newif\ifanon\anonfalse   % anon = light double-blind reviewing
\newif\iffull\fullfalse   % full = includes things that were cut from
\newif\ifdiff\difffalse % Show diffs, to be included only as auxiliary material
\newif\iflongrefs\longrefsfalse % Long references (e.g. for journal)
\newif\ifbackref\backreffalse % backref option for hyperref
\newif\ifsooner\soonerfalse
\newif\iflater\laterfalse
\newif\ifieee\ieeetrue
\newif\ifcamera\cameratrue % Camera-ready version
\newif\ifappendix\appendixfalse % Appendix
\newif\ifallcites\allcitesfalse
\newif\ifneedspace\needspacetrue
\makeatletter \@input{texdirectives.tex} \makeatother

\ifieee
\documentclass[10pt, conference, compsocconf, letterpaper, times]{IEEEtran}

\DeclareMathAlphabet{\mathit}{\encodingdefault}{\familydefault}{m}{it}

\else % ACM

\ifcamera
\documentclass[acmsmall]{acmart}
\else
\ifanon
\documentclass[acmsmall,review,anonymous,nonacm]{acmart}
\else
\documentclass[acmsmall,review=false,screen,nonacm]{acmart}
\fi
\fi

\acmJournal{PACMPL}
\acmYear{2018}
\acmMonth{1}
\acmVolume{2}
\acmNumber{POPL}
\acmPrice{}
\acmArticle{65}
\acmMonth{1}
\acmDOI{10.1145/3158153}
\ifcamera\else\copyrightyear{2017}\fi %% If different from \acmYear
\ifanon
\setcopyright{none}             %% For review submission
\else
\setcopyright{rightsretained}   %% CH: for anything else
\fi

\makeatletter
\def\@copyrightpermission{\ifcamera\\\\\\\fi This work is licensed under a \href{https://creativecommons.org/licenses/by/4.0/}{Creative Commons Attribution 4.0 International License}}
\makeatother

\ifcamera\else
\makeatletter
\def\@authorsaddresses{}
\fancypagestyle{firstpagestyle}{%
  \fancyhf{}%
  \renewcommand{\headrulewidth}{\z@}%
  \renewcommand{\footrulewidth}{\z@}%
    \fancyhead[L]{\ifanon\ACM@linecountL\fi}%
}
\fancypagestyle{standardpagestyle}{%
  \fancyhf{}%
  \renewcommand{\headrulewidth}{\z@}%
  \renewcommand{\footrulewidth}{\z@}%
    \fancyhead[LE]{\ifanon\ACM@linecountL\fi\@headfootfont\thepage}%
    \fancyhead[RO]{\@headfootfont\thepage}%
    \fancyhead[RE]{\@headfootfont\@shortauthors}%
    \fancyhead[LO]{\ifanon\ACM@linecountL\fi\@headfootfont\shorttitle}%
    \fancyfoot[RO,LE]{}%
}
\def\@mkbibcitation{}

\makeatother
\fi

\fi %ieee/acm

\usepackage{booktabs}   %% For formal tables:
\usepackage{subcaption} %% For complex figures with subfigures/subcaptions
\usepackage[inference]{semantic}
\usepackage{amsmath}\allowdisplaybreaks
\usepackage{mathpartir}
\usepackage{amsthm}
\usepackage{amssymb}
\usepackage{stmaryrd} %\usepackage{MnSymbol}
\usepackage{xspace}
\usepackage{latexsym}
\usepackage{ifthen}
\usepackage{mathtools}
\usepackage{color}
\usepackage{listings}
\usepackage{verbatim}
\usepackage{tikz}
\usetikzlibrary{positioning,shadows,arrows,calc,backgrounds,fit,shapes,shapes.multipart,decorations.pathreplacing,shapes.misc,patterns,decorations.markings}
\usepackage{tikzscale}
\usepackage{tikz-qtree}
\usepackage[T1]{fontenc}
\usepackage[scaled=.83]{beramono}
\usepackage{epigraph}
\usepackage{booktabs}
\usepackage{float}
\usepackage{etoolbox}
\usepackage{wrapfig}
\usepackage{scalerel}
\usepackage{nameref}
\usepackage[strict]{changepage}
\usepackage{bm}
\usepackage{graphicx}
\usepackage[makeroom]{cancel}
\usepackage{centernot}
\usepackage{thm-restate}
\usepackage{enumitem}
\usepackage{fancyvrb}
\usepackage{varwidth}
\usepackage{combelow}

\usepackage{titlecaps}
\usepackage[normalem]{ulem}
\usepackage{cals}
\usepackage{afterpage}
\usepackage{placeins}

\ifieee
\usepackage[numbers,sort]{natbib}
\newcommand\citepos[1]{\citeauthor{#1}'s\ \cite{#1}}
\fi

\usepackage{chngcntr}
\counterwithin{figure}{section}

\definecolor{dkblue}{rgb}{0,0.1,0.5}
\definecolor{dkgreen}{rgb}{0,0.4,0}
\definecolor{dkred}{rgb}{0.6,0,0}
\definecolor{dkpurple}{rgb}{0.7,0,1.0}
\definecolor{purple}{rgb}{0.9,0,1.0}
\definecolor{olive}{rgb}{0.4, 0.4, 0.0}
\definecolor{teal}{rgb}{0.0,0.4,0.4}
\definecolor{azure}{rgb}{0.0, 0.5, 1.0}
\definecolor{gray}{rgb}{0.5, 0.5, 0.5}
\definecolor{dkgray}{rgb}{0.3, 0.3, 0.3}

\ifieee
\usepackage[
    pdftex,%
    pdfpagelabels,%
    colorlinks,%
    linkcolor=dkblue,%
    citecolor=dkblue,%
    filecolor=dkblue,%
    urlcolor=dkblue%
    \ifbackref,backref\fi%
]{hyperref}
\else
\usepackage{hyperref}
\hypersetup{
  breaklinks=true,
  citecolor=black,
  linkcolor=black,
  urlcolor=black,
}
\fi

\usepackage[capitalize]{cleveref}

\def\Snospace~{\S{}}

\def\Nnospace~{}

\usepackage{etoolbox}
\makeatletter
\patchcmd{\hyper@makecurrent}{%
    \ifx\Hy@param\Hy@chapterstring
        \let\Hy@param\Hy@chapapp
    \fi
}{%
    \iftoggle{inappendix}{%true-branch
        \@checkappendixparam{chapter}%
        \@checkappendixparam{section}%
        \@checkappendixparam{subsection}%
        \@checkappendixparam{subsubsection}%
        \@checkappendixparam{paragraph}%
        \@checkappendixparam{subparagraph}%
    }{}%
}{}{\errmessage{failed to patch}}

\newcommand*{\@checkappendixparam}[1]{%
    \def\@checkappendixparamtmp{#1}%
    \ifx\Hy@param\@checkappendixparamtmp
        \let\Hy@param\Hy@appendixstring
    \fi
}
\makeatletter

\newtoggle{inappendix}
\togglefalse{inappendix}

\apptocmd{\appendix}{\toggletrue{inappendix}}{}{\errmessage{failed to patch}}

\AtBeginEnvironment{tikzpicture}{\catcode`$=3 } % $

\newcommand{\mi}[1]{\ensuremath{\mathit{#1}}}
\newcommand{\ii}[1]{\mi{#1}}

\newcommand{\mf}[1]{\ensuremath{\mathbf{#1}}}

\newcommand{\ms}[1]{\ensuremath{\mathsf{#1}}}

\newcommand{\neutcol}[0]{black}
\newcommand{\stlccol}[0]{RoyalBlue}
\newcommand{\ulccol}[0]{RedOrange}
\newcommand{\commoncol}[0]{black}    % CarnationPink

\newcommand{\col}[2]{\ensuremath{{\color{#1}{#2}}}}

\newcommand{\src}[1]{\ms{\col{\stlccol}{#1}}}
\newcommand{\trg}[1]{\mf{\col{\ulccol }{#1}}}
\newcommand{\bl}[1]{\col{\neutcol }{#1}}
\newcommand{\com}[1]{\mi{\col{\commoncol }{#1}}}

\lstdefinestyle{codestyle}{
    basicstyle=\ttfamily,
    frame=single,
}
\newcommand{\doesprefix}{{\rightsquigarrow}^{*}}

\newcommand{\tricl}[0]{\mi{TrICL}\xspace}

\newcommand{\rsp}[0]{\textbf{\emph{RSP$^{\sim}$}}\xspace}

\newcommand{\sourcelanguage}{\src{SafeP}\xspace}
\newcommand{\targetlanguage}{\trg{Mach}\xspace}

\newcommand{\cmp}[1]{#1\!\!\downarrow}
\newcommand{\back}[1]{#1\hspace{-0.35em}\uparrow}

\newcommand{\step}[1]{\ensuremath{\overset{#1}{\longrightarrow}}}
\newcommand{\emitsdfT}[0]{\ensuremath{\rightsquigarrow_{\ms{DF}}^{*}}}

\newcommand{\projectname}{\emph{SecurePtrs}\xspace}

\newtheorem{theorem}{Theorem}[section]
\newtheorem{lemma}[theorem]{Lemma}

\newtheorem{assumption}[theorem]{Assumption}
\theoremstyle{definition}
\newtheorem{definition}[theorem]{Definition}

\crefname{assumption}{assumption}{assumptions}
\Crefname{assumption}{Assumption}{Assumptions}

\theoremstyle{definition}
\newtheorem{example}[theorem]{Example}
\theoremstyle{plain}
\newcommand{\defeq}[0]{\ensuremath{\overset{\mathsf{def}}{=}}}

\newcommand{\comm}[3]{\ifdraft{{\color{#1}[#2: #3]}}\fi}
\newcommand{\dg}[1]{\comm{dkpurple}{DG}{#1}}
\newcommand{\ch}[1]{\comm{teal}{CH}{#1}}
\newcommand{\jt}[1]{\comm{violet}{JT}{#1}}
\newcommand{\rb}[1]{\comm{orange}{RB}{#1}}
\newcommand{\aek}[1]{\comm{dkred}{AEK}{#1}} % @Akram: choose your color! :)

\newcommand{\remove}[1]{\ifdiff{\color{red}\sout{#1}}\fi}
\newcommand{\add}[1]{\ifdiff{\color{blue}#1}\else #1\fi}

\newcommand*{\EG}{e.g.,\xspace}
\newcommand*{\IE}{i.e.,\xspace}

\newcommand*{\ETC}{etc.\xspace}

\newcommand{\titleString}{\projectname: Proving Secure Compilation with\\Data-Flow Back-Translation and Turn-Taking Simulation}
\title{\titleString}
\ifneedspace
\title{\huge\bf{\titleString}
\ifneedspace\ifieee\ifanon\vspace{-1.5em}\else\vspace{-0.5em}\fi\fi\else\ifanon\vspace{2em}\fi\fi}
\else
\title{\titleString}
\subtitle{\subtitleString}
\fi

\ifanon
\author{}
\else

\ifieee\large
\author{
  Akram El-Korashy\textsuperscript{1} \quad
  Roberto Blanco\textsuperscript{2} \quad
  J\'er\'emy Thibault\textsuperscript{2} \quad
  Adrien Durier\textsuperscript{2} \quad
  Deepak Garg\textsuperscript{1} \quad
  C\u{a}t\u{a}lin Hri\cb{t}cu\textsuperscript{2} \quad
\\[1em]
  \textsuperscript{1}Max Planck Institute for Software Systems (MPI-SWS)\quad
  \textsuperscript{2}Max Planck Institute for Security and Privacy (MPI-SP)\\[0em]
}
\else %ACM

\author{AUTHOR1}
\affiliation{
  \ifcamera\institution{AFF1}\city{CITY}\country{COUNTRY}
  \else\institution{AFF1}\fi}
\email{EMAIL}

\makeatletter
\renewcommand{\@shortauthors}{SHORTAUTHORS}
\makeatother

\fi %ieee/acm

\fi % anon

\begin{document}

%% Paper note
%% The \thanks command may be used to create a "paper note" ---
%% similar to a title note or an author note, but not explicitly
%% associated with a particular element.  It will appear immediately
%% above the permission/copyright statement.
%\thanks{with paper note}                %% \thanks is optional
                                        %% can be repeated if necessary
                                        %% contents suppressed with 'anonymous'

\ifieee\maketitle\fi

%% Abstract
%% Note: \begin{abstract}...\end{abstract} environment must come
%% before \maketitle command
%%%%%%%%%%%%%%%%%%%%%%%%%%%%%%%%%%%%%%%%%%%%%%%%%%%%%%%%%%%%%%%%%%%%%%%%%%%%%%%%

%\ch{Let's please call it ``secure compilation'' everywhere, not ``compiler
%  security''.  First, ``secure compilation'' is more standard; it's what we used
%  for our workshop (PriSC), Dagstuhl seminar,
%  \href{https://blog.sigplan.org/2019/07/01/secure-compilation/}{blog post},
%  previous papers (by us and other people from our community), etc. Second,
%  secure compilation is about the whole compilation chain, including the source
%  and target languages, the linker, etc.  It's not solely about the compiler, as
%  ``compiler security'' makes it sound.  In fact, if you look at the current
%  work, there's little security happening in the compiler. Security is all
%  enforced at the target level, and that would also be the case if we targeted a
%  machine with HW capabilities or micro-policies, for instance.
%\aek{Ok. Agreed. I was the one trying to use ``compiler
%security'' instead. Thanks for clarifying this point.}
%}

\begin{abstract}
Proving secure compilation of partial programs typically requires back-translating an attack against the compiled program to an attack against the source program. To prove back-translation, one can syntactically translate the target attacker to a source one---i.e., syntax-directed back-translation---or show that the interaction traces of the target attacker can also be emitted by source attackers---i.e., trace-directed back-translation.

Syntax-directed back-translation is not suitable when the target attacker may use unstructured control flow that the source language cannot directly represent. Trace-directed back-translation works with such syntactic dissimilarity because only the external interactions of the target attacker have to be mimicked in the source, not its internal control flow. Revealing only external interactions is, however, inconvenient when sharing memory via unforgeable pointers, since information about shared pointers stashed in private memory is not present on the trace. This made prior proofs unnecessarily complex, since the generated attacker had to instead stash all reachable pointers.

In this work, we introduce more informative \emph{data-flow traces}, combining the best of syntax- and trace-directed back-translation in a simpler technique that handles both syntactic dissimilarity and memory sharing well, and that is proved correct in Coq. Additionally, we develop a novel \emph{turn-taking simulation} relation and use it to prove a recomposition lemma, which is key to reusing compiler correctness in such secure compilation proofs. We are the first to mechanize such a recomposition lemma in the presence of memory sharing.

We use these two innovations in a secure compilation proof for a code generation compiler pass between a source language with structured control flow and a target language with unstructured control flow, both with safe pointers and components.

\end{abstract}

\section{Introduction}\label{sec:intro}
% !TEX root = ./paper.tex

Compiler correctness, a.k.a. semantics preservation, is the current gold
standard for formally verified
compilers~\cite{milner1972proving,Leroy2009,KumarMNO14,patterson2019next}.
However, compiler correctness alone is insufficient for reasoning
about \add{the} security of compiled partial programs linked with arbitrary
target contexts (\EG components such as libraries) because
compiler correctness shows that the compiled program simulates the
source program, only under the \emph{assumption} that the target
context obeys all restrictions of
the \emph{source language} semantics, \IE it does not perform any
low-level attacks disallowed by the source language.
This assumption is usually false in practice: compiled programs are
routinely linked with arbitrary, unverified target-language code
that may be buggy, compromised, or outright malicious in contravention
of source semantics.
In these cases, compiler correctness (even in compositional
form~\cite{KangKHDV15,10.1145/2784731.2784764,10.1145/2676726.2676985,10.1145/3371091}),
establishes no security guarantees for compiled partial programs.

% CH: was careful to no longer imply that the only kind of secure compilation
%     is the kind we do here
This problem \add{can be} addressed by \emph{secure
compilation}~\cite{abate2019journey,Patrignani:2019:FAS:3303862.3280984},
\remove{wherein one shows}\add{by enforcing}
% CH: It's not the proof that solves the problem above, but the enforcement!
that any violation of a security property of a
compiled program in some target context also appears for the
source program in some source context. Formally, this requires proving
the existence of a property-violating source context given a
target-level violation and the corresponding violating target
context. This proof step, often called \emph{back-translation},
is crucial for establishing that a vulnerable compiled program only arises
from a vulnerable source program, thus preserving the
security of partial source programs even against adversarial target contexts.
Although there is a long line of work on proving secure compilation for
prototype compilation chains
% formal definitions of secure compilation
that differ in the specific security properties preserved and
the way security is \remove{practically}
% CH: most of this work is not that practical (and we already call them prototype)
enforced~\cite{abate2019journey,
Patrignani:2019:FAS:3303862.3280984,DBLP:journals/lmcs/DevriesePPK17,patrignani2015secure,popl-backtrans,capableptrs,Skorstengaard:2019:SEW:3302515.3290332,DBLP:journals/pacmpl/StrydonckPD19,busi2020provably,Abate:2018:GCG:3243734.3243745,TsampasStelios2019Tsfs,ahmedCPS,AhmedFa,New,Patrignani:2019:FAS:3303862.3280984},
back-translation is a common, large element of such secure compilation proofs.

Back-translation \remove{can be}\add{is usually} done
in \add{one of} two different ways:
% \ch{I think we should be more careful that we don't imply that this or covers all cases.
%   After the Dagstuhl talk people complained that also interpreting (instead of
%   compiling) the target context could be yet another way to back-translate.
%   Akram thinks that an interpreter is still syntax-directed:
%   \url{https://secure-compilation.zulipchat.com/\#narrow/stream/215770-capabilities/topic/Interpreter.20for.20back-translation}
%   But still, it seems safest to not claim exhaustiveness.}
%
\emph{syntax-directed} or \emph{trace-directed}.
Syntax-directed back-translation defines a function
from the violating target context (a piece of syntax) to a
source context, \add{basically treating
back-translation as a target-to-source compiler}.
While this approach is easy to use in some
situations~\cite{DBLP:journals/lmcs/DevriesePPK17,popl-backtrans,Skorstengaard:2019:SEW:3302515.3290332,DBLP:journals/pacmpl/StrydonckPD19,TsampasStelios2019Tsfs,ahmedCPS,AhmedFa,New,Patrignani:2019:FAS:3303862.3280984},
% \footnote{\add{And
% its correctness proof can also benefit from standard ideas
% from the literature on compiler correctness}
% CH: Unconvinced by this footnote. Didn't see any use of standard compiler techn
% AEK: Fine. My statement was speculation actually. But I am not aware of any concrete benefit from standard compiler literature that I can point to.
it has a significant limitation: it cannot be used if some
constructs of the target language cannot be easily mimicked
% was "simulated"
%\aek{mimicked? Let's reserve ``simulation'' for ``simulation proof''?} CH: +1
in the source language. For example,
% syntax-directed back-translation
it is not well suited when the source language only has structured
control flow, while the target language has unstructured control
flow (goto or jump), as representing unstructured
control flow in the source
would require complex transformations
or rely on heuristics that may not always
work~\cite{unstructured-to-structured,Myreen:2008:MVM:1517424.1517444}.\footnote{\add{An alternative to
representing unstructured control flow in the source could be to write
an emulator for the target
language \emph{in the source language}, but we think that proving such an emulator correct in a proof 
assistant would be a challenging undertaking.}
\ifsooner
\ch{Some good citation(s) would be nice.
  At least people like Dominique seemed to say this has been done, 
  and Catalin was vaguely remembering that too, but we didn't find a good reference, did we?
  \url{https://secure-compilation.zulipchat.com/\#narrow/stream/215770-capabilities/topic/Interpreter.20for.20back-translation}}
\fi
}
% \jt{This seems like circular reasoning: ``back-translating a context with unstructured control flow is difficult
% because back-translating a context with unstructured control flow is difficult''}%
% \aek{better now?}
% CH: looks well to me
Yet this kind of a difference between source and target
languages is commonplace, e.g., when compiling any block-structured
language to assembly.

In contrast, trace-directed back-translation works by defining a
target \emph{interaction trace semantics} that represents all
the interactions between the compiled program and its context (e.g.,
cross-component calls and returns)
% but not their internal behaviors, -- CH: too many detours in this long phrase, anyway explained below
and constructing the violating source context from the violating target
trace instead of the target
context~\cite{capableptrs,Abate:2018:GCG:3243734.3243745,patrignani2015secure,10.1145/3436809}. This
has the advantage of not having to mimic
the internal behavior of
the target context in the source language.
%
% CH: to me the following seemed spurious (repeats what was said above):
% Instead, only its end-to-end interaction with the compiled program has to be
% mimicked by the source context.
%
So in contrast to syntax-directed back-translation, this
trace-directed method
works well even when some target language construct cannot be easily
mimicked in the source language, as long as the construct's effect
does not cross linking (program-context) boundaries.

Although very powerful in principle, trace-directed back-translation is
rather understudied for settings where the program and its context
can \emph{share private memory} by passing pointers or references to each
other, something that is common in practical languages like Java, Rust,
and ML. There is a good reason for this relative paucity of work:
memory sharing is a source of interesting interaction between the program and
its context, so allowing it makes the definition of
traces~\cite{capableptrs,9249fdc52dd3414e83b2e3f9a89cb117}, the
back-translation, and the proof of secure compilation significantly
more complex. Moreover, as we explain below, memory sharing changes
parts of the proof conceptually and needs fundamentally new
techniques.

This is precisely the gap that this paper fills: it significantly
advances proofs of secure compilation \add{from a memory-safe
source language with memory sharing to a target language that
provides fine-grained memory protection.} For this, we introduce two new proof techniques:
(1)~\emph{data-flow back-translation}, \remove{which is a simpler form
of back-translation}\add{a form of back-translation that is simpler
and more mechanization-friendly than the closest prior work~\cite{capableptrs}},
% based on data-flow traces,
and (2)~\emph{turn-taking simulation},
which\remove{simplifies} 
\add{we used to adapt}\remove{the remaining} 
\add{another key lemma in the }secure compilation proof
\add{to memory sharing}. Next, we
briefly \remove{describe why these new techniques are needed and what they do}\add{explain the need
for these new techniques}.

% \ch{Even from the start of the paragraph above I'm already worrying about the
%   safety of pointers in the target.  How do we early eliminate the worry that
%   the target context could walk all over memory to break security? Or to put
%   this differently, how do we say early on that we assume safe pointers in the
%   target for enforcement, but the problems of proving are still very big?
%   Otherwise the worry of the target context walking all over memory will only
%   get worse reading on; for instance complicated stashing attacks in the target
%   like the ones you describe below are not at all needed if one can forge
%   pointers and walk all over memory in the target, which is normally the case
%   pretty much everywhere at the ASM level! (with the exception of CHERI
%   or properly programmed PUMP)}
% \aek{Agree a bit. But I do not think it 
% is too late. I think the clarification
% that safety of pointers in the target
% is prerequisite can happen in the next
% paragraph. Anyways, see my comment below for a 
% possible fix?}
% CH: Very nice, I like the first phrase below. Thank you!

%% \ch{If we like the current title we should switch to the precise terms there}

\paragraph*{Data-flow back-translation}
\remove{To understand the need for this technique, consider a situation where
source language has memory safety and}
\add{Consider a compiler from a memory-safe source language that }prevents pointer forging (as in
Java, Rust, or ML) \add{to a target that provides fine-grained memory protection,
such as a capability machine~\cite{cheriageofrisk,chericompartment} or a
tagged-memory architecture~\cite{pump,micropolicies}}.
Suppose a compiled program has shared a
pointer to its private memory with the co-linked target context in the
past, and the context has stored this pointer somewhere in its \add{(\IE the context's)} private
memory. \remove{At some later point}\add{Later}
 in the execution, the context may use
a \emph{chain of memory dereferences within its private memory} to
recover this \add{shared} pointer and write \remove{to the memory to which it points}\add{through it}.
Since this write changes \emph{shared} memory, it must be recorded on the
trace and must be mimicked by the source context constructed by
back-translation. To mimic this write in the source language, the
back-translated source context cannot forge a pointer.
Instead, it must follow a similar chain of dereferences in the source
to the one used by the target context.
%\aek{Here, we can add ``\add{ assuming the target context is prevented from
%		forging pointers too---an assumption
%		that is realizable on some recent architectures \cite{cheriageofrisk,chericompartment,pump,micropolicies}}''}\ch{I don't see why this assumption about the target
%  is needed/helpful in this sentence about the source.
%  So for me adding this {\em here} would only add confusion.
%  Can we maybe put this information where it's more relevant?
%  For instance as a footnote in the ``Later in the execution'' sentence above?}
%
However, the chain of memory dereferences leading to this pointer is
in the target context's \emph{private} memory and interaction traces
omit these private dereferences by design!

Consequently, information needed to
reconstruct \emph{how} to access the shared pointer is missing from
interaction traces, which\remove{makes the back-translation extremely difficult}
\add{led prior work \cite{capableptrs,10.1145/3436809}
to have the back-translation perform complex bookkeeping
in order to reconstruct this missing information}.
\remove{Prior work that has even considered this situation 
%\cite{capableptrs,10.1145/3436809}
relied on extensive bookkeeping to
reconstruct this missing information, which is unwieldy and complex.}
For instance, the source context generated by \citet{capableptrs} had
to fetch all reachable pointers every time it got control and \remove{store}\add{stash}
them in its internal state. 
This required complex simulation
invariants, on top of the usual invariants between the states of the
target and source contexts. \add{We think this informal stashing
  approach is unnecessarily complex and
  would be difficult to mechanize in a proof assistant (see
  \autoref{sec:key-ideas-data-flow-traces})}.

\remove{This is where our new idea of data-flow back-translation comes in.}
To back-translate, we
\add{instead }first enrich the standard interaction traces with
information about data-flows \emph{within} the context.
% this breaks the existing trace abstraction selectively,
% \jt{What does selectively mean? Our back-translation is applied to whole
% programs, not just contexts; but this is not mentionned before.}
% \aek{I think selectively means not getting rid of the whole abstraction, i.e., not exposing control flow for example.}
% CH: don't think this abstraction breaking discussion is very helpful anyway, so commented out
This considerably simplifies the back-translation definition by
providing precisely the missing chain of private memory dereferences
in the trace itself.  
\add{The data-flow back-translation function then translates
each such dereference (or in general each data-flow event) one
by one to % one or more predefined
simple source expressions. 
For each data-flow event, we prove the correctness 
of its back-translation, \IE that its corresponding
predefined source expression keeps \emph{source} 
memory related to the \emph{target} memory.
The proof
relies on just setting up an invariant between the 
\emph{target} 
memory 
that now appears in each data-flow event and the \emph{source}
memory in the
state after executing the source expression
obtained by back-translating the given event.
Crucially, proofs about stashing all reachable pointers are not needed any more.}

We see data-flow back-translation as a sweet
spot between standard trace-directed back-translation, which abstracts
away all internal behavior of the context, and syntax-directed
back-translation, which mimics the internal behavior of the context in
detail, \add{but which cannot handle syntactic dissimilarity well}.

%% The usefulness of data-flow traces is not limited to memory
%% sharing. They are useful for trace-directed back-translation in any
%% setting, with or without memory sharing, where pointers known to the
%% context cannot be forged but influence the program-context
%% interaction in some way.
%% \aek{I do not have a concrete example in mind of what this setup
%%   looks like.}\ch{This para also seems quite mysterious to me; if we want
%%   to keep it we need to make it more concrete.}

\paragraph*{Turn-taking simulation\remove{s}}
Turn-taking simulation is useful when one tries to \emph{reuse}
compiler correctness as a lemma in the secure compilation proof to
\add{separate concerns and}
avoid duplicating large amounts of work. \add{Specifically, 
one defines a simulation relation between
the run of the \emph{source} program in the 
back-translated source 
context on one hand and the property-violating run of the
\emph{compiled} program in the target context on the other.}
\remove{after
(trace-directed) back-translation has been defined, one still has to
prove that the source program and the back-translated source context
actually reproduce the given (property-violating) trace of the
compiled program and the target context. 
This is a difficult
simulation proof over reduction steps, but many }\add{Some }of the source and
target steps are executed by the source program and its compilation
and \remove{these two} are \add{thus} already related by\remove{ the statement of the} compiler
correctness\remove{ theorem}.
Having to reprove the simulation for these steps
would be tantamount to duplicating an involved compiler correctness proof~\cite{Leroy2009}.
This duplication can be avoided by proving a
\remove{simpler}\add{so-called} 
\emph{recomposition} lemma in the target
language\add{, as proposed by} \citet{Abate:2018:GCG:3243734.3243745}.
Intuitively, recomposition says
that if a program $P_1$ linked with a context $C_1$, and a program
$P_2$ linked with a context $C_2$ both \remove{produce}\add{emit} the same trace, then
one may \emph{recompose}---link $P_1$ with $C_2$---to obtain again the
same trace.

The proof of recomposition is a ternary simulation between the runs of
$P_1 \cup C_1$, $P_2 \cup C_2$, and the recomposed program $P_1 \cup
C_2$.  The question that becomes nuanced with memory sharing is how
should the memory of the recomposed program be related to those of the
given programs in this simulation. \emph{Without memory
sharing}, this is straightforward: at any point in the simulation, the
projection of $P_1$'s memory in the recomposed run of $P_1 \cup C_2$
will equal the projection of $P_1$'s memory from the run $P_1 \cup
C_1$ (and dually for $C_2$'s memory). With memory sharing, however,
this simple relation does not work because $C_2$ may change parts
of $P_1$'s shared memory in ways that $C_1$ does not.  Specifically,
while control is not in $P_1$, the projections of $P_1$'s memories in
the two runs mentioned above will not match.
%% $C_1 \cup P_1$, $C_2 \cup P_2$, and the recomposed program $C_1 \cup
%% P_2$.  The question that becomes nuanced with dynamic memory sharing is how
%% the \emph{memory} of the recomposed program should be related to those
%% of the given programs in this simulation. In the \emph{absence of
%% memory sharing}, this is straightforward: at any point in the
%% simulation, we can project $C_1$'s private memory from the run of
%% $C_1 \cup P_1$ and $P_2$'s private memory from the run of $C_2 \cup
%% P_2$, and juxtapose them to get exactly the memory of $C_1 \cup
%% P_2$.
%% However, with memory sharing, this simplistic relation does not
%% work because $C_1$ may change parts of $P_2$'s memory that are shared
%% in the run $C_1 \cup P_2$, while $C_2$ may not necessarily share the same pointers
%% in the run $C_2 \cup P_2$, so projecting and juxtaposing would not even
%% produce a valid memory because of dangling pointers.

This is where our turn-taking simulation comes in. We relate the
memory of $P_1$ from the run of $P_1 \cup C_2$ to that from the run of
$P_1 \cup C_1$ only while control is in $P_1$. When control shifts to
the contexts ($C_2$ or $C_1$), this relation is limited to $P_1$'s
private memory (which is not shared with the context). The picture for
$C_2$'s memory is exactly dual. Overall, the relation takes ``turns'',
alternating between two memory relations depending on where the
control is. This non-trivial
% replacement for the juxtaposition-based\ch{The word
%   ``juxtaposition'' was not used so far (bad flow). Can we refer to a word that was actually used above?}
relation allows us to prove recomposition and therefore reuse a
standard compiler correctness result even with memory sharing.

\paragraph{Concrete setting}
We illustrate our two new proof techniques by extending an existing mechanized
secure compilation proof by \citet{Abate:2018:GCG:3243734.3243745} to cover
dynamic memory sharing.
The \remove{original proof was done for a compilation pass }\add{compilation pass we extend goes }from an imperative
source language with structured control flow (\EG calls and returns,
if-then-else) to an assembly-like target 
\remove{language }with
unstructured jumps.
Both languages had components and \add{safe }pointers\remove{, and
for the purpose of this compilation pass, 
pointers in both languages were assumed to be safe}---\IE out of bound accesses are errors that
stop execution.\footnote{While this is orthogonal to our current work on proof
  techniques, such safe pointers can be efficiently implemented using, for
  instance, hardware capabilities~\cite{cheriageofrisk,chericompartment} or
  programmable tagged architectures~\cite{pump,micropolicies}.}
In both languages, the program and the context had their own
private memories, and pointers to these memories could not be shared
with other components.
The program and the context interacted only by calling each
other's functions, and passing only primitive values via call
arguments and return values.

We extend both languages by allowing their safe pointers to be passed to and
dereferenced by other components, thus introducing dynamic memory sharing.
We then prove that this extended compilation step is secure
with respect to a criterion called ``robust safety
preservation''~\cite{10.1145/3436809,abate2019journey,DBLP:conf/esop/AbateBCD0HPTT20}.
For this, we apply our two new techniques, data-flow
back-translation and turn-taking simulation\remove{s}. Since the parts of the
proof using these new techniques are fairly involved and non-trivial,
we also \add{fully} mechanize them in the Coq proof assistant.

\paragraph{Summary of contributions:}
% \medskip
% In summary, we make the following {\bf contributions}:
\begin{itemize}
\item
We introduce data-flow back-translation and turn-taking simulation\remove{s},
two new techniques \add{for proving secure compilation from
memory-safe source languages to target languages with fine-grained
memory protection}, when the languages support memory sharing and may
be syntactically dissimilar.

% CH: old bullet, seemed too unfocused to me
% We observe that current proof techniques for secure compilation using
% on trace-directed back-translation (which is forced in settings where the
% source and target languages differ significantly) do not scale to
% memory sharing. We introduce two new conceptual techniques --
% data-flow back-translation and turn-taking simulations -- to fill this
% gap.
\item
We apply these conceptual techniques to prove secure compilation for a
code generation pass between a source language with structured control
flow and a target with unstructured control flow. Both
languages \add{have safe pointers only}, and in both memory is
dynamically shared by passing safe pointers between components.

\item
% \ch{This point now sounds too weak to me, given how much of {\bf secure compilation} we fully mechanized now.}
% \ch{In particular, for someone not knowing how many other steps are there in this proof
%   that weren't formalized (basically just compiler correctness), this can now easily give the wrong impression 
%   that only a small part of the proof was formalized.}
\remove{We mechanize (in Coq) the parts of our proofs centered around these two key proof techniques.}
\add{We formalize this secure compilation proof in Coq,
focusing on back-translation and recomposition,
which illustrate our techniques and which we fully mechanized.}
\end{itemize}

%% \ch{One important point from the old intro that is now missing
%%   is that these safe pointers can be efficiently implemented given
%%   HW support for capabilities or metadata tags. That's not what
%%   we're focused in this paper, but mentioning this early enough
%%   might help set expectations right. We could also return to this
%%   point in conclusions/future work.}

\paragraph{Mechanized proof}
\ifanon
The anonymous supplementary material uploaded with this
submission includes Coq proofs of back-translation and recomposition for the
compilation pass outlined above.
\else
The Coq proof\remove{s} of
\remove{back-translation and recomposition}\add{secure compilation}
% \ch{Same here, the start here sounds too weak now.
%   The fact that we have a mechanized {\bf secure compilation} proof that is complete under
%   the assumptions mentioned below never becomes clear enough.}
for the compilation pass outlined above \remove{are}\add{is} available
  \ifappendix at\\
  {\tt\url{https://github.com/secure-compilation/SecurePtrs}}
  \else as supplementary material uploaded with this submission.
  \fi
\fi
The size of \remove{these two}\add{the back-translation and recomposition}
proof steps---\add{which we fully mechanized and which constitute
the biggest and most interesting parts of the proof---}is 3k lines
of specifications and 29k lines of proof.
\ifsooner\ch{Need to update these numbers again if we ever finish some last simplifications.}\fi
For comparison, in the Coq development without memory sharing on which we are
building~\cite{Abate:2018:GCG:3243734.3243745}, these two steps were 2.7k
lines in total\remove{, so an order of magnitude smaller.
We believe that a significant part of this increase in proof size can be attributed to the
increase in the conceptual difficulty of the two proofs  in the
presence of shared memory}.
%\rb{The comparison isn't very apples-to-apples as the old proofs were made much more compact (it used to be that the old incomplete proof of composition alone was around 5k lines, as I recall) --- maybe it's ok not to worry about this too much anyway?}
%\aek{I agree a bit.
%And in any case,
%I actually would prefer we remove 
%all of the paragraph that talks about
%proof size. 
%We are using the order-of-magnitude difference
%in size as evidence to support the argument 
%that allowing memory
%sharing makes the problem challenging, 
%however reviewer A has 
%used (partly rightly so) the same piece of
%evidence to
%counter the other argument we're making, 
%which is that the novel proof
%ideas simplify our life.}

% \ch{Here we can quickly report on the readiness status of the Coq proofs
%   in general, with more precise disclaimers coming up later if needed.}

\add{As in \citet{Abate:2018:GCG:3243734.3243745}, }our mechanized
\add{secure compilation} proof\remove{s currently}
assume\add{s} \remove{not }only standard axioms (excluded
middle, functional extensionality, \ETC) 
\remove{but also some low-level specifications
about the data structure we use for memory maps, as well as about allocation,
reachability, and well-formedness of trace events }\add{and axioms about whole-program compiler 
correctness that are mostly standard and stated in the style of corresponding
CompCert theorems} (these are all documented \add{in \autoref{sec:axioms}}
 and the included {\tt README.md}).
\remove{We believe that with a bit of extra effort these low-level specifications used
transitively in our proofs can be proved as well.}%
\remove{Even in the current state though,}%
\add{Compared to}
previous paper proofs of secure compilation with memory
sharing~\cite{10.1145/3436809,capableptrs},
\add{all details of} our proofs are
\add{mechanized with respect to the
clear axioms mentioned above.}%
\remove{done in much 
greater detail, which gives us much higher confidence}
We found that the use of a proof assistant was vital in 
getting \add{all} the
invariants right \add{and alleviating the human burden of
checking our proof}.
\remove{checking the thousands of lines of proof of all
the relevant simulation lemmas, some of which appear later in this
paper.}

\newcommand{\supp}{
\paragraph{Supplementary material}
\add{
Beyond our mechanized Coq proof, the supplementary material of this submission includes:
(i) a {\tt summary\_of\_changes.md} file describing how we addressed each
requested change and then answering to individual reviewer comments;
(ii) a long version of the paper that: (ii\nobreakdash-a)~explicitly marks text which was
added or removed usually in response to reviewer feedback and (ii\nobreakdash-b)~includes our
appendices with some extra material that was requested by the reviewers, but
that did not fit in the page limit of this submission (we will include this appendix on arXiv).
}
}
%% \ifappendix\else\supp\fi
\ifdiff\supp\fi

\paragraph{Outline}
The rest of the paper is organized as follows:
In \autoref{sec:background} we illustrate our secure compilation criterion and outline
a previous proof~\cite{Abate:2018:GCG:3243734.3243745} that did not support memory sharing.
In \autoref{sec:keyideas} we explain the challenges of memory sharing,
introduce data-flow back-translation and turn-taking simulation,
and show how they fit into the existing proof outline.
In \autoref{sec:compiler} we show the
source and target languages to which we apply these techniques.
\autoref{sec:back-trans} \add{and \autoref{sec:recomb} }provide details of applying
data-flow back-translation \add{and turn-taking simulation} to our
setting\remove{in the appendix does the same for turn-taking
simulation,} and \autoref{sec:axioms} explains our assumptions.  Finally,
we discuss related work (\autoref{sec:related-work})\add{, scope and limitations
(\autoref{sec:limitations})}, and future work
(\autoref{sec:conclusion}).
%\dg{Fill this in once
%the paper's structure is stable. Leave the empty space below as is
%till this text is filled.}\ch{TODO: Akram, can you fill this?}
%\vspace{40mm}

\section{Background}
\label{sec:background}
% !TEX root = ./paper.tex

We start with\remove{giving some background:} a motivating
example explaining the broad setting 
we work with (\autoref{sec:background-example}), the formal
secure compilation criterion we prove (called robust safety preservation;
\autoref{sec:rsp}),
and a proof strategy from prior work on which we build
(\autoref{sec:background-proof-strategy}).

\subsection{Motivating Example and Setting}
\label{sec:background-example}
% !TEX root = ./paper.tex

\ifsooner
\ch{Not sure whether this was already discussed, but I have the feeling that
  this simplified example is a bit watered down. It does show a reasonable
  attack, but not the kind of attack that could appear in our target language.}
\fi

Broadly speaking, we are interested in the common scenario where a \emph{part} of a
program is written in a \emph{memory-safe source language}, compiled to a
target language and then linked against other target-language program
parts, possibly untrusted or prone to be compromised, to finally
obtain an executable target program. By ``program part'', we mean a
collection of components (modules), each of which contains a set of
functions. These functions may call other functions, both within this
part and those in other parts. We use the terms
``program'' and ``program part'' to refer to the program part we
wrote and compiled, and ``context'' to refer to the remaining,
co-linked program part that we didn't write.
\ifsooner\ch{Still not happy with this confusing naming scheme,
  in which the context is a ``program part'', which also means ``program''}\fi

As an example, consider the following source program part, a single
component, \src{Main}, which implements a \remove{server-side
function \src{set\_ads\_image} that prepares a page to be shown to the
end-user. The function calls a helper
function \src{populate\_partner\_ads} which is implemented by a
third-party library from an advertising company.}\add{\src{main}
function that calls two other functions \src{Net.init\_network}
and \src{Net.receive}, both implemented by a third-party 
networking library \src{Net} (not shown).}

%% Consider the scenario where the author of
%% a program uses a \emph{safe source language}
%% but writes only a program part 
%% (\IE a set of components that has external dependencies)
%% like the part shown in the snippet below
%% that implements a server-side
%% function \src{set\_ads\_image} that prepares
%% the view to be shown to the end-user. 
%% The function \src{set\_ads\_image} will
%% call the \src{populate\_partner\_ads} function  that is
%% implemented by another library---a library 
%% implemented by an advertising company. 
\begin{lstlisting}[mathescape]
import component Net
component Main {
  static iobuffer[1024];
  static user_balance_usd;

  main () {
    Net.init_network(iobuffer);
    Net.receive();
  }
}
\end{lstlisting}

%\aek{Should we use \sourcelanguage syntax instead
%of C syntax?}
%\dg{Yes, unless that's very difficult to read.}
%\ch{Syntax aside, it's crucial to explain this is a {\bf safe} source language.
%  Otherwise it's unclear why a compiled source context (\IE your Scenario 1)
%  can't break the properties using undefined behavior.}
%\ch{This is better now, but we still don't call it \sourcelanguage.
%  Syntax aside, can \sourcelanguage express this?}
%\aek{Yes, \sourcelanguage can express this. The 
%two variables will need to be allocated dynamically 
%though.}
%\aek{But in any case, now that this example comes in the
%background section, I don't think we should use 
%\sourcelanguage syntax because we did not introduce it
%yet.}

Suppose that the source language is memory safe and that the program
part above is compiled using a \emph{correct} compiler to some
lower-level language, then linked to a context that
implements \remove{\src{populate\_partner\_ads}}\add{\src{Net.init\_network}
and \src{Net.receive}}, and the resulting program is
executed. Our goal is to ensure a safety
property \textbf{nowrite}---that \remove{\src{populate\_partner\_ads}}\add{\src{Net.receive}} never
modifies the variable \src{user\_balance\_usd} (which is high
integrity). Note that it is okay for \remove{\src{populate\_partner\_ads}}\add{\src{Net.receive}} to
modify the array \remove{\src{ads\_image}}\add{\src{iobuffer}}, 
whose pointer is passed as a parameter to \add{the previous call to the \src{Net} library (to the function \src{Net.init\_network})}.
The
concern really is that a low-enough implementation
of \remove{\src{populate\_partner\_ads}}\add{\src{Net.receive}} 
may overflow the
array \remove{\src{ads\_image}}\add{\src{iobuffer}} to overwrite \src{user\_balance\_usd}.

% Since the source language is memory safe,
% \ch{How does source memory safety help with Setting 2?}
Broadly speaking, we can attain the invariant \textbf{nowrite} in at least two different
ways, which we call \textbf{Setting~1} and \textbf{Setting~2}.
In \textbf{Setting~1}, \add{we compile to any target language, possibly memory-unsafe, but} restrict the compilation of the program
part above to be linked \emph{only} to target-language contexts that
were obtained by compiling program parts written in the same source
language. Since the source language is safe, there is no way for any
source function to cause a buffer overflow and a correct compiler will
transfer this restriction to the target language so, in particular,
the compilation of \remove{\src{populate\_partner\_ads}}\add{\src{Net.receive}} 
cannot
overwrite \src{user\_balance\_usd}, thus
ensuring \textbf{nowrite}. This kind of restriction on linking---and the
verification of compilers under such restrictions---has been studied
extensively in \emph{compositional compiler correctness}~\cite{KangKHDV15,10.1145/2784731.2784764,10.1145/2676726.2676985,10.1145/3371091}.

%Although useful, this offers no guarantees when the context
%containing \remove{\src{populate\_partner\_ads}}\add{\src{Net.receive}}
%is arbitrary target code.

\add{In \textbf{Setting 2}, we compile to a target language with support for fine-grained memory
protection, e.g., a capability
machine~\cite{cheriageofrisk,chericompartment} or a tagged-memory
architecture~\cite{pump,micropolicies}, but allow target contexts to
be arbitrary.} \remove{This restriction on linking is lifted
in \textbf{Setting 2}, where the context is arbitrary target code
(written directly in the target language, or maybe compiled from
another less safe language).} \add{The compiler uses the target
language's memory protection to defend against malicious attacks that
do not necessarily adhere to source language's memory-safety semantics.}
Now, \textbf{nowrite} does not follow from source memory safety and
the correctness of the compiler. Instead, we must show that the
compilation chain satisfies some additional security property.  It is this
second setting that interests us here and, more broadly, a large part
of the literature on secure compilation.

\subsection{Robust Safety Preservation (\rsp)}
\label{sec:rsp}

The next question is what security criterion the compilation chain must satisfy
to ensure that \textbf{nowrite} or, more generally, any
property of interest, holds in \textbf{Setting 2}. The literature on
secure compilation has proposed many such criteria
(see~\citet{abate2019journey,Patrignani:2019:FAS:3303862.3280984}). 
Here we describe and adopt one of the simplest criteria that
ensures \textbf{nowrite}, namely, \emph{robust safety preservation} or
\rsp~\cite{DBLP:conf/esop/AbateBCD0HPTT20}:\footnotemark

%% \begin{definition}[Compilation chain has \rsp~\cite{DBLP:conf/esop/AbateBCD0HPTT20}]\label{rsp-def}
%% 	\begin{align*}
%% 	&\rsp~\defeq~
%% 	\forall \src{P}\ \trg{C_t}\ t.\\
%% 	&\ \ \ \ \trg{(C_t\ \cup\ }\cmp{\src{P}}\trg{)}\ 
%% 	\trg{\rightsquigarrow^{*}}\ t \implies\\
%% 	&\ \ \ \ \exists \src{C_s}\ t'.~
%% 	\src{(C_s \cup P)\rightsquigarrow^{*}}\ t'\
%% 	\wedge\ t'\sim t
%% 	\end{align*}
%% \end{definition}

\begin{definition}[Compilation chain has \rsp~\cite{DBLP:conf/esop/AbateBCD0HPTT20}]\label{rsp-def}
	\begin{align*}
	\rsp~\defeq~&
	\forall \src{P}\ \trg{C_t}\ t.\ \trg{(C_t\ \cup\ }\cmp{\src{P}}\trg{)}\ 
	\trg{\rightsquigarrow^{*}}\ t\\
	&\ \ \ \  \implies \exists \src{C_s}\ t'.~
	\src{(C_s \cup P)\rightsquigarrow^{*}}\ t'\
	\wedge\ t'\sim t
	\end{align*}
\end{definition}

\footnotetext{As a notational convention, we use different fonts and
  colors for \src{source\ language\ elements} and
  \trg{target\ language\ elements}. Common elements are written in
  normal black font. We also use the symbol $\downarrow$ for the
  compiler's translation function.}

This definition states the following: Consider any source program part\footnote{We use the notation uppercase 
	$P$ for a program, partial or whole, but only whole
	programs can execute. Whole-program execution
	is denoted $P \rightsquigarrow^{*}\  t$ or
	$P \step{t} s$ where $s$ is a state reached after
	emitting a trace prefix $t$.}
$\src{P}$ and its compilation $\cmp{\src{P}}$. If $\cmp{\src{P}}$
linked with some (arbitrarily chosen) target context \trg{C_t} emits a
finite trace prefix $t$, then there must exist a source context
\src{C_s} that when linked to \src{P} is able to cause \src{P} to emit a
related trace prefix~$t'$.%
%
%Footnote was here

To understand why this definition captures secure compilation in
\textbf{Setting 2}, consider the case where $t$ is a trace witnessing
the violation of a safety property of interest. Then, if the compiler
has \rsp, there must be a source context which causes a similar
violation entirely \emph{in the source language}. In other words, an
attack from some \emph{target}-level context can only arise if the source
program is vulnerable to a similar attack from some \emph{source}-level
context. In our particular example, since there clearly is no source
context violating \textbf{nowrite}, 
\add{\rsp guarantees that} no target context can violate
\remove{\textbf{nowrite}}\add{it} either.

% Practically, 
A compilation chain attains \rsp
% (or a similar criterion for secure
% compilation) \aek{Is this parenthesis necessary?}
by enforcing source language abstractions \remove{like
memory safety} against arbitrary target contexts. \add{The specific
source abstraction of interest to us here is memory safety.}
%
% For this, the compiler and its associated runtime may rely
% on hardware support for memory isolation~\cite{pma},  or bounds
% checking (e.g., hardware capabilities~\cite{cheriageofrisk,chericompartment}),
% or programmable tagged architectures~\cite{pump, micropolicies},
% or it may combine control-flow integrity~\cite{cfijournal},
% software fault isolation~\cite{sfi_sosp1993},
% and software bounds checking~\cite{10.1145/1542476.1542504}.
% \ch{\bf I for one am still not convinced this is doable in pure software.
%   Can we think this through a bit more?} CH: Deepak agreed on Skype
%
Our goal in this paper is to explain that {\em proving} \rsp in the
presence of memory sharing and source memory safety is difficult and
to develop proof techniques for doing this. \add{For simplicity,} the
concrete target language we use is memory safe, \add{but our
techniques benefit any compiler that targets a language with
fine-grained memory protection.}%
% \footnote{Note that our source and target languages differ
%   significantly in their control flow constructs, so the compiler itself is
%   nontrivial; only its security enforcement aspect is rather straightforward.}
\remove{However, even in this setting, the difficulties in proving \rsp with
memory sharing show up prominently.}
Our source and target languages differ significantly in their control
flow constructs, which makes back-translation challenging.%
\iffull
\footnote{\add{
Although we focus on \emph{proof} challenges,
an orthogonal challenge is the efficient enforcement of the safe pointer abstraction
all the way down to machine code. This can be done, for instance, using hardware
capabilities~\cite{cheriageofrisk,chericompartment} or programmable tagged
architectures~\cite{pump,micropolicies}.
}}
\fi

The definition of \rsp is indexed by a relation $\sim$
between source and target traces. The concrete instantiation of this
relation determines how safety properties transform from source to
target~\cite{DBLP:conf/esop/AbateBCD0HPTT20}.  In our setting, $\sim$
is a bijective renaming relation on memory addresses, which we
describe later (\Cref{def-trace-rel}).

\subsection{A Proof Strategy for Robust Safety Preservation}
\label{sec:background-proof-strategy}

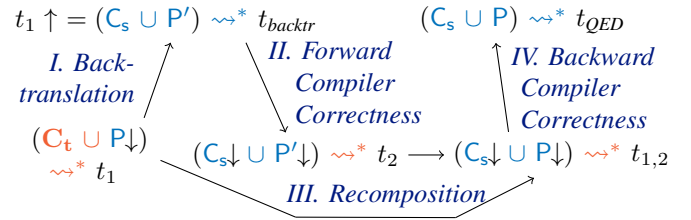
\begin{figure}[h]
	\hspace*{-5mm}\makebox[\textwidth][l]{
		\begin{tikzpicture}[auto]
  \node(cPi) [text width=3cm] {$(\trg{C_t}\ \trg{\cup}\ \cmp{\src{P}})$ \\ \ \ $\trg{\doesprefix}\ t_1$ };
  \node[right = of cPi,xshift=-6em] (CP1t) { $(\cmp{\src{C_s}}\src{\cup}\ \cmp{\src{P'}})\ \trg{\doesprefix}\ t_2$ };
%  \node[align = left, above = of cPi,yshift=-1.5em] (cPiT) { $(\trg{C_t}\ \trg{\cup}\ \cmp{\src{P}})\  \trg{\doesprefix}\ T_1$ };
  \node[align = left, above = of cPi,xshift=1em] (CP1s) { $\back{t_1\ }$ $= (\src{C_s}\ \src{\cup}\ \src{P'})\  \src{\doesprefix}\ t_\mi{backtr}$ };
  \node[right = of CP1t,xshift=-1.5em] (CPt) { $(\cmp{\src{C_s}}\src{\cup}\ \cmp{\src{P}})\  \trg{\doesprefix}\ t_{1,2}$ };
%  \node[above = of cPi, xshift=12em, yshift=1.5em] (prec) { $m \leq t \lor t \prec_P m$ };
  \node[above = of CPt, right = of CP1s] (CPs) { $(\src{C_s}\ \src{\cup}\ \src{P})\ \src{\doesprefix}\ t_\mi{QED} %\wedge (m \leq t \lor t \prec m)
  	$
  };

  \draw[->] (cPi.90) to node
  [text width=2cm, xshift=-0.5em,yshift=-1.2em,align=right] {\hyperref[lemma-abateetal-backtrans]{\em I. Back-\\translation}} (CP1s.-90);

%  \draw[->] (cPi.90) to node
%[text width=2cm,align=right,font=\itshape] {\hyperref[lemma-enrichment]{\em I. Enrichm.}} (cPiT.-90);

  \draw[->] ([xshift=2em]CP1s.-45) to node [align=left,font=\itshape,xshift=-0.25em,yshift=-1.7em,
  text width=2.5cm]{\hyperref[assm-fwdsim]{II. Forward \\\ \ \ Compiler \\ \ \ \ Correctness}} (CP1t);

  %\draw[->] (CP1s.-8) to node [right,yshift=-0.5em, xshift=3em]{\em \hyperref[asm:blame]{5 Blame}} (prec.180);
  %\draw[->] (CPs.140) to node {} (prec.0);

  % composition target
  \draw[->] (CP1t.0) to node [below,xshift=-1.5em,yshift=-0.8em]{\em \hyperref[lemma-robustptrs-recomp]{III. Recomposition}} (CPt.180);
  \coordinate [below=1.5em of CP1t,xshift=0em] (compoint1);
  \coordinate [below=1.5em of CPt,xshift=-4em] (compoint2);
  \draw[-] ([xshift=-4em,yshift=1.2em]cPi.south east) to (compoint1);
  \draw[-] (compoint1) to (compoint2);
  \draw[->] (compoint2) to (CPt.220);

  \draw[->] ([xshift=-2em]CPt.90) to node [right,align=left,font=\itshape,text width = 2.7cm]{\hyperref[assm-bwdsim]{IV. Backward \\ \  Compiler \\ \ Correctness}} ([xshift=-1em]CPs.-90);

  % side labels
  %\node[left = of CP1s, xshift=0em] () {\color{gray} Source};
  %\node[left = of cPiT, xshift=4em,yshift=-2em] () {\color{gray} Target};
\end{tikzpicture}
	}
%	\vspace{-0.8em}
	\caption{Generic proof
          technique~\cite{Abate:2018:GCG:3243734.3243745} for
          \rsp. The traces $t_1$, $t_\mi{backr}$, $t_2$, $t_{1,2}$,
          and $t_\mi{QED}$ are pairwise related by 
          $\sim$.}
	\label{fig:rsc-proof-old}
\end{figure}
%% \dg{This
%% 	figure needs an explanation. In that explanation, we explain how CC
%% 	is reused.}

\rsp can be proved in various ways~\cite{Abate:2018:GCG:3243734.3243745,10.1145/3436809,abate2019journey}. 
Here, we adapt a
proof strategy by~\citet{Abate:2018:GCG:3243734.3243745},
since it \emph{reuses} the proof of compiler correctness,
 % in the proof of \rsp,
thus avoiding duplication of work. \Cref{fig:rsc-proof-old}
summarizes the proof
strategy.\footnote{\citet{Abate:2018:GCG:3243734.3243745} instantiate the
  strategy mostly % the last step is an exception because of undefined behavior!
  for $\sim$ set to equality, while we use a nontrivial $\sim$ everywhere, but
  this difference is less important here. We also removed everything they do
  about undefined behavior, which we do not consider in this work
  (see also \autoref{sec:conclusion}).}

Overall, \citet{Abate:2018:GCG:3243734.3243745}'s proof of \rsp
consists of four steps, two of which are immediate from compiler
correctness. \rsp requires starting from
$(\trg{C_t}\ \trg{\cup}\ \cmp{\src{P}})\ \trg{\rightsquigarrow^{*}}\ t_1$
to demonstrate the existence of a $\src{C_s}$ such that $({\src{C_s}}
\mathbin{\src{\cup}} {\src{P}})\ \src{\rightsquigarrow^{*}}\ t_\mi{QED}$. The first
proof step uses \emph{back-translation}
(\Cref{lemma-abateetal-backtrans}) to show from
$({\trg{C_t}} \mathbin{\trg{\cup}} {\cmp{\src{P}}})\ \trg{\rightsquigarrow^{*}}\ t_1$
that there exist $\src{C_s}$ and $\src{P'}$ such that $(\src{C_s}
\mathbin{\src{\cup}} {\src{P'}})\ \src{\rightsquigarrow^{*}}\ t_\mi{backtr}$ with $t_1 \sim t_\mi{backtr}$. Note that
the back-translation produces both a new context and a new program part, and that
\src{P'} may be completely different from \src{P}.
%% \ch{Would also call out
%%   that the back-translation produces both a new context and a new program.
%%   That's crucial for using {\em whole-program} compiler correctness (with
%%   separate compilation).}\jt{Is it really crucial? One could use \src{C_S} linked with
%%   \src{P} (the original program). But then we wouldn't need all of the
%%   compiler correctness and recomposition, we would simply prove directly that \src{C_S}
%%   linked with \src{P} produces the same trace.}\ch{For the $\src{P}$ part you still need
%%   compiler correctness, since you only know the behavior of $\cmp{\src{P}}$, together with $C_T$,
%%   so I don't think this would save you any step.}
%
The second step
directly uses a form of compiler correctness called forward compiler
correctness (\Cref{assm-fwdsim}), to conclude that the compilation of
this new source program, ${\cmp{({\src{C_s}} \mathbin{\src{\cup}}{\src{P'}})}} =
{\cmp{\src{C_s}}} \mathbin{\src{\cup}} {\cmp{\src{P'}}}$, produces $t_2$,
related to $t_1$.
At this point,
we have two target programs -- ${\trg{C_t}} \mathbin{\trg{\cup}} {\cmp{\src{P}}}$
and ${\cmp{\src{C_s}}} \mathbin{\trg{\cup}} {\cmp{\src{P'}}}$ -- that produce related traces $t_1$ and $t_2$.
The third step uses an innovative target-language
lemma, \emph{recomposition} (\Cref{lemma-abateetal-recomb}), to show
that a third program ${\cmp{\src{C_s}}} \mathbin{\trg{\cup}} {\cmp{\src{P}}}$, which
takes $\cmp{\src{P}}$ from the first program and $\cmp{\src{C_s}}$
from the second, also produces a related
trace $t_{1,2}$. The final, fourth step uses
another form of compiler correctness, called backward compiler
correctness (\Cref{assm-bwdsim}), to conclude from this that the
corresponding source, $\src{C_s} \mathbin{\src{\cup}} \src{P}$ produces
a related trace $t_\mi{QED}$. This concludes the proof.

\ifsooner
\ch{Let's try to use the same names for traces below as in the diagram.}
\aek{Tried to do that, but space is tricky.}
\fi

\begin{lemma}
	[Whole-Program Back-translation~\cite{Abate:2018:GCG:3243734.3243745}]
	\label{lemma-abateetal-backtrans}
	\begin{align*}
	&\forall \trg{P}\ t.\ 
	\trg{P \rightsquigarrow^{*}\ } t \implies
	\exists \src{P}\ t'.\ 
	\src{P \rightsquigarrow^{*}\ } t' \ \wedge\ t' \sim t
	\end{align*}
\end{lemma}
%Given the data-flow trace $T_1$, we show there is a
%source program \src{C_s \cup P'} that emits
%an observation trace $t_{\mi{backtr}}$ that is related
%to $T_1$.
\begin{assumption}[Whole-Program Forward Compiler Correctness]\label{assm-fwdsim}
	$\forall \src{P}\ t.\ 
	\src{P \rightsquigarrow^{*}}\ t \implies
	\exists t'.\ 
	\cmp{\src{P}} \trg{\rightsquigarrow^{*}}\ t'\
	\wedge\ t' \sim t$
\end{assumption}

%By relying on our novel turn-taking simulation (\autoref{sec:key-ideas-turn-taking}),
%we show there is a related trace $t_{1,2}$ emitted by the
%recomposed program $(\cmp{\src{C_s}}\ \src{\cup}\ \cmp{\src{P}})$.
\begin{lemma}
	[Recomposition~\cite{Abate:2018:GCG:3243734.3243745}]
	\label{lemma-abateetal-recomb}
	\label{lemma-robustptrs-recomp}
	\begin{align*}
	&\forall \trg{P_1}\ \trg{C_1}\ \trg{P_2}\ \trg{C_2}\ 
	t_1 t_2.\\
	&\ \ \trg{(P_1 \mathbin{\cup} C_1) \rightsquigarrow^{*}\ }t_1 \implies
	\ \trg{(P_2 \mathbin{\cup} C_2) \rightsquigarrow^{*}\ }t_2 \implies\\
	&\ \ t_1 \sim t_2 \implies\ \exists t_{1,2}.\ \trg{(P_1 \mathbin{\cup} C_2) \rightsquigarrow^{*}\ }t_{1,2}\ \wedge\
	t_{1,2} \sim t_1
	\end{align*}
\end{lemma}

\begin{assumption}[Whole-Program Backward Compiler Correctness]\label{assm-bwdsim}
        $\forall \src{P}\ t.\ 
	\cmp{\src{P}}\ \trg{\rightsquigarrow^{*}}\ t
	\implies
	\exists t'.\ \src{P \rightsquigarrow^{*}}\ t'\ \wedge\
	t' \sim t$
\end{assumption}

By following this proof strategy,
\citet{Abate:2018:GCG:3243734.3243745} are able to reuse compiler
correctness \add{(\Cref{assm-fwdsim,assm-bwdsim})} and reduce the entire proof of \rsp to two key lemmas:
back-translation (\Cref{lemma-abateetal-backtrans}) and recomposition
(\Cref{lemma-abateetal-recomb}).

However, Abate et al.\ execute this strategy for languages without any
memory sharing between components.  Their components---both source and
target---interact \emph{only through} 
\add{\emph{integers} passed as}
function call arguments and return
values. As such, our earlier example cannot even be expressed in their
setting. In the rest of this paper, we adapt their proof strategy
for \rsp to the setting where memory sharing is allowed. We show that
memory sharing significantly complicates the proofs of both
back-translation and recomposition, and requires new proof
techniques. However, before explaining these, we briefly show what
traces actually look like.

\paragraph{Interaction traces}
A trace or, more precisely, an interaction trace, is a modeling and
proof artifact that arises from an instrumented reduction semantics of a
language, wherein certain steps are labeled with descriptors called
\emph{events}. The sequence of events along a reduction sequence forms
a trace, denoted~$t$. In prior work on secure compilation, only steps
involving cross-component interactions or external communication
(input-output) have been labeled with events.
% In contrast, internal
% steps within a component have not been labeled with events.
% CH: this adds nothing
For example, in \citepos{Abate:2018:GCG:3243734.3243745} setting without
shared memory, cross-component interaction happens through
calls and returns only,
%  (information crosses components
% via function arguments and return values only).
hence, their events
are only cross-component calls and returns.
% \aek{Shorten?} CH: done
We denote these events
$e_\mi{no\_shr}$ where the subscript $\mi{no\_shr}$ stands for ``no
memory sharing''.
\begin{align*}
e_{\mi{no\_shr}} ::=
\mathtt{Call}~\mathit{c_{caller}~c_{callee}.f(v)}
\mid
\mathtt{Ret}~\mathit{c_{prev}~c_{next}~v}
\end{align*}
The event $\mathtt{Call}~\mathit{c_\mi{caller}~c_\mi{callee}.f(v)}$
represents a call from component $c_\mi{caller}$ to the function $f$ of
component $c_\mi{callee}$ with argument $\mi{v}$. The dual event
$\mathtt{Ret}~\mathit{c_\mi{prev}~c_\mi{next}~v}$ represents a return from
component $c_\mi{prev}$ to component $c_\mi{next}$ with return value
$\mi{v}$. Along a trace, calls and returns are always well-bracketed (the
semantics of both the source and target languages enforce this).

In our setting, memory shared between components is another medium of
interaction, so reads and writes to it must be represented on
interaction traces. However, our languages are sequential (only one
component executes at a time), so writes to shared memory made by a
component become visible to another component only when the writing
component transfers control to the other component. As such, to
capture interactions between components, it suffices to record the
state of the shared memory only when control transfers from one component
to another, i.e., at cross-component calls and returns. For this, we
modify call and return events to also record the state of the memory
shared up to the time of the event (the shared part of memory grows
along an execution as more pointers are passed across components). The
new events, denoted $e$, are defined below. The shared memory on each
event, written $\mathit{Mem}$, is underlined for emphasis
only. Technically, $\mathit{Mem}$ is a just a partial map from
locations $l$ to values $\mi{v}$, which themselves can be pointers to locations.
\begin{definition}[Interaction-trace events w/ memory sharing]
  \label{def-mem-sharing-trace-event}
  \begin{align*}
    e ::=
    \mathtt{Call}~\mathit{\underline{Mem}~c_{caller}~c_{callee}.f(v)}
    \mid
    \mathtt{Ret}~\mathit{\underline{Mem}~c_{prev}~c_{next}~v}
  \end{align*}
\end{definition}

Interaction traces serve two broad purposes. First, they are used to
express safety properties of interest, such as the \textbf{nowrite}
property in our earlier example. \ifappendix{(\Cref{app:safety-example} shows how
\textbf{nowrite} can be expressed as a predicate on interaction
traces.)}\fi  Second, as we explain in \autoref{sec:keyideas}, interaction
traces are essential to the proof of back-translation,
\Cref{lemma-abateetal-backtrans}. One of our key insights is that, with
memory sharing, enriching interaction traces with selective
information about data-flows \emph{within} a component can simplify
the proof of back-translation considerably.

\section{Key Technical Ideas}\label{sec:keyideas}
% !TEX root = ./paper.tex

%% \subsection{The interaction model in the presence of memory sharing}
%% \label{sec:background-memory-sharing-traces}
%% \input{memory-sharing-interaction.tex}

%Next,
We describe why the proofs of \Cref{lemma-abateetal-backtrans,lemma-abateetal-recomb}
 become substantially more difficult in
the presence of memory sharing, and our new techniques---data-flow back-translations
and turn-taking simulations---that offset some of the extra difficulty.

\subsection{Data-Flow Back-translation}
\label{sec:key-ideas-data-flow-traces}
% !TEX root = ./paper.tex

%\ch{The start of this subsection re-introduces stuff that was already
%introduced in the intro: back-translation, syntax-directed vs
%trace-directed. Is this repetition really needed here?}  \dg{I think
%so. The reader has likely forgotten specific details in the intro by
%now, and we never connected that to Lemma 2.2 anyhow.}

%	2) In RobustPtrs, we use (for the first time in a compiler security proof)
%	the idea of "data-flow trace labels". These labels are designed to capture 
%	all the data flow events happening in the register file and the memory,
%	but they ignore all internal control-flow events (i.e., control-flow events
%	that do not switch control to a different program component). 
%	The data-flow trace labels are handy because they guide the construction
%	of the mimicking back-translated source program, and they also make 
%	explicit the states of the target execution at which the simulation proof
%	(of back-translation correctness) needs to re-establish the cross-language
%	invariant (between the given target states and the source ones of the 
%	back-translated program). The definition of these data-flow trace labels
%	is \emph{not} part of the TCB of our theorem statement, thanks to
%	transfer lemmas that allow us to state our top-level theorem (compiler
%	preservation of robust safety) in terms of other more conventional
%	trace labels that are less informative and more intuitive than the
%	data-flow traces.

\ifsooner
\ch{I generally find the first page of this section too repetitive, many things
  which were already said in the abstract and intro and Section 2 are repeated
  multiple times here too.}
\fi

In proving back-translation (\Cref{lemma-abateetal-backtrans}), we are
given a target language whole program $\trg{P}$ and an interaction
trace $t$ that it produces, and we have to construct a whole source
program $\src{P}$ that produces a related interaction trace
$t'$. \remove{This process of constructing the source program
$\src{P}$ is often called \emph{back-translation} (hence, the name of
the lemma). Obviously,}\add{For \rsp, }we can construct $\src{P}$ from either
$\trg{P}$ or $t$. Prior work has considered both approaches.
Construction of $\src{P}$ from $\trg{P}$, which we
call \emph{syntax-directed back-translation}, typically works by
simulating $\trg{P}$ in the source
language~\cite{DBLP:journals/lmcs/DevriesePPK17,popl-backtrans,Skorstengaard:2019:SEW:3302515.3290332,DBLP:journals/pacmpl/StrydonckPD19,TsampasStelios2019Tsfs,ahmedCPS,AhmedFa,New,Patrignani:2019:FAS:3303862.3280984}.
This is tractable when every construct of the target language can be
simulated easily in the source. However, \add{as explained earlier,}
this is not the case for many pairs of languages including our source
and target languages (\autoref{sec:compiler}).
The alternative then is to construct $\src{P}$ from the given target
trace
$t$ \cite{capableptrs,Abate:2018:GCG:3243734.3243745,patrignani2015secure,10.1145/3436809}. This
alternative, which we call \emph{trace-directed back-translation},
should be easier in principle, since the interaction trace only records
cross-component interactions, so there is no need to simulate every
language construct in the source; instead, only constructs that can
influence cross-component interactions need to be simulated.

Indeed, trace-directed back-translation is fairly straightforward
when there is no memory
sharing~\cite{Abate:2018:GCG:3243734.3243745,patrignani2015secure} or
when memory references (pointers) can be constructed from primitive
data like integers in the source language\remove{ (the latter is true in
unsafe languages like C)}.
%\ch{Not convinced by this paren. I'm sure Peter Sewell has a lot
%  to say about this, but AFAIR it is undefined behavior in C to turn an arbitrary integer
%  (which wasn't obtained from a pointer) into a pointer. And the way UB interacts with
%  secure compilation seems too complex to explain in this paper. Could it be that
%  pointers and integers were convertable in K\&R C (\IE pre ISO standard)?}
However, with memory sharing and unforgeable
memory references in the source---something that is common in safe
source languages like Java, Rust, Go and ML---trace-directed
back-translation is really difficult. To understand this, consider \remove{the
following simple example.}\add{the following
run of the \emph{compiled version}
of our example from \autoref{sec:background-example}.}

\begin{example}\label{example-sharing-history-dependence}
Suppose we want to back-translate the following four-event target interaction trace:
\begin{align*}
	&\mathtt{Call}~\mathit{
        Mem~c_\mathtt{Main}~c_\mathtt{Net}.
	\mathtt{init\_network}(l_\mathtt{iobuffer})}\\
::\ &\mathtt{Ret}~\mathit{
	Mem~c_\mathtt{Net}~c_\mathtt{Main}~0}\\
::\ &\mathtt{Call}~\mathit{
		Mem~c_\mathtt{Main}~c_\mathtt{Net}.
		\mathtt{receive}()}\\ 
::\ &\mathtt{Ret}~\mathit{
		Mem'~c_\mathtt{Net}~c_\mathtt{Main}~0}
\end{align*}
where \remove{$l_1$ and $l_2$ are distinct memory locations, $\mi{Mem} =
[l_1 \mapsto l_2, l_2 \mapsto 0]$ and $\mi{Mem'} = [l_1 \mapsto 100,
l_2 \mapsto 0]$.}\add{$\mi{Mem} = [l_\mathtt{iobuffer}\ \mapsto 0,\ l_\mathtt{iobuffer} + 1\ \mapsto 0,\ \ldots,\  
l_\mathtt{iobuffer} + 1023\ \mapsto 0 ]$ and
$\mi{Mem'} = [l_\mathtt{iobuffer}\ \mapsto 4,\ l_\mathtt{iobuffer} + 1\ \mapsto 4,\ \ldots,\  
l_\mathtt{iobuffer} + 1023\ \mapsto 4 ]$}\footnote{Technically, in our languages, function
calls and returns and, hence, interaction traces carry \emph{pointers}
to locations, not locations themselves. However, in this section,
we blur this distinction.}
\end{example}

%% \ch{I don't like using the word location to mean pointer,
%%   when we talk about values that are unforgeable and can be passed around.}
%% \dg{I agree. However, I don't know of a better way to explain this without introducing the language first. If you have an idea, please go ahead and change consistently, both here and in \autoref{sec:compiler}.}

\remove{In this example, the module $c_1$ calls $c_2$ twice -- first it calls
$c_2.f()$ and then it calls $c_2.g()$. Assume that prior to these
calls, $c_1$ dynamically allocated locations $l_1$ and $l_2$. In the
first call, $c_1$ passes the location $l_1$ to $c_2$ as the function
call argument. At this point $l_1$ happens to contain $l_2$. As a
result, the first call shares \emph{both} $l_1$ and $l_2$ with $c_2$
-- it shares $l_1$ directly, and shares $l_2$ indirectly via $l_1$.
$c_2.f()$ returns $0$ to $c_1$ without changing the shared memory.
Later, $c_1$ overwrites $l_1$'s contents with the value $100$ (to get
a new shared memory $\mi{Mem'}$), and calls $c_2.g()$. This time $c_2$
returns $l_2$ to $c_1$.

Note that at the time of second call, $l_2$ is actually not reachable
from the shared memory $\mi{Mem'}$. However, $c_2$ could have stashed
$l_2$ somewhere in its private memory during the first call and
retrieved it from there during the second call to return it.}

\add{In this example run, the program first shares some memory
(corresponding to \src{iobuffer}) by
calling \trg{Net.init\_network} with the pointer
$l_\mathtt{iobuffer}$. This call does not modify the shared memory
(the shared memory's state is $\mi{Mem}$ both before and after the
call). Later the program calls the function $\trg{Net.receive}$
without any arguments, but this call changes the shared memory
to $\mi{Mem}'$. (Assuming that our compiler uses the target's memory
protection correctly, this could only have happened if the \trg{Net}
library stashed the pointer $l_\mathtt{iobuffer}$ during the first
call and retrieved it during the second call.})
%% \aek{Here, we need to add sth like, ``One
%% can eliminate this latter possibility by 
%% enforcing fine-grained memory access control,
%% \EG by relying on hardware
%% capabilities~\cite{cheriageofrisk,chericompartment} or programmable tagged
%% architectures~\cite{pump,micropolicies}.''
%% Or instead, we can replace the last sentence
%% of the previous paragraph by sth like: 
%% ``By relying on an enforcement mechanism for
%% fine-grained memory access control,
%% one can in principle provably eliminate the 
%% possibility that the \trg{Net} library forged
%% the pointer $l_\mathtt{iobuffer}$ during the
%% second call, thus the \trg{Net} library must 
%% have stashed it during the first call to use
%% it during the second.''}

%% (\mi{Mem}) 
%% of the shared memory (which is the memory corresponding to the source
%% array \src{iobuffer}) is shown on the first $\mathtt{Call}$ event in
%% which
%% \src{iobuffer} became initially shared with
%% the module $\trg{Net}$. The contents
%% of the shared memory remain unchanged at the
%% point $\trg{Net.init\_network}$ returns. Then
%% the compiled
%% component $\trg{Main}$ calls the target context's
%% $\trg{Net.receive()}$ without passing any
%% arguments, but upon returning, we observe
%% that the contents of the \src{iobuffer} have now
%% been overwritten by $\trg{Net.receive}$, which
%% overwrites the zeros with fours.

The question is how we can back-translate this interaction sequence
into a source program\add{, as required
by \Cref{lemma-abateetal-backtrans}}. If pointers were forgeable in
the source, this would be quite easy: $l_\mathtt{iobuffer}$,
being forgeable, could simply be hardcoded in the body of the
simulating source function $\src{Net.receive()}$.
However, \add{in our memory-safe source language},\remove{this is not
straightforward. Now, we cannot hardcode $l_2$ into the
back-translated $c_2.g()$'s body since $l_2$ is dynamically allocated
by $c_1$ during execution!  Consequently, the back-translated
component $c_2$ \emph{must} store $l_2$ in an indexed data structure
during the first call (to $c_2.f()$), and then somehow retrieve it
from that data structure in the second call.}
\add{the only option is to construct a source \src{Net.init\_network} that stashes
$l_\mathtt{iobuffer}$ for
\src{Net.receive}'s use.}
%
%% \ch{Let's mention more explicitly that the problem is more difficult than the
%%   simple example above might make it seem. The generated context would need to
%%   fetch and store all reachable pointers (which the example doesn't really show).}
%
\add{Even though this ``stashing'' solution may
seem straightforward,}
\remove{the problem is actually more difficult than this example shows: The}\add{it is actually quite difficult because the}
back-translated context must fetch and 
\remove{store }\add{stash (\EG }in 
\remove{its }\add{an }indexed data
structure\add{)} \emph{all} pointers that become accessible to it directly
or indirectly by following shared pointers, since any of these pointers {\em may}
\remove{show up on the trace later}\add{be dereferenced later}.
\aek{They will definitely show up later 
since they all are ``shared so far'' pointers, 
but the over-approximation, \IE the ``may'' 
modality applies to whether we will need to
access their pointees or not.}

Prior work~\cite{capableptrs,10.1145/3436809} 
has used such a stashing data
structure. \remove{, because additional, complex
invariants about this data structure must be proved.}
\add{They fetch pointers
	by a custom graph traversal 
	(pointers are the edges and pointed locations are the nodes)
	whose output is a list of source commands. Each source command
	is responsible for traversing a path in memory and stashing the
	content of the destination location in private memory (in 
	anticipation that this stashed content might be a pointer
	in which case it might be
	needed when back-translating a later interaction event).}
\add{We found that proving
the correctness of} this\remove{is an immensely difficult}
construction \add{in a proof assistant is difficult, even
though the proofs seem easy on paper.  For instance, even leaving
aside the correctness of this traversal (which seems rather difficult),
just proving its termination is nontrivial in a proof assistant.}

%% \add{in a proof assistant, 
%% 	many termination/finiteness obligations
%% 	arise, let alone the main simulation
%% 	proof for the source 
%% 	expressions that the graph traversal outputs, 
%% 	which is already a non-trivial
%% 	simulation proof.}

\add{Note that traversing the entire shared memory is a proactive,
	over-approximating strategy on the part of the generated source
	context, by which it mimics \emph{all possible} stashing steps
	that the target context \emph{could} have made. This complex
	strategy was needed in prior work because information about
	data flows within the target context is missing from standard
	interaction traces, which prior work relied on.  If only
	a trace recorded precisely which memory paths were actually
	traversed, we could eliminate the complexity of the full
	traversal. This is exactly what our new data-flow
	back-translation idea supports.}
%% More generally, the problem is that traces (deliberately) abstract
%% away information about how a component stores and retrieves locations
%% (pointers) internally. So, every back-translated source component must
%% maintain a \emph{bookkeeping} data structure where every location that
%% becomes accessible to it -- directly or indirectly -- must be stored
%% and somehow accessed whenever the location appears on the trace
%% later. 

\paragraph{The new idea: \add{data-flow back-translation}}
\remove{This is where our new idea of \emph{data-flow back-translation} comes
in.} We enrich the interaction traces of the target language---only for
the purposes of the back-translation proof---with information
about \emph{all} data-flows, even those
\emph{within} (the private state of) a single component. We
call these enriched traces \emph{data-flow traces}. From the target
language's reduction semantics, we can easily prove that every
interaction trace as described above can be enriched to a data-flow
trace (\Cref{lemma-enrichment} below). And, given such a data-flow
trace, we can easily back-translate to a simulating source program,
since we know exactly how pointers flow. In the example
above, the enriched trace would tell us exactly what
$\trg{Net.init\_network}$ did to \remove{store}\add{stash}
$l_\mathtt{iobuffer}$ and how $\trg{Net.receive}$ retrieved it
later. We can then mimic this in the constructed source
program, \add{without having to stash all reachable pointers in memory
whenever passing control to the context} (see \Cref{example-bookkeeping-explicit}
below).

Concretely, we define a new type of data-flow traces, denoted $T$,
whose events, $\mathcal{E}$, extend those of  
interaction traces to
capture all possible data flows in the target language. In the
following, we show the events for our target language
(\autoref{sec:compiler}), which is a memory-safe assembly-like language
with registers and memory. The events $\mathtt{dfCall}$ and
$\mathtt{dfRet}$ are just the $\mathtt{Call}$ and $\mathtt{Ret}$
events of interaction traces (\Cref{def-mem-sharing-trace-event}). The
remaining events correspond to target language instructions that cause
data flows: loading a constant to a register ($\mathtt{Const}$),
copying from a register to another ($\mathtt{Mov}$), binary operations
($\mathtt{BinOp}$), copying from a register to memory or vice-versa
($\mathtt{Store}$, $\mathtt{Load}$) and allocating a fresh location
($\mathtt{Alloc}$). Importantly, in a data-flow trace, every event
records the entire state---both shared state and state private to
individual components. Accordingly, in the events below, $\mi{Mem}$
also includes locations that were not shared to other components, and
$\mi{Reg}$ is the state of the register file.

\begin{definition}[Events of data-flow traces]
	\label{def-dataflow-events-concrete}
	\begin{align*}
	\mathcal{E} ::=\ 
	&\mathtt{dfCall}~\mathit{Mem~Reg~c_{caller}~c_{callee}.proc(v)}\\
	\mid\ 
	&\mathtt{dfRet}~\mathit{Mem~Reg~c_{prev}~c_{next}~v}\\
	\mid\ 
	&\mathtt{Const}~\mathit{Mem~Reg~c_{cur}~v~r_{dest}}\\
	\mid\ 
	&\mathtt{Mov}~\mathit{Mem~Reg~c_{cur}~r_{src}~r_{dest}}\\
	\mid\ 
	&\mathtt{BinOp}~\mathit{Mem~Reg~c_{cur}~op~r_{src1}~r_{src2}~r_{dest}}\\
	\mid\ 
	&\mathtt{Load}~\mathit{Mem~Reg~c_{cur}~r_{addr}~r_{dest}}\\
	\mid\ 
	&\mathtt{Store}~\mathit{Mem~Reg~c_{cur}~r_{addr}~r_{src}}\\
	\mid\ 
	&\mathtt{Alloc}~\mathit{Mem~Reg~c_{cur}~r_{ptr}~r_{size}}
	\end{align*}
\end{definition}

\begin{example}\label{example-bookkeeping-explicit}
Consider the following data-flow trace, which expands a part
of \Cref{example-sharing-history-dependence}'s interaction trace---the part that
covers the call and return to \trg{Net.init\_network()} only.  Here, $l$ is a fixed,
hardcodable location that can always be accessed by \trg{Net},
$r_\mi{COM}$ is a special register used to pass arguments and return
values, and $\mi{Mem_1}$ and $\mi{Reg_1}$ are some initial states of
memory and registers, respectively.
\begin{align*}
       &\mathtt{dfCall}~\mathit{Mem_1~(Reg_1[r_\mi{COM} \mapsto l_\mathtt{iobuffer}])}\\
       	&\ \ \ \ \ \ \ \ \  ~\mathit{c_\mathtt{Main}~c_\mathtt{Net}.\mathtt{init\_network}(l_\mathtt{iobuffer})}\\
::\    &\mathtt{Const}~\mathit{Mem_1~(Reg_1[r_\mi{COM} \mapsto l_\mathtt{iobuffer}, r_1 \mapsto l])~c_\mathtt{Net}~l~r_1}\\
::\
&\mathtt{Store}~\mathit{(Mem_1[l\mapsto l_\mathtt{iobuffer}])}~\\
&\ \ \ \ \ \ \ \ \ \mathit{(Reg_1[r_\mi{COM} \mapsto l_\mathtt{iobuffer}, r_1 \mapsto l])~c_\mathtt{Net}~r_1~r_\mi{COM}}\\
::\    
&\mathtt{dfRet}~\mathit{(Mem_1[l\mapsto l_\mathtt{iobuffer}])}~\\
&\ \ \ \ \ \ \ \ \ \mathit{(Reg_1[r_\mi{COM} \mapsto l_\mathtt{iobuffer}, r_1 \mapsto 0])~c_\mathtt{Net}~c_\mathtt{Main}~l_\mathtt{iobuffer}}
\end{align*}
\end{example}
This data-flow trace shows clearly how \trg{Net.init\_network} stashed away $l_\mathtt{iobuffer}$:
\remove{It copied $l_2$ to the register $r_2$ and then to the
location $l$. }
\add{It copied $l_\mathtt{iobuffer}$ to its
private memory location $l$.}
The rest of the data-flow trace (not shown) will also show precisely how \remove{$c_2.g()$}\add{\trg{Net.receive()}} later retrieved \remove{$l_2$}\add{$l_\mathtt{iobuffer}$}.
It is not difficult to
construct a source program that mimics these data flows step-by-step,
by using source memory locations to mimic the target's register file
and memory (see \autoref{sec:back-trans} for further details).
\add{The step-by-step mimicking induces a 
step-for-step inductive invariant that we found
much simpler to prove than the coarse-grained invariants
from prior work~\cite{capableptrs} in which
the back-translation input 
was just the non-informative trace of
\Cref{example-sharing-history-dependence}, and
the lost
target steps were compensated using the full graph traversal.
}

\paragraph{Outline of data-flow back-translation proof}
%% In the presence of  memory sharing between components and non-forgeable
%% pointers in the source, trace-based back-translation using
%% standard interaction traces requires extensive and difficult bookkeeping to track
%% and access previously shared references. Our data-flow
%% back-translation simplifies this by introducing a new proof artifact,
%% data-flow traces, which make all data-flow in the program explicit and
%% simplify the proof of back-translation significantly.
%% \jt{Paragraph seems to be repeating what was said above, maybe cut it
%% if we need space}

Data-flow traces simplify the proof of back-translation
(\Cref{lemma-abateetal-backtrans}) by splitting it into two key lemmas:
Enriching interaction traces to data-flow traces
(\Cref{lemma-enrichment}) and back-translation of data-flow traces
(\Cref{lemma-robustptrs-backtrans}), both of which are shown below and
are much easier to prove than standard trace-directed
backtranslation. Recall that $T$ denotes a data-flow trace.
$\mathtt{remove\_df}(T)$ denotes the interaction trace obtained by
removing all internal data-flow events from $T$, i.e., by retaining
only $\mathtt{Call}$ and $\mathtt{Return}$ events.

%% \ch{Seems worth mentioning that the program P we're producing by
%%   back-translation is a whole program, not just the ``program part''.
%%   In fact we are mostly interested in the context part!}
%% \ad{maybe it would make more sense to state that when the idea 
%% of recomposition is introduced, in the intro? in order to justify 
%% why we need recomposition. I made a suggestion in the intro}

\begin{lemma}[Enrichment]
\label{lemma-enrichment}
\begin{align*}
&\forall \trg{P}\ t.\ 
\trg{P\ \rightsquigarrow^{*}\ } t \implies 
\exists T.\ \trg{P\ \emitsdfT\ } T\ \wedge\ 
t = \mathtt{remove\_df}(T)
\end{align*}
\end{lemma}
\begin{proof}
Immediate from the definition of the target-language semantics.
\end{proof}

\begin{lemma}[Data-flow back-translation]
\label{lemma-robustptrs-backtrans}
\begin{align*}
&\forall \trg{P}\ T.\ 
\trg{P\ \emitsdfT\ }T \implies
\exists \src{P}\ t.\ 
\src{P \rightsquigarrow^{*}\ }t\ \wedge\ 
t \sim \mathtt{remove\_df}(T)
\end{align*}
%% where $\mathtt{remove\_df}$ clears the ``$\mathtt{FlowStep}$'' kind of events from a data-flow
%% trace (\Cref{def-dataflow-events-abstract}) 
%% to obtain a trace of observable events 
%% (\Cref{def-mem-sharing-trace-event}).
\end{lemma}
\begin{proof}[Proof sketch]
By constructing a \src{P} that simulates the data flows in $T$, thus
keeping its state in lock-step with the state in $T$'s
events. See \autoref{sec:back-trans} for further details.
\end{proof}
Composing these two lemmas  yields
\Cref{lemma-abateetal-backtrans}.

\subsection{Turn-Taking Simulation for Recomposition}
\label{sec:key-ideas-turn-taking}
% !TEX root = ./paper.tex

%% To date, there is only one mechanized proof of secure compilation that reuses
%% compiler correctness \cite{Abate:2018:GCG:3243734.3243745}. This technique
%% relies on a recomposition property (\Cref{lemma-abateetal-recomb}) \rb{Better now?} that we can
%% adapt to the setting of languages with support for memory sharing.
%% %% \citet{Abate:2018:GCG:3243734.3243745} have proposed the only mechanized secure
%% %% compilation proof technique for reusing compiler correctness. As explained
%% %% in \autoref{sec:background-proof-strategy}, their strategy relies on a
%% %% recomposition lemma (\Cref{lemma-abateetal-recomb}).\ch{This lemma
%% %%   reference really looks terrible.}\aek{better now?}\ch{a bit}
%% %% %
%% %% We prove a version of this lemma in the novel setting of a language
%% %% that supports memory sharing.
%% % (with a trace relation different from equality) -- CH: too much of a detail?
%% In this section, we introduce
%% the high-level design of the novel \emph{memory invariant}
%% required by the proof.

Next, we turn to recomposition (\Cref{lemma-abateetal-recomb}). This
lemma states that if two programs \trg{P_1 \cup C_1} and \trg{P_2 \cup
C_2} produce two related interaction traces, then the
program \trg{P_1 \cup C_2} can also produce an interaction trace
related to both those traces. We refer
to \trg{P_1 \cup C_1} and \trg{P_2 \cup C_2} as \emph{base} programs,
and to \trg{P_1 \cup C_2} as the \emph{recomposed} program.  We say
that the partial programs \trg{P_1} and \trg{C_2} are
\emph{retained} by the recomposition, 
and that \trg{P_2} and \trg{C_1} are \emph{discarded}.
Traces in this section refer to the interaction traces
of \Cref{def-mem-sharing-trace-event}. Data-flow traces are used only
for back-translation, not for recomposition.

The proof of recomposition is a ternary simulation over executions of
the three programs. For this, we need a ternary relation between a
pair of states \trg{s_1} and \trg{s_2} of the base programs and a
state \trg{s_{1,2}} of the recomposed program. The question is how we
can relate the \emph{memories} in \trg{s_1} and \trg{s_2} to that
in \trg{s_{1,2}}.

In the absence of memory sharing, as
in~\citet{Abate:2018:GCG:3243734.3243745}, this is
straightfoward: We simply project $\trg{P_1}$'s memory from
$\trg{s_1}$, $\trg{C_2}$'s memory from $\trg{s_2}$, put them together
(take a disjoint union), and this yields the memory of $\trg{s_{1,2}}$:
% (up to location renaming, which we capture with a relation $\sim_\mathtt{ren}$ that
% is formally defined later).
% \rb{Although in this scenario there is no sharing, so no actual need for location renaming (but useful to keep in mind for example below?)}

\begin{definition}[Memory relation of~\citet{Abate:2018:GCG:3243734.3243745}]\footnote{\add{See \Cref{sec:compiler} for the precise definition
of $\mathtt{proj}$.}}
	\label{def-mem-rel-abateetal}
	\begin{align*}
	&\mathtt{mem\_rel}
	(\trg{s_1}, \trg{s_2}, \trg{s_{1,2}}) \defeq\\
	&\trg{s_{1,2}.\mi{Mem}} =
	(\mathtt{proj}_{\trg{P_1}} (\trg{s_1.\mi{Mem}})
	\ \uplus\ 
	\mathtt{proj}_{\trg{C_2}} (\trg{s_2.\mi{Mem}}))
	\end{align*}
\end{definition}
%% \begin{wrapfigure}{r}{0.25\textwidth}
%% 	%Edit figure here: https://docs.google.com/drawings/d/1MJOUTNHT3LD7vvIsOaqWMuyZBQadn4Oos8zBqBB8-ZA/edit?usp=sharing
%% 	\centering
%% 	\includegraphics[trim={7.5cm 8cm 10cm 3cm},clip,width=0.23\textwidth]{mem-rel}
%% \end{wrapfigure}

%% The operator $\mathtt{proj}_{\trg{P_1}}$ applied to
%% a memory \trg{\mi{Mem}} gives the subset of the memory
%% addresses that were allocated for program part
%% \trg{P_1} and their corresponding contents. We keep
%% the relation $\sim$ abstract for now and define it later.
%% $\mathtt{mem\_rel}$ can be used to describe
%% the \emph{code memory} of the recomposed program,
%% as code is immutable.
%% %% can be described
%% %% by the relation $\mathtt{mem\_rel}$ of 
%% %% \Cref{def-mem-rel-abateetal}.
%% However, if we
%% attempt to use it on \emph{data memories},
%% we quickly realize that due to memory sharing, 
%% a program part's memory can become writable by other 
%% program parts,
%% and $\mathtt{mem\_rel}$ is too strong to allow this effect.

However, with memory sharing, this definition no longer works, as illustrated by the following example. 

\begin{example}
\label{example-c-changes-shared-memory-internally}
Consider the following three target-language components $\trg{C_1}$,
$\trg{C_2}$ and $\trg{P_1}$, represented in C-like syntax for
simplicity. The fourth component $\trg{P_2}$ is irrelevant for this
explanation, hence not shown.
%% ~\\	
\begin{lstlisting}[mathescape]
component $\trg{C_1}$ {
  int* ptr_to_P1 = malloc();
  void store(int* arg) {
    ptr_to_P1 = arg;
    int val_to_revert = *ptr_to_P1;
    *ptr_to_P1 = 42;
    ...
    *ptr_to_P1 = val_to_revert;
  }
}
\end{lstlisting}
\begin{lstlisting}[mathescape]
component $\trg{C_2}$ {
  int* ptr_to_P1_or_P2 = malloc();
  void store(int* arg) {
    ptr_to_P1_or_P2 = arg;
  }
}
\end{lstlisting}
\begin{lstlisting}[mathescape]
component $\trg{P_1}$ { 
  int* priv_ptr = malloc();
  int* shared_ptr = malloc();
  void call_store() {
    store(shared_ptr);
  }
}
\end{lstlisting}
\end{example}
In the base program $\trg{P_1 \cup C_1}$, $\trg{P_1}$
shares \lstinline{shared_ptr} with the
function \trg{C_1}.\lstinline{store}(). This
function \emph{temporarily} updates \lstinline{shared_ptr} but reverts
it to its original value before returning. Somewhat differently, in
the recomposed program $\trg{P_1 \cup C_2}$,
$\trg{C_2}$.\lstinline{store}() does \emph{not}
modify \lstinline{shared_ptr} at all. Thus, even though the
end-to-end interaction behavior of \lstinline{store}() in both
the programs is exactly the same, \lstinline{shared_ptr} (which is
actually in $\trg{P_1}$'s memory) has been temporarily modified
in \trg{C_1}.\lstinline{store}() but not
in \trg{C_2}.\lstinline{store}(). Consequently, \emph{during} the
execution of the context's function \lstinline{store}(), the memory
relation of \Cref{def-mem-rel-abateetal} does not hold.

%% Here, \trg{C_1} modifies the memory that \trg{P_1}
%% shares with it, but reverts
%% the changes before returning to \trg{P_1}.
%% Thus, those changes
%% are not observable to \trg{P_1} from whose perspective
%% \trg{C_1} and \trg{C_2} behave the same.
%% However, there is a mismatch in
%% \trg{P_1}'s part of the memory in the internal
%% executions of
%% the base program \trg{P_1 \cup C_1} and the
%% recomposed program \trg{P_1 \cup C_2}.

More abstractly, the problem here is that $\trg{P_1}$'s \emph{shared}
memory in the recomposed program $\trg{P_1 \cup C_2}$ can be related
to that in the base program $\trg{P_1 \cup C_1}$ \emph{only while}
control is in $\trg{P_1}$. When control is in $\trg{C_2}$, the
contents of $\trg{P_1}$'s shared memory can change unrelated to the
base runs. This naturally leads to the following program counter-aware
memory relation, where the relation $\sim_\mathtt{ren}$ captures
location renaming and is formally defined later in this section.

%% To address this, we need to distinguish -- in the memory he states with which a mismatch is allowed (\EG internal states of
%% the base program
%% \trg{P_1 \cup C_1} at which the \emph{discarded} part
%% \trg{C_1} is executing)
%% from other states in which a mismatch is not allowed
%% (\EG internal states of the base program
%% \trg{P_1 \cup C_1} at which the \emph{retained} part
%% \trg{P_1} is executing), thus,
%% we need a program counter (PC)-aware memory invariant.
%% %\ch{This is not specific enough. When is it allowed
%% %  and when is it not allowed? Executing in the program
%% %  vs executing in the context?}
%% %\aek{Mismatch is allowed for states of a \trg{part} when the 
%% %recomposed program is \emph{not} executing in this \trg{part}}
%% \rb{TODO: Have another look at this after clarification on internal execution
%% 	(Also: could we say simply ``PC-aware''?)}

\begin{definition}[First attempt at our memory relation]
	\label{def-mem-rel-renaming-pc}
	\begin{align*}
	&\mathtt{mem\_rel\_pc}
	(\trg{s_1}, \trg{s_2}, \trg{s_{1,2}}) \defeq\\
	&\texttt{if}\ \trg{s_{1,2}}\ \texttt{is executing in}\ \trg{P_1}\ \texttt{then:}\\
	&\ \ \ \ \mathtt{proj}_{\trg{P_1}} (\trg{s_{1,2}.\mi{Mem}})
	\sim_\mathtt{ren}
	%_{\mathtt{ren}(\trg{P_1}, \trg{P_2}, \trg{C_1}, \trg{C_2})}
	\mathtt{proj}_{\trg{P_1}} (\trg{s_1.\mi{Mem}})\\	
	&\texttt{else:}\ (\IE, \trg{s_{1,2}}\ \texttt{is executing in}\ \trg{C_2})\\
	&\ \ \ \ \mathtt{proj}_{\trg{C_2}} (\trg{s_{1,2}.\mi{Mem}})
	\sim_\mathtt{ren}
	%_{\mathtt{ren}(\trg{P_1}, \trg{P_2}, \trg{C_1}, \trg{C_2})}
	\mathtt{proj}_{\trg{C_2}} (\trg{s_2.\mi{Mem}})
	\end{align*}
\end{definition}

Although this definition relates \emph{shared} memory correctly, it is
inadequate for $\trg{P_1}$'s \emph{private} memory---the memory
$\trg{P_1}$ has not shared with the context in the past, such as the
pointer \lstinline{priv_ptr} in
\Cref{example-c-changes-shared-memory-internally}. This private memory must
remain related in the base program $\trg{P_1 \cup C_1}$ and the
recomposed program $\trg{P_1 \cup C_2}$ independent of where the
execution is. However, \Cref{def-mem-rel-renaming-pc} does not say
this.

Accordingly, we revise our definition again. To determine which
locations have been shared and which are still private, we rely on the
interaction trace prefixes \trg{t_1}, \trg{t_2} and \trg{t_{1,2}} that are emitted
before reaching the states \trg{s_1},
\trg{s_2} and \trg{s_{1,2}}, respectively. For a memory $\mi{mem}$
and a trace $t$, we write $\mathtt{shared}(\mi{mem}, t)$ for the
projection of $\mi{mem}$ on addresses that are 
\add{transitively}
shared on the trace $t$
and $\mathtt{private}(\mi{mem}, t)$ for the projection of $\mi{mem}$
on all the other addresses.  With this, we can finally define a
\emph{turn-taking relation} $\mathtt{mem\_rel\_tt}$ that accurately describes
the memory \trg{s_{1,2}.\mi{Mem}} of the recomposed program in terms
of the memories \trg{s_1.\mi{Mem}} and
\trg{s_2.\mi{Mem}} of the two base programs:

\begin{definition}[Turn-Taking Memory Relation]
	\label{def-mem-rel-tt}
	\begin{align*}
	&\mathtt{mem\_rel\_tt}(\trg{s_{1,2}}, \trg{s_1}, \trg{s_2}, t_{1,2}, t_1, t_2) \defeq\\
	&\texttt{if}\ \trg{s_{1,2}}\ \texttt{is executing in}\ \trg{P_1}\ \texttt{then:}\\
	&\ \ \ \ \mathtt{mem\_rel\_exec}(\trg{P_1}, t_1, t_{1,2}, \trg{s_1.\mi{Mem}}, \trg{s_{1,2}.\mi{Mem}})\ 
	\wedge\ \\
	&\ \ \ \ \mathtt{mem\_rel\_not\_exec}(\trg{C_2}, t_2, t_{1,2}, \trg{s_2.\mi{Mem}}, \trg{s_{1,2}.\mi{Mem}})\\
	&\texttt{else:}\ (\IE, \trg{s_{1,2}}\ \texttt{is executing in}\ \trg{C_2})\\
	&\ \ \ \ \mathtt{mem\_rel\_exec}(\trg{C_2}, t_2, t_{1,2}, \trg{s_2.\mi{Mem}}, \trg{s_{1,2}.\mi{Mem}})\ 
	\wedge\ \\
	&\ \ \ \ \mathtt{mem\_rel\_not\_exec}(\trg{P_1}, t_1, t_{1,2}, \trg{s_1.\mi{Mem}}, \trg{s_{1,2}.\mi{Mem}})
	\end{align*}
	where
	\begin{align*}
	&\mathtt{mem\_rel\_exec}(\trg{part}, t, t_{1,2}, \trg{m_{base}}, \trg{m_{recomp}}) \defeq \\
	&\mathtt{proj}_{\trg{part}} (\trg{m_{recomp}})
	\sim_\mathtt{ren}
	%_{\mathtt{ren}(\trg{P_1}, \trg{P_2}, \trg{C_1}, \trg{C_2})}
	\mathtt{proj}_{\trg{part}} (\trg{m_{base}})\ \wedge\\
	&\mathtt{shared}(\trg{m_{recomp}}, t_{1,2})
	\sim_\mathtt{ren}
	%_{\mathtt{ren}(\trg{P_1}, \trg{P_2}, \trg{C_1}, \trg{C_2})}
	\mathtt{shared}(\trg{m_{base}}, t)
	\end{align*}
	and
	\begin{align*}
	&\mathtt{mem\_rel\_not\_exec}(\trg{part}, t, t_{1,2}, \trg{m_{base}}, \trg{m_{recomp}}) \defeq \\
	&\mathtt{proj}_{\trg{part}} (\trg{m_{recomp}})\ \cap\ 
	\mathtt{private}(\trg{m_{recomp}}, t_{1,2}) \\
	&\sim_\mathtt{ren}
	%_{\mathtt{ren}(\trg{P_1}, \trg{P_2}, \trg{C_1}, \trg{C_2})}
	\mathtt{proj}_{\trg{part}} (\trg{m_{base}})\ \cap\ 
	\mathtt{private}(\trg{m_{base}}, t)
	\end{align*}
\end{definition}

Intuitively, \Cref{def-mem-rel-tt} says the following
about \trg{P_1}'s memory: (a) While $\trg{P_1}$ executes, \trg{P_1}'s
entire memory---both private and shared---is related in the runs of the
base program $\trg{P_1 \cup C_1}$ and the recomposed program
$\trg{P_1 \cup C_2}$. (b) While the contexts (\trg{C_1} and \trg{C_2})
execute, only the private memory of $\trg{P_1}$ in these two runs is
related. For the context's memory, the dual relation
holds. \Cref{fig-mem-rel-tt} depicts this visually.
\begin{figure}[t]
	\centering \includegraphics[width=.5\textwidth]{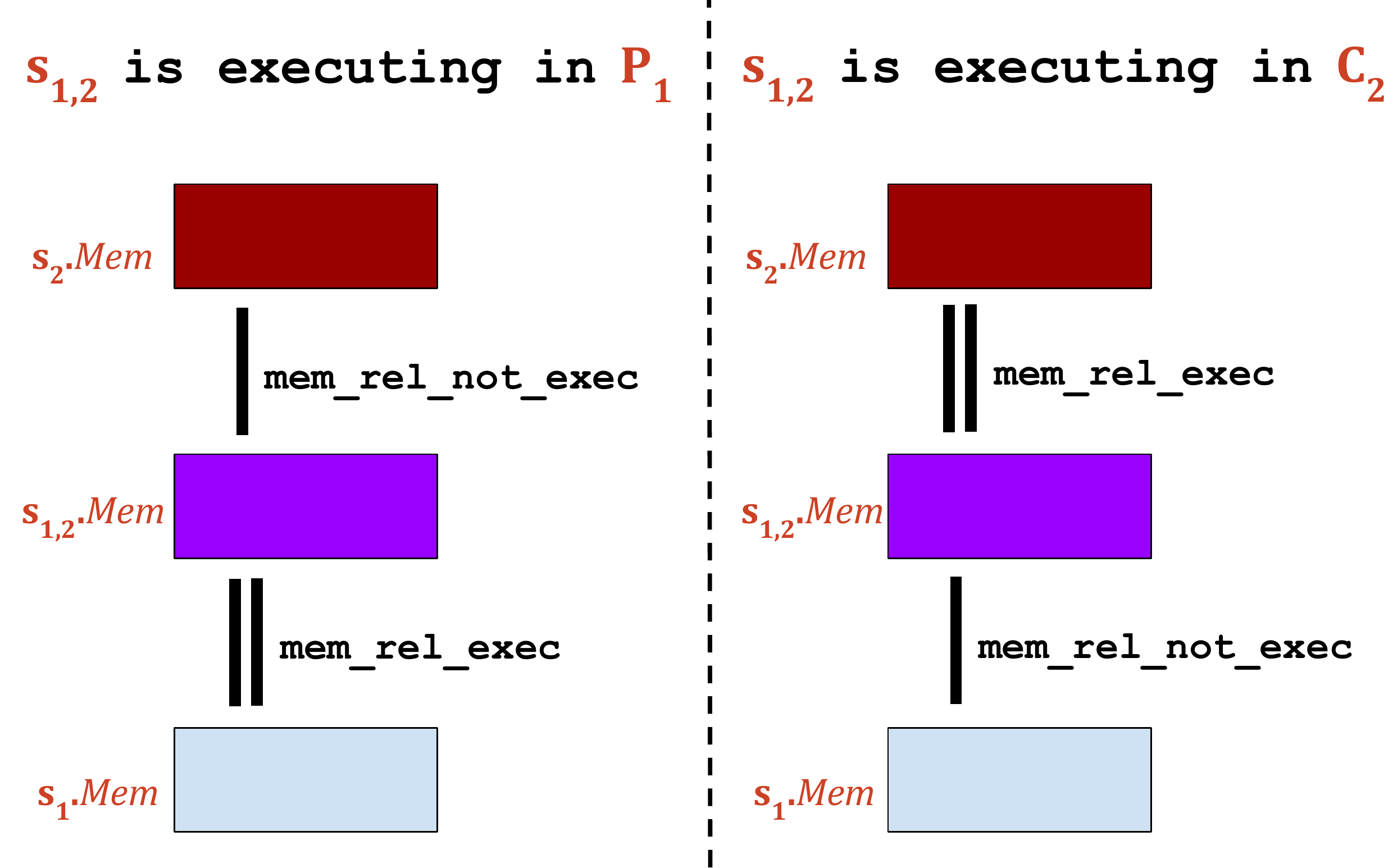} \caption{The
	turn-taking memory relation,
	$\mathtt{mem\_rel\_tt}$.}  \label{fig-mem-rel-tt}
\end{figure}

%% is executing, we are able to relate the recomposed memory
%% to the memory of the corresponding base program at all the locations
%% that \trg{part} allocates and also at all the shared locations (this
%% strong relation is given by $\mathtt{mem\_rel\_exec}$). On the other
%% hand, the relation between the recomposed memory and the memory of the
%% \emph{other base program} (i.e., from which \trg{part} did
%% not come) is weaker and holds only at the 
%% \mi{private} locations ($\mathtt{mem\_rel\_not\_exec}$).

%% It is helpful to visualize (using \Cref{fig-mem-rel-tt}) 
%% the definition of the invariant
%% $\mathtt{mem\_rel\_tt}$ by observing the turn-taking pattern
%% of using $\mathtt{mem\_rel\_exec}$ and $\mathtt{mem\_rel\_not\_exec}$ interchangeably with the
%% arguments \trg{P_1} and \trg{C_2}.

\paragraph*{The memory relation $\sim_\mathtt{ren}$.}
We now explain the memory relation $\sim_\mathtt{ren}$ that appears in the
above definitions. This relation simply allows for a consistent
renaming of memory locations up to a partial bijection. The need for
this renaming arises because corresponding program parts may differ in
the layouts of their private memories. In the example above, consider the
case where the component $\trg{P_2}$, which we didn't show until now, is
the same as $\trg{P_1}$, just without the private
pointer \lstinline{priv_ptr} and the corresponding malloc. In this
case, the exact value of \lstinline{shared_ptr} could differ across
the base run $\trg{P_2 \cup C_2}$ and the recomposed run
$\trg{P_1 \cup C_2}$.
Formally, $\mathtt{ren}$ denotes a partial bijection that may depend
on \trg{P_1}, \trg{P_2}, \trg{C_1} and \trg{C_2}, and
$\sim_\mathtt{ren}$ is renaming of memories (both locations and their
contents) up to $\mathtt{ren}$.

\paragraph{Proof of recomposition.}
In our Coq proof we\add{ effectively} 
 show  that the turn-taking memory relation
of \Cref{def-mem-rel-tt} is an invariant\add{ of
the execution of any recomposed program} \remove{(for the languages and
compiler of \autoref{sec:compiler})}
\add{of the target language of \autoref{sec:compiler}}. 
\add{Formally, this follows from two lemmas 
(\cref{lemma-option-sim,lemma-lockstep-sim})
that can be seen as \emph{expected properties of 
the memory relation}.}
Using these lemmas, we are able to prove
recomposition (\Cref{lemma-abateetal-recomb}). A key additional idea
we use is \emph{strengthening}, 
which we apply at cross-component calls and
returns to strengthen
$\mathtt{mem\_rel\_not\_exec}$ into 
$\mathtt{mem\_rel\_exec}$. 
The former relates only the private memory
of a component, the latter relates private and
shared memories of the same component. Strengthening follows from the assumption that the two base
runs emit related interaction traces. \autoref{sec:recomb} provides
additional details.
\ifsooner
\ch{Unclear it's a good idea to introduce strengthening here.
  Can't this move to Section 5 where we now discuss it?}
\fi

\subsection{Applying our ideas to an \rsp proof}
\label{sec:rsp-proof}
% !TEX root = ./paper.tex

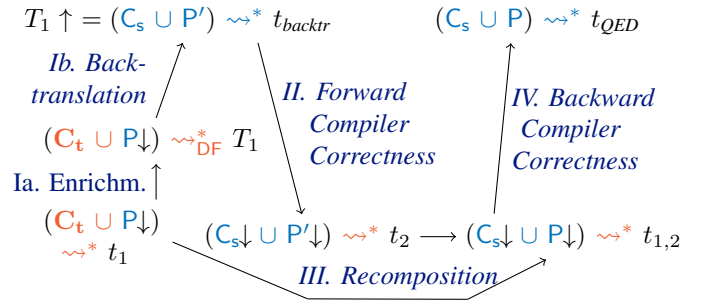
\begin{figure}
	\hspace*{-5mm}\makebox[\textwidth][l]{
		\begin{tikzpicture}[auto]
  \node(cPi) [text width=3cm] {$(\trg{C_t}\ \trg{\cup}\ \cmp{\src{P}})$ \\ \ \ $\trg{\doesprefix}\ t_1$ };
  \node[right = of cPi,xshift=-6em] (CP1t) { $(\cmp{\src{C_s}}\src{\cup}\ \cmp{\src{P'}})\ \trg{\doesprefix}\ t_2$ };
  \node[align = left, above = of cPi,yshift=-1.5em] (cPiT) { $(\trg{C_t}\ \trg{\cup}\ \cmp{\src{P}})\  \trg{\emitsdfT}\ T_1$ };
  \node[align = left, above = of cPiT,xshift=1em] (CP1s) { $\back{T_1\ }$ $= (\src{C_s}\ \src{\cup}\ \src{P'})\  \src{\doesprefix}\ t_\mi{backtr}$ };
  \node[right = of CP1t,xshift=-1.5em] (CPt) { $(\cmp{\src{C_s}}\src{\cup}\ \cmp{\src{P}})\  \trg{\doesprefix}\ t_{1,2}$ };
%  \node[above = of cPi, xshift=12em, yshift=1.5em] (prec) { $m \leq t \lor t \prec_P m$ };
  \node[above = of CPt, right = of CP1s] (CPs) { $(\src{C_s}\ \src{\cup}\ \src{P})\ \src{\doesprefix}\ t_\mi{QED} %\wedge (m \leq t \lor t \prec m)
  	$
  };

  \draw[->] (cPiT.90) to node
  [text width=2cm, xshift=-0.5em,yshift=-1.2em,align=right] {\hyperref[lemma-robustptrs-backtrans]{\em Ib. Back-\\translation}} (CP1s.-90);

  \draw[->] (cPi.90) to node
[text width=2cm,align=right,font=\itshape] {\hyperref[lemma-enrichment]{\em Ia. Enrichm.}} (cPiT.-90);

  \draw[->] ([xshift=2em]CP1s.-45) to node [align=left,font=\itshape,xshift=-0.25em,yshift=-1.7em,
  text width=2.5cm]{\hyperref[assm-fwdsim]{II. Forward \\\ \ \ Compiler \\ \ \ \ Correctness}} (CP1t);

  %\draw[->] (CP1s.-8) to node [right,yshift=-0.5em, xshift=3em]{\em \hyperref[asm:blame]{5 Blame}} (prec.180);
  %\draw[->] (CPs.140) to node {} (prec.0);

  % composition target
  \draw[->] (CP1t.0) to node [below,xshift=-1.5em,yshift=-0.8em]{\em \hyperref[lemma-robustptrs-recomp]{III. Recomposition}} (CPt.180);
  \coordinate [below=1.5em of CP1t,xshift=0em] (compoint1);
  \coordinate [below=1.5em of CPt,xshift=-4em] (compoint2);
  \draw[-] ([xshift=-4em,yshift=1.2em]cPi.south east) to (compoint1);
  \draw[-] (compoint1) to (compoint2);
  \draw[->] (compoint2) to (CPt.220);

  \draw[->] ([xshift=-3em]CPt.90) to node [right,align=left,font=\itshape,text width = 2.7cm]{\hyperref[assm-bwdsim]{IV. Backward \\ \  Compiler \\  Correctness}} ([xshift=-1em]CPs.-90);

  % side labels
  %\node[left = of CP1s, xshift=0em] () {\color{gray} Source};
  %\node[left = of cPiT, xshift=4em,yshift=-2em] () {\color{gray} Target};
\end{tikzpicture}
	}
	\caption{Our proof technique for \rsp with memory sharing. The
          interaction traces $t_1$, $\mathtt{remove\_df}(T_1)$,
          $t_\mi{backr}$, $t_2$, $t_{1,2}$, and $t_\mi{QED}$ are
          pairwise related by the trace relation $\sim$.}
	\label{fig:rsc-proof}
\end{figure}

\Cref{fig:rsc-proof} summarizes our overall proof technique for
proving \rsp with dynamic memory sharing. $\back{\ }$ denotes the data-flow
back-translation function. Relative
to~\citepos{Abate:2018:GCG:3243734.3243745} proof technique shown in
\Cref{fig:rsc-proof-old}, the two key changes are that: (1)~Step I
(back-translation) has now been factored into two steps Ia and Ib to
use data-flow traces. Steps Ia and Ib correspond to
\Cref{lemma-enrichment} and \Cref{lemma-robustptrs-backtrans},
respectively. (2)~The proof of step III (recomposition) now relies on
turn-taking simulations. Steps II and IV, which simply reuse compiler
correctness, remain unchanged.

\paragraph{Trace relation $\sim$}
We now also define the trace relation $\sim$, mentioned in
\autoref{sec:background} and \autoref{sec:keyideas}. It says that
two traces are related if corresponding events have the same kind
(both call or both return, and between the same components), and there
is a bijective renaming of locations $\mathtt{ren}$ such that the
memories mentioned in corresponding events of the traces are related
by $\sim_\mathtt{ren}$ (\autoref{sec:key-ideas-turn-taking}), and so are
the arguments of calls and returns.

\begin{definition}[Relation on interaction traces]
	\label{def-trace-rel}
	For address renaming relations $\mathtt{ren}$, 
	suppose
	$\sim_\mathtt{ren}$ is the memory renaming
	relation described in \autoref{sec:key-ideas-turn-taking}.
	\begin{align*}
	t_1 \sim t_2\ \defeq\  \exists \mathtt{ren}.\ \forall i.
        &\ \ \ t_1[i].\mi{Mem}\ \sim_\mathtt{ren}\ t_2[i].\mi{Mem}\\
	&\ \ \wedge\ \mathtt{match\_events}(t_1[i], t_2[i])\\
	&\ \ \wedge\ \mathtt{valren}_\mathtt{ren} 
	(t_1[i].\mi{arg}, t_2[i].\mi{arg})
	\end{align*}
	\add{Here $t[i]$ denotes the $i$th event of trace $t$.
	Notation $t[i].\mi{Mem}$ is the memory that appears in the
	event $t[i]$ (see the \Cref{def-mem-sharing-trace-event}
	of events).
}
	$\mathtt{match\_events}(e_1, e_2)$ 
	says that the kind of events
        $e_1$ and $e_2$ \add{(again see \Cref{def-mem-sharing-trace-event}
	for the two possible kinds)} 
        and the component ids 
        appearing on them (\EG caller and callee) are the same.
	$\mathtt{valren}_\mathtt{ren}$ is a value renaming relation that
	just lifts the \remove{memory}\add{address} renaming relation $\mathtt{ren}$ to pointers.
	\add{We give a precise definition of $\mathtt{valren}_\mathtt{ren}$ in \autoref{sec:compiler}.}
\end{definition}

\section{Concrete Languages and Compiler Pass}
\label{sec:compiler}
% !TEX root = ./paper.tex

Next, we describe specific source and target languages%
---\sourcelanguage and \targetlanguage, respectively---and a specific
compiler from the source to the target language. This
specific setup is the testbed on which we have instantiated our new
ideas from \autoref{sec:keyideas}.
In both languages, a program $P$ consists of a set
of named functions, a set of statically allocated data buffers and an interface. The
interface divides the program into components (denoted $c$) and
assigns every function to a component. It also defines which
functions are imported and exported by each component.
%% \rb{A little better now. Technically functions are assigned to components
%% by the ``code map'' (which has to agree with the interface)
%% and static buffers we sort of consider part of the interface for
%% back-translation purposes.}

\paragraph{Values, pointers and memory}
Both languages are memory safe and \remove{share}%
\add{use} the same memory model, which
is adapted from CompCert's block-memory
model~\cite{10.1007/s10817-008-9099-0}.
A
value $v$ may be an integer $i$, an (unforgeable) pointer, or a special
\remove{undefined}\add{error} value
$\mathtt{error}$ used to initialize memory.\footnote{Another
possibility could have been to model 
$\mathtt{error}$ 
as a fixed default
integer instead (like zero), so not
necessarily a separate runtime type.
}
\ifsooner
\ch{Since this language is safe, I think we should get rid of
	the undefined value (\IE replace it with the defined value zero). It currently causes
	confusion about register cleaning (see comment about compiler at end of this
	section), and also keeps popping up in the explanation of the backtranslation.}
\fi
A \emph{pointer} is a tuple $(\mi{perm}, c, b, o)$ consisting of a
permission $\mi{perm}$ (used to distinguish code and data pointers),
the \add{identifier $c$ of the }allocating component\remove{ $c$}, a unique block
identifier $b$,
%\rb{within the component (but doesn't make a difference)}
% CH: doesn't seem like a very relevant detail for the paper indeed
and an integer offset $o$ within the block. A location, which we denoted by
$l$ so far, is a triple of a component\add{ id}, a block\add{ id}, and an offset, $(c, b,
o)$.

A memory maps locations to values.
\remove{Briefly,}\add{CompCert's} memory consists of an
unbounded number of finite and isolated \emph{blocks} of values. 
\remove{The memory }\add{The memory in both 
our languages is similar, but is additionally
partitioned by component ids. It }can be seen as a collection 
\add{($c \rightharpoonup \mi{cMem}$)}
of individual
component memories 
\add{($\mi{cMem} = b \rightharpoonup (o \rightharpoonup v)$).}
\add{The projection operator that we
% introduced informally and
used in \Cref{def-mem-rel-abateetal,def-mem-rel-tt}
is formally defined as $\mathtt{proj}_{P}(\mi{Mem})
\defeq [ c \mapsto (\mi{Mem}\ c) \mid 
c\ \in\ \mathit{component\_ids}(P) ]$, returning
a sub-collection of the collection $\mi{Mem}$
containing
just the component memories that correspond
to the components of the program part $P$.}
\remove{because }\add{Although }each memory block is initially accessible
to only the allocating component\remove{.}\add{, 
	memory sharing is allowed, so the \emph{contents (i.e., values $v$)}
	of a component memory can be pointers to other
	component memories. In particular,
the contents of a component memory in the collection
$\mathtt{proj}_{P}(\mi{Mem})$ can very well be
pointers to a component memory that happens
to \emph{not} be in the collection.}
% The ownership\ch{Ownership seems too fancy, and can
%   give the wrong impression now that we have dynamic memory sharing.
%   What's the purpose of this ``ownership'' information anyway? It's not
%   like other components can't access shared pointers just because of this
%   ``ownership''. Is it just so that addresses are stable to allocations
%   done in other components? Can we call it ``allocating component'' instead
%   of ``owner''?}
% \aek{Suggestion: ``The memory can be seen as a collection of
% individual component memories because each memory block
% is initially accessible to only the allocating component.''}
% of each memory block is given
% to the component that allocates it, so a memory can also be seen as a
% collection of individual component memories.

Pointers can be incremented or decremented (pointer arithmetic), but
this only changes the offset $o$. The block identifier $b$ cannot be
changed by any language operation. Additional metadata not shown here
tracks the size of each allocated block. Any dereference of a data
pointer with an offset beyond the allocated size or any call/jump to a
code pointer with a non-zero offset causes the program to halt,
which enforces memory safety.
\add{Code pointers can be shared between components, but a 
	component cannot dereference code pointers to another
component.\footnote{This condition is not unrealistic and can be
realized on, \EG CHERI by implementing code pointers either as mere
integer offsets or as sealed capabilities, but \emph{not} as unsealed
capabilities with execute permission.} Components interact only by
calling exported functions of other components and by sharing memory.}

Our languages are strongly inspired by those
of \citet{Abate:2018:GCG:3243734.3243745} but, unlike them, we allow a
component to pass pointers
% to the locations it allocates -- also to locations allocated by others!
to other components. The
receiving components can dereference these pointers, possibly after
changing their offsets. However, a component cannot access a block
without allocating it itself or receiving a location from it. Hence, our
languages provide memory protection at block granularity.

\add{Block ids are subject to renaming when relating two component implementations.
Our memory and trace relations
(\Cref{def-mem-rel-tt,def-trace-rel})
relate two 
 implementations
of a component even when the 
concrete block ids of pointers
that they share with the outside
world are different, as long as there exists 
a function\footnote{See
the Coq file \texttt{Common/RenamingOption.v}}
that consistently renames the pointers shared
by one implementation into those of the second.
With such a block id renaming
$\mathtt{ren} : b \rightharpoonup b$ in hand,
one can define value renaming (which we introduced
informally in \Cref{sec:rsp-proof}) as follows:}
\begin{align*}
&\bullet i_1 = i_2 \implies \mathtt{valren_{ren}}(i_1, i_2)\\
&\bullet  \mathtt{valren_{ren}}(\mathtt{error},\mathtt{error})\\
&\bullet \mathtt{ren}(b, b') \implies 
\mathtt{valren_{ren}}((\mathtt{DATA}, c, b, o),
(\mathtt{DATA}, c, b', o))\\
&\bullet \mathtt{valren_{ren}}((\mathtt{CODE}, c, b, o),
(\mathtt{CODE}, c, b, o))
\end{align*}
\add{The only block-id-renaming
relations we actually use in our proofs
are the identity, and increment-by-1
(in \Cref{fig:layout-back-translation}).
}

The operational semantics of both languages produce interaction traces
of events from \Cref{def-mem-sharing-trace-event}, recording
cross-component calls and returns. Calls and returns are necessarily
well-bracketed in the semantics.

The two languages differ
significantly in the constructs allowed within the bodies of
functions, as we describe next.

\begin{figure}
%% \begin{small}
  \begin{tabular}{lll}
  \texttt{\src{exp}}& \texttt{::=} \quad\src{v}               & values            \\
                 & | \src{ arg}                & function argument \\
                 & | \src{ local}              & local static buffer \\
                 & | \src{ exp_1\ \otimes\ exp_2} & binary operations \\
                 & | \src{ exp_1;\ exp_2}       & sequence          \\
                 & | \src{ if\ exp_1\ then\ exp_2\ else\ exp_3} & conditional \\
                 & | \src{ alloc\ exp}            & memory allocation \\
                 & | \src{ !exp}                 & dereferencing     \\
                 & | \src{ exp_1 := exp_2}     & assignment        \\
                 & | \src{ c.func(exp)}             & function call    \\
                 & | \src{ *[exp_1](exp_2)}    & call pointer      \\
                 & | \src{ \&func}              & function pointer  \\
                 & | \src{ exit}               & terminate
\end{tabular}
%% \end{small}
%% \vspace{-0.8em}
\vspace{5pt}
\caption{Syntax of source language expressions
  %\cite{AbateABEFHLPST18}
}
\label{fig:source-syntax}
%% \vspace{-1em}
\end{figure}

\paragraph{The source language (\sourcelanguage)}
The body of a \sourcelanguage function is a single
expression, $\src{exp}$, whose syntax is shown
in \Cref{fig:source-syntax} and is inspired by the source
language of \citet{Abate:2018:GCG:3243734.3243745}.
The construct
$\src{arg}$ evaluates to the argument of the current function, which
is a value (which may be a pointer). There are constructs for
if-then-else, dereferencing a pointer ($\src{!exp}$), assigning value
to a pointer ($\src{exp_1 := exp_2}$), calling a function $\src{func}$
in component $\src{c}$ with argument $\src{exp}$
($\src{c.func(exp)}$), calling a function pointer $\src{exp_1}$
($\src{ *[exp_1](exp_2)}$), and taking the address of a function
($\src{\&func}$). Additionally, every component has access to a
separate statically allocated memory block, whose pointer is
returned by the construct $\src{local}$.

Importantly, the source language has only \emph{structured control
flow}: Calls and returns are well-bracketed by the semantics, the only
explicit branching construct is if-then-else, and indirect function
calls with non-zero offsets beyond function entry points are stopped
by the semantics.

%% \ch{\bf Do we want to mention that function pointers is something we added,
%%   not only because it makes the source more realistic, but also that
%%   it made the proof easier/possible? It's not a crazy construct like
%%   parallel or, but still, it's in part an artifact of the proof.}
%% \aek{\bf Yes, we should but only if we will explain \emph{how} it
%%   makes the proof easier/possible.\ch{+1} Maybe such a discussion
%%   belongs in a Limitations sections?}\ch{\bf This doesn't seem a big limitation
%%   to me, how about adding it here and we will see.}

Function pointers exist in \sourcelanguage not only because they
are a natural programming feature, but also to make specific steps of the back-translation
convenient.
Function pointers allow us, \EG to easily mimic a store of the program counter to
memory, an operation that a target-language program routinely performs.
Without function pointers in the source, our cross-language value relation may
have been more complex---a complexity that would propagate to the trace
relation (\Cref{def-trace-rel}) and to the top-level theorem (\Cref{rsp-def}).

%% \rb{Discuss safety of the setting vs. \cite{Abate:2018:GCG:3243734.3243745}}
%% \ch{The main thing that's different is that at CCS'18 we compiled Interim
%%   in an unsafe way. But that's not an inherent characteristic of the source language.
%%   So if we want to mention this we could do that after explaining the compilation
%%   step we prove?}

\begin{figure}
%% \begin{small}
\begin{tabular}{l@{\hskip0pt}l@{\hskip20pt}l}
  \trg{instr ::=} & ~\trg{ Const\ i\ \texttt{->}\  r}                   & | \trg{ Bnz\ r\ L}         \\
                     & | \trg{ Mov\ r_s\ \texttt{->}\  r_d}                 & | \trg{ Jump\  r}             \\
                     & | \trg{ BinOp\ r_1\ {\otimes}\ r_2\ \texttt{->}\ r_d}   & | \trg{ JumpFunPtr\  r}       \\
                     & | \trg{ Label\  L}                         &  | \trg{ Jal\ L}             \\
                     & | \trg{ PtrOfLabel\ L\ \texttt{->}\  r_d}           & | \trg{ Call\ c\ func}         \\
                     & | \trg{ Load\ *r_p\ \texttt{->}\  r_d}              & | \trg{ Return}               \\
                     & | \trg{ Store\ *r_p\ \texttt{<-}\  r_s}             & | \trg{ Nop}                  \\
                     & | \trg{ Alloc\ r_1\  r_2}                 & | \trg{ Halt}                 \\
  \end{tabular}
%% \end{small}
%% \vspace{-0.8em}
\vspace{5pt}
\caption{Instructions of the target language
  %% \cite{AbateABEFHLPST18}
}
\label{fig:target-syntax}
%% \vspace{-1em}
\end{figure}

\paragraph{The target language (\targetlanguage)}
\targetlanguage is an assembly-like language inspired by RISC
architectures, with two high-level features: the block-based memory
model shared with \sourcelanguage and the component structure provided
by interfaces. Its instructions are shown
in \Cref{fig:target-syntax}. Its state comprises a register file with
a separate program counter and an abstract (protected) call stack for
cross-component calls, which enforces well-bracketed
cross-component \trg{Call}s and \trg{Return}s. A designated
register \trg{r_{COM}} is used for passing arguments and return
values. At every cross-component call or return, all registers except
$\trg{r_{COM}}$ are set to \remove{the undefined value.}\add{$\mathtt{error}$.}\ifsooner\ch{It would be much better
  if this said: ``are set to zero''}\fi

Importantly, \targetlanguage has \emph{unstructured control flow}: One
may label statements (instruction $\trg{Label~L}$), jump to labeled
statements ($\trg{Jump~L}$, $\trg{Bnz~r~L}$), and call labeled
statements ($\trg{Jal~L}$). Such unstructured jumps
are \remove{limited}\add{confined} to a single component
\add{(see the boxed premise of
the rule for \trg{Jump} that enforces this
restriction)}%
\remove{;
cross-component jumps are forbidden by the semantics}, \add{but may
cross intra-component function boundaries. This makes it infeasible
to syntactically back-translate \targetlanguage to \sourcelanguage.}
\rb{Tried to center rules quickly but failed :-) Do these need a bit of explanation?}

\medskip
\inferrule[Jump]{
	\trg{\ii{fetch}(E, \ii{pc})} = \trg{Jump\ r}\\
	\boxed{\trg{\ii{pc'}} = \trg{\ii{reg}[r]}} \\
	\trg{\ii{is\_code\_pointer(pc')}} \\
	\boxed{\trg{\ii{comp}(\ii{pc})} = \trg{\ii{comp}(\ii{pc'})}}
}{
	\trg{E \vdash (\sigma, \ii{mem}, \ii{reg}, \ii{pc}) \xrightarrow{\bl{[]}}
		(\sigma, \ii{mem}, \ii{reg}, \ii{pc'})
}}
~\\

\remove{the presence of unstructured control flow in
the target language means that a syntax-directed back-translation to
the source language, which has only structured control flow, is
infeasible. }

In addition to the interaction trace
semantics (\Cref{def-mem-sharing-trace-event})
like \sourcelanguage's, \targetlanguage
also enriches the trace with 
 data-flow events
(\Cref{def-dataflow-events-concrete}) as explained
in \autoref{sec:key-ideas-data-flow-traces} and
illustrated by the boxed premise of the \trg{Store} rule:

\medskip
\inferrule[Store]{
	\trg{\ii{fetch}(E, \ii{pc})} = \trg{Store\  *r_p\ \texttt{<-}\  r_s}\\
	\trg{\ii{ptr}} = \trg{\ii{reg}[r_p]}\\
	\trg{v} = \trg{\ii{reg}[r_s]}\\
	\trg{\ii{mem'}} = \trg{\ii{mem}[\ii{ptr} \mapsto v]}\\
	\boxed{\bl{\alpha} = \trg{\bl{\mathtt{Store}}~\mathit{mem'~reg~\ii{comp}(pc)}~r_p~r_s}}
}{
	\trg{E \vdash (\sigma, \ii{mem}, \ii{reg}, \ii{pc}) \xrightarrow{\bl{[\alpha]}}
		(\sigma, \ii{mem'}, \ii{reg}, \ii{pc} + 1)
}}

%% \aek{Talk about
%%  Intermediate/CSInvariants.v? talk intuitively about the fact that the recombination 
%% 	lemma relies on security invariants on the execution of 
%% 	a \targetlanguage program.}

\paragraph{Compiler from \sourcelanguage to \targetlanguage}
Our compiler from \sourcelanguage to \targetlanguage is single-pass
and quite simple. It implements \sourcelanguage's structured control
flow with labels and direct jumps, intra-component calls using
jump-and-link (\trg{Jal}), and function pointer calls using indirect
jumps (\trg{Jump} and \trg{JumpFunPtr}).%

Even this simple compilation chain, where security is mostly enforced
by the target language semantics,
suffices to bring out the difficulties in proving secure compilation in
the presence of memory sharing. Using the ideas developed in \autoref{sec:keyideas}, we have proved
that this compilation chain provides \rsp security.

\begin{theorem}\label{theorem-rsp}
Our \sourcelanguage to \targetlanguage compiler is
\rsp (\IE it satisfies \Cref{rsp-def}).
\end{theorem}

%% \rb{Here we should say a few things about the calling
%%   convention enforce by the compiler, 
%%   the ``axioms'' we rely on (separate
%%   compilation\ldots). But we probably don't have much space.}

%\section{\rsp proof for the \sourcelanguage to \targetlanguage compiler}
%\label{sec:rsp-proof}
%\input{rsp-proof}

\section{Some Details of the Coq Proof}

\subsection{Data-Flow Back-Translation of \targetlanguage}
\label{sec:back-trans}

We provide some details of how we back-translate \targetlanguage's
data-flow traces to \sourcelanguage, i.e., how we
prove \Cref{lemma-robustptrs-backtrans}. The back-translation
function, written $\back{\ }$, takes as input a data-flow trace $T$ and
outputs a \sourcelanguage whole program \src{P} that produces the (standard)
trace $\mathtt{remove\_df}(T)$ in \sourcelanguage.\footnote{$\back{\
}$ also takes as input the \emph{interface} of the given
target-language program to be able to mimic the same interface in the
source program, but we elide the details as they are not very
insightful, and largely similar to those
in \citet{Abate:2018:GCG:3243734.3243745}. \src{P} can then be split
into a context \src{C_S} and a program part \src{P'} by slicing it along this
interface.}
As for \citet{Abate:2018:GCG:3243734.3243745},
each component in \src{P} maintains an \emph{event counter} to keep track of which
trace event the component is currently mimicking.  This counter, as
well as a small amount of other metadata used by the back-translation,
is stored inside the statically allocated buffers of each component
of \src{P}, which are accessed using the $\src{local}$ construct.

\add{\paragraph{Control flow of the
		result of the
		back-translation}
	The outermost structure and control flow of the
	result of our data-flow back-translation is very
	similar to that of
	\citet{Abate:2018:GCG:3243734.3243745}'s
	interaction-trace-directed back-translation.
	Every procedure has a main loop (implemented using
	a tail-recursive call) that emits, one after the other,
	the events this procedure's component is responsible for
	emitting.
	In the ``loop body'', the event counter mentioned above is checked
	using a switch statement to determine the event whose turn it is
	to be emitted.
	%% So far, this construction is similar to
	%% \cite{Abate:2018:GCG:3243734.3243745}.
	%% JT: already mentionned earlier in the paragraph
}

\paragraph{Mimicking register operations}
A technical difficulty in the back-translation is that, unlike
\targetlanguage, \sourcelanguage does not have registers.
In order to mimic data-flow events involving registers, \src{P}
simulates these registers and operations on them within the static
buffer of the active component.
%
% Whenever the back-translated program must mimick an event that interacts
% with one or several registers, it instead interacts with the memory
% locations that mimic those registers.
%
For instance, a
$\mathtt{Mov}~\mathit{Mem~Reg~c_{cur}~r_{src}~r_{dest}}$ event (which
copies a value from register $\mathit{r_{src}}$ to register
$\mathit{r_{dest}}$) is simulated by the expression $\src{(local +
OFFSET(\mathit{r_{dest}})) :=~!(local + OFFSET(\mathit{r_{src}}))}$,
where $\src{OFFSET(\mathit{r})}$ is statically expanded to the offset
corresponding to register $\mathit{r}$ in the simulated register file.

\paragraph{Mimicking memory operations}
Because like in CompCert the source and target memory models coincide,
we are able to back-translate memory events quite easily.
That is, a $\mathtt{Store}$ event is back-translated using assignment
($\src{:=}$) and a $\mathtt{Load}$ event is back-translated using
dereferencing ($\src{!}$).
Since the static buffer (whose block number is $0$ in our semantics)
is already used by the back-translation to store metadata and
simulated registers, the back-translated program's memory shifts by
one block relative to the memory in the target: for each component,
block $b$ in the target corresponds to block $b{+}1$ in the source.
The memory layout of the back-translated program
($\src{P} {=} T\!\!\!\uparrow$) relative to
the given \targetlanguage program \trg{P} is shown
in~\Cref{fig:layout-back-translation}.
\begin{figure}
	%Edit figure here:
    % https://docs.google.com/drawings/d/1Ki9tqYvYqWeAIiCZUUS2hK3wH6xUlcre4rOsqpwpRac/edit?usp=sharing
	\centering
	\includegraphics[width=0.49\textwidth]{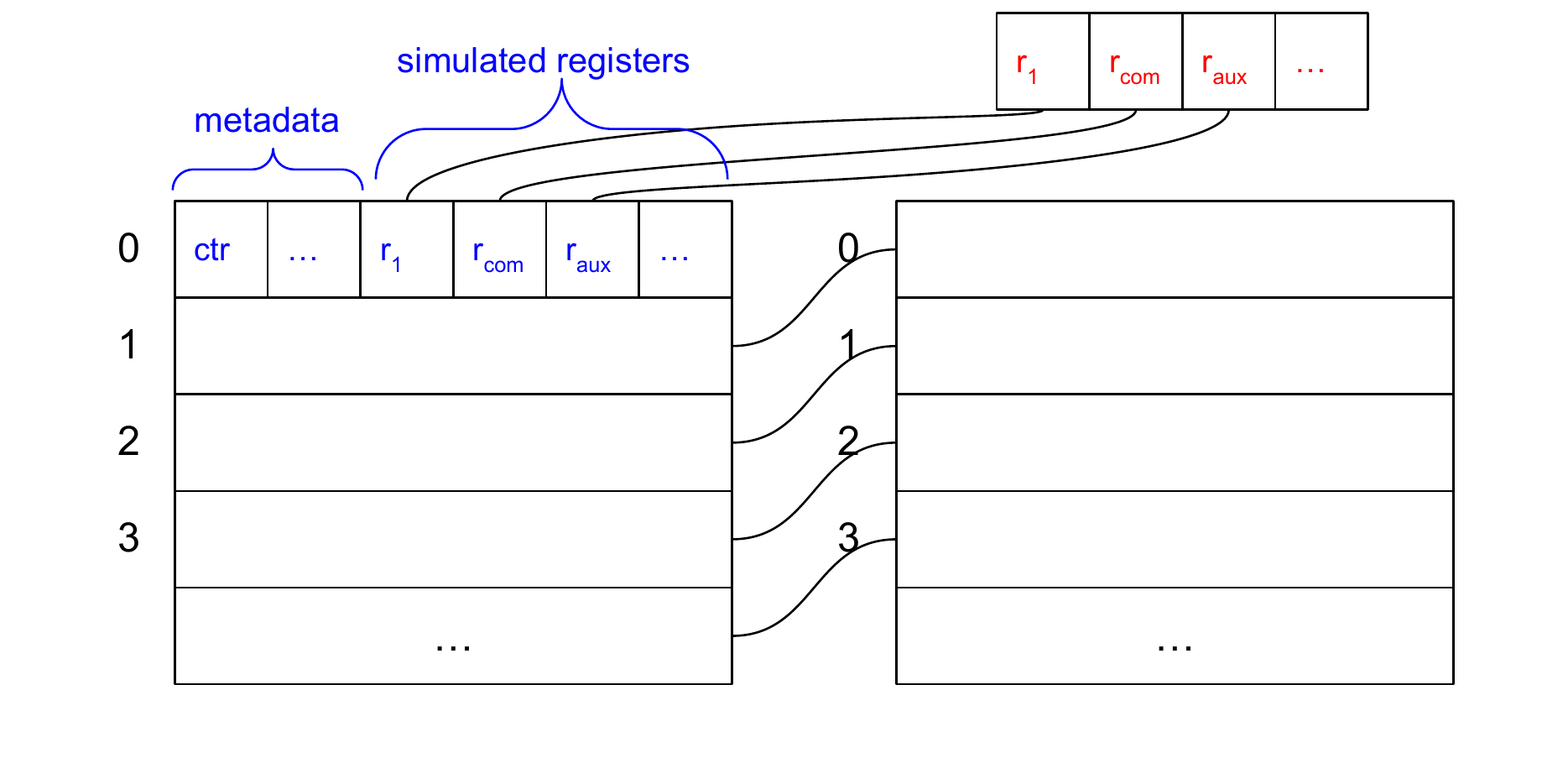}
	\caption{Memory layout of a back-translated component (left) compared to a target component (right) witnessing the increment-by-1 block-id-renaming relation}\label{fig:layout-back-translation}
\end{figure}
\src{P} maintains the invariant that, after simulating an event in $T$,
\src{P}'s memory and its current component's simulated registers are
synchronized with the target memory $\mi{Mem}$ and the target register
file $\mi{Reg}$ mentioned in the simulated event
(This is part of a \texttt{mimicking\_state}
invariant---see \Cref{lemma-definability} below).

\paragraph{Mimicking calls and returns}
\targetlanguage's semantics enforce a calling convention: calls and
returns store the argument or return value in \trg{r_{COM}}, and set
all other registers to \remove{the undefined value}\add{$\mathtt{error}$}.\ch{to zero}
Therefore, calls and returns in \src{P} need extra administrative
steps to mimic this convention. For example, mimicking a call event
requires two administrative steps: (1) In the caller, dereference the
content of the location simulating \trg{r_{COM}} to get the argument
and pass it to the function. (2) In the callee, assign the
function argument $\src{arg}$ to the location
simulating \trg{r_{COM}}, and set all other registers to \remove{the undefined value}\add{$\mathtt{error}$}.\ch{zero}
Similar administrative steps are needed for mimicking a return event.

\paragraph{Proof of back-translation}
To prove back-translation (\Cref{lemma-robustptrs-backtrans}), we use
a simulation lemma that ensures a relation \texttt{mimicking\_state}
holds between the state of \src{P} and the prefix mimicked so
far. Intuitively, $\texttt{mimicking\_state}\ T_{\mi{pref}}\
T_{\mi{suff}}\ \src{s}$ means that $\src{s}$ is the state reached
after mimicking all the data-flow events in $T_{\mi{pref}}$, and that
the starting state of the remaining trace $T_{\mi{suff}}$ matches
$\src{s}$. 

\begin{lemma}[Trace-prefix mimicking]
	\label{lemma-definability}
\begin{align*}
  &\forall \trg{P}\ T\ T_{\mi{pref}}\ T_{\mi{suff}}.\ \
  \trg{P \emitsdfT}\ T \implies T = T_{\mi{pref}}\ \texttt{++}\ T_{\mi{suff}} \implies\\
  &\ \ \ \exists \src{s}\ t_{\mi{pref}}'.\ \ 
	T\!\uparrow\src{\ \step{\bl{t'_\mi{pref}}}^{*}\ }
	\src{s}\ \wedge
        t_{\mi{pref}}' \sim \mathtt{remove\_df}(T_\mi{pref})\ \\
  &\ \ \ \ \ \ \ \ \ \ \ \ \ \wedge \texttt{mimicking\_state}\ T_{\mi{pref}}\ T_{\mi{suff}}\ \src{s}
\end{align*}
\end{lemma}
%% \begin{proof}
%% By induction on the data-flow trace prefix $T_{\mi{pref}}$.
%% \end{proof}
Because \Cref{lemma-definability} ensures the
relation \texttt{mimicking\_state} holds for every prefix, it
effectively states that the memory of the back-translation is in
lock-step with the $\mi{Mem}$ and $\mi{Reg}$ appearing in each
data-flow event $\mathcal{E}$ from $T$.
\texttt{mimicking\_state} is also strong enough to ensure
that the trace relation holds between the projection of the prefix
mimicked so far $\mathtt{remove\_df}(T_\mi{pref})$ and the
corresponding prefix $t_\mi{pref}'$ that the back-translation emits.

\add{
The fully mechanized Coq proof of \Cref{lemma-robustptrs-backtrans} is in
\texttt{Source/DefinabilityEnd.v}, which in turn uses
\texttt{Source/Definability.v} and \texttt{Source/NoLeak.v}.\footnote{For
a total of 1.3k lines of specification and 14.3k lines of proof.}
}

%% differences in memory layout
%% simulated registers
%% administrative steps (return value from register, in right place later - copy argument, simulate Calling Convention of target language)
%% initialization: from static data
%% register correspondence

\subsection{Proof of Recomposition for \targetlanguage}
\label{sec:recomb}
% !TEX root = ./paper.tex

We use the turn-taking memory relation from
\autoref{sec:key-ideas-turn-taking} to prove 
recomposition (\Cref{lemma-robustptrs-recomp}).
%The trace relation $\sim_\mi{recomp}$
%(an instance of the relation $\sim$ in \Cref{def-trace-rel}) \dg{How are is no $\sim$ in \Cref{def-trace-rel} and $\sim_\mi{recomp}$ different?}
%renames addresses of
%the shared memories that appear in the two related traces.
%
To do that, we prove that 
\Cref{def-mem-rel-tt} of $\mathtt{mem\_rel\_tt}$ 
is an invariant. 
%\dg{The previous sentence didn't make sense to me: $\sim$ is used in the definition of The proof is part of $\mathtt{mem\_rel\_tt}$, so how can proving $\mathtt{mem\_rel\_tt}$  establish $\sim_\mi{recomp}$? Maybe ``this relation'' in the previous sentence refers to something other than $\sim_\mi{recomp}$?}
%\aek{Yeah, well, technically
% $\mathtt{state\_rel\_tt}$ establishes 
%$\sim_\mi{recomp}$. That one shouldn't be
%surprising: $\mathtt{state\_rel\_tt}$ is stronger than
%$\sim_\mi{recomp}$.
%}
\Cref{def-mem-rel-tt} is part of
a bigger invariant $\mathtt{state\_rel\_tt}$ 
on execution states
that we elide here for space reasons. The Coq proof of
\Cref{lemma-robustptrs-recomp} is, however, available in 
\texttt{Intermediate/RecompositionRel.v}, which in turn uses
all of {\small \texttt{RecompositionRelCommon.v},
	\texttt{RecompositionRelOptionSim.v},
	\texttt{RecompositionRelLockStepSim.v}} and
{\small \texttt{RecompositionRelStrengthening.v}}\footnote{For
a total of 830 lines of specification and 
12.6k lines of proof.}
% \rb{Should we have file references here (and above, for back-translation), or rather not?}} CH: fine
%\rb{Present these stats and those of back-translation in the same style---footnote?}

As explained at the end of \Cref{sec:key-ideas-turn-taking}, a key
requirement of the recomposition proof is a strengthening lemma that
recovers a stronger invariant, $\mathtt{state\_rel\_border}$, which
holds at states that emit interaction events.  We show the memory part
of $\mathtt{state\_rel\_border}$:
\dg{Can we already move from $\mathtt{mem\_}$ to $\mathtt{state\_}$ here? The change at a later point is very confusing.}
\aek{Better now?}
%% that is stronger than $\mathtt{mem\_rel\_tt}$, namely:

\begin{definition}[Memory Relation At Interaction 
	Events]
	\label{def-mem-rel-border}
\begin{align*}
&\mathtt{mem\_rel\_border}(\trg{s_{1,2}}, \trg{s_1}, \trg{s_2}, t_{1,2}, t_1, t_2) \defeq\\
&\ \ \ \ \mathtt{mem\_rel\_exec}(\trg{P_1}, t_1, t_{1,2}, \trg{s_1.\mi{Mem}}, \trg{s_{1,2}.\mi{Mem}})\ 
\wedge\ \\
&\ \ \ \ \mathtt{mem\_rel\_exec}(\trg{C_2}, t_2, t_{1,2}, \trg{s_2.\mi{Mem}}, \trg{s_{1,2}.\mi{Mem}})\ 
\end{align*}
where $\mathtt{mem\_rel\_exec}$ is exactly as in \Cref{def-mem-rel-tt}.
\end{definition}

Among other things, $\mathtt{mem\_rel\_border}$ ensures 
that the shared memories
of the three states (of the recomposed program and the two
base programs) are all in sync.
We are able to instantiate this strong invariant \emph{only}
at interaction events, because at these points
we can use the assumption
that the traces of the two base programs
are related (last assumption of \Cref{lemma-strengthening}), 
which implies that the
shared memories of the base programs are related.
This assumption can be combined with 
$\mathtt{mem\_rel\_tt}$ (which holds universally for every
%% non-observable or observable
triple of corresponding states)
to obtain $\mathtt{mem\_rel\_border}$. 
%\Cref{lemma-strengthening} 
%is the strengthening lemma whose proof relies on the combining mentioned.
%The last one of the assumptions of 
%\Cref{lemma-strengthening} is key because it 
%is the missing piece from $\mathtt{mem\_rel\_tt}$
%that helps us establish $\mathtt{mem\_rel\_border}$.\ch{This paragraph could be streamlined
%  (e.g. to repeat less)}
\begin{lemma}[Strengthening at interaction events]
\label{lemma-strengthening}
\begin{align*}
&\forall \trg{s_{1,2}}\ \trg{s_1}\ \trg{s_2}\ t_{1,2}\ t_1\ t_2\ 
\trg{s_1'}\ \trg{s_2'}\ e_1\ e_2.\\
&\ \ \mathtt{state\_rel\_tt}(\trg{s_{1,2}},\ \trg{s_1},\ \trg{s_2},\ t_{1,2},
\ t_1,\ t_2) \implies\\
&\ \ \trg{s_1\ \step{\bl{[e_1]}}\ s_1'} \implies\\
&\ \ \trg{s_2\ \step{\bl{[e_2]}}\ s_2'} \implies\\
&\ \ t_1\ \texttt{++}\ [e_1]\  \sim\  t_2\ \texttt{++}\ [e_2] 
\implies\\
&\exists \trg{s_{1,2}'}\ e_{1,2}.\ \trg{s_{1,2}\ \step{\bl{[e_{1,2}]}}\ s_{1,2}'}\ 
\wedge\ \\
&\ \ \ \mathtt{state\_rel\_border}(\trg{s_{1,2}'},\ \trg{s_1'},\ \trg{s_2'},\ \\
&\ \ \ \ \ \ \ \ \ \ \ \ \ \ 
t_{1,2}\ \texttt{++}\ [e_{1,2}],\ t_1\ \texttt{++}\ [e_1],\ t_2\ \texttt{++}\ [e_2])
\end{align*}
\end{lemma}
The relation $\mathtt{state\_rel\_tt}$ is a turn-taking
simulation invariant. It ensures that the memory
relation $\mathtt{mem\_rel\_tt}$ holds of the memories of the
three related states. Similarly, the stronger state relation
$\mathtt{state\_rel\_border}$ ensures that the memory relation
$\mathtt{mem\_rel\_border}$ holds of the memories of the three
related states.

The exact definition of the relation $\mathtt{state\_rel\_tt}$ 
is in \texttt{RecompositionRelCommon.v}. We show here two key
lemmas:
\begin{lemma}[Option simulation w.r.t. \emph{non}-executing \trg{part}]
	\label{lemma-option-sim}
\begin{align*}
&\forall \trg{s_{1,2}}\ \trg{s_1}\ \trg{s_2}\ 
t_{1,2}\ t_1\ t_2\ 
\trg{s_1'}.\\
&\ \ \trg{s_{1,2}}\ \texttt{is executing in}\ \trg{C_2}\ (\IE \emph{not in}\ \trg{P_1}) \implies\\ 
&\ \ \mathtt{state\_rel\_tt}(\trg{s_{1,2}},\ \trg{s_1},\ \trg{s_2},\ t_{1,2},
\ t_1,\ t_2) \implies\\
&\ \ \trg{s_1\ \step{\bl{[]}}^{*}\ s_1'} \implies\\
&\ \ \mathtt{state\_rel\_tt}(\trg{s_{1,2}},\ \trg{s_1'},\ \trg{s_2},\ t_{1,2},
\ t_1,\ t_2)
\end{align*}
\end{lemma}
The last assumption (\trg{s_1\ \step{\bl{[]}}^{*}\ s_1'})
 of the option simulation
(\cref{lemma-option-sim}) 
says that state \trg{s_1} of the base
program \trg{P_1 \cup C_1} takes some 
non-interaction steps. This base program contributes just \trg{P_1} to the
recomposed program (\trg{P_1 \cup C_2}), and we know by
assumption ``\trg{s_{1,2}}\ \texttt{is executing in}\ \trg{C_2}'' 
that the recomposed state \trg{s_{1,2}} is \emph{not} executing in 
\trg{P_1}. The invariant $\mathtt{state\_rel\_tt}$ ensures
that \trg{s_{1,2}} executes in \trg{P_1}
whenever \trg{s_1} executes in \trg{P_1}.
Thus, the steps that \trg{s_1} has made must be
taken by the discarded part \trg{C_1}, 
not the retained part \trg{P_1}. 
As shown in \Cref{example-c-changes-shared-memory-internally},
we know that steps taken by \trg{C_1} can
cause a mismatch between the memory of the recomposed program
and the memory of the base program \trg{P_1 \cup C_1}. 
The option
simulation lemma ensures that this mismatch is tolerated by
the $\mathtt{state\_rel\_tt}$ invariant.

\begin{lemma}[Lock-step simulation w.r.t. executing \trg{part}]
	\label{lemma-lockstep-sim}
\begin{align*}
&\forall \trg{s_{1,2}}\ \trg{s_1}\ \trg{s_2}\ 
t_{1,2}\ t_1\ t_2\ 
\trg{s_1'}.\\
&\ \ \trg{s_{1,2}}\ \texttt{is executing in}\ \trg{C_2}\ (\IE \emph{not in}\ \trg{P_1}) \implies\\ 
&\ \ \mathtt{state\_rel\_tt}(\trg{s_{1,2}},\ \trg{s_1},\ \trg{s_2},\ t_{1,2},
\ t_1,\ t_2) \implies\\
&\ \ \trg{s_2\ \step{\bl{[]}}\ s_2'} \implies\\
&\exists \trg{s_{1,2}'}.\ \trg{s_{1,2}\ \step{\bl{[]}}\ s_{1,2}'}\ \wedge\ \\
&\ \ \ \ \ \ 
\mathtt{state\_rel\_tt}(\trg{s_{1,2}'},\ \trg{s_1},\ \trg{s_2'},\ t_{1,2},
\ t_1,\ t_2)
\end{align*}
\end{lemma}

Lock-step simulation (\Cref{lemma-lockstep-sim})
ensures that the invariant $\mathtt{state\_rel\_tt}$ 
is strong enough to keep every non-interaction step
of a retained part in sync between the recomposed program 
and the corresponding base program.
%\ch{Our terminology
%  is very verbose, and this shows up in mouthful places like this.}
%\aek{Introduced ``retained'' vs ``discarded'' 
%program part.}

Although both \Cref{lemma-option-sim,lemma-lockstep-sim} hold
only for the scenario when 
``\trg{s_{1,2}}\ \texttt{is executing in}\ \trg{C_2} (\IE \emph{not in}\ \trg{P_1})'',
they are still general enough because
we can apply symmetry 
lemmas to our invariant $\mathtt{state\_rel\_tt}$ 
to reduce the other scenario
``\trg{s_{1,2}}\ \texttt{is executing in}\ \trg{P_1}'' to 
the former scenario---thus avoiding lots of duplicate proof.
The symmetry lemmas are proved
in \texttt{RecompositionRelCommon.v}.
\add{Here is the main symmetry lemma we use:
\begin{lemma}[Symmetry of $\mathtt{state\_rel\_tt}$]
	\label{lemma-state-relation-symmetry}
\begin{align*}
&\forall \trg{s_{1,2}}\ \trg{s_1}\ \trg{s_2}\ t_{1,2}\ t_1\ t_2.\\
&\ \ \mathtt{state\_rel\_tt}(\trg{s_{1,2}},\ \trg{s_1},\ \trg{s_2},\ t_{1,2},
\ t_1,\ t_2) \implies\\
&\ \ \mathtt{state\_rel\_tt}(\trg{s_{1,2}},\ \trg{s_2},\ \trg{s_1},\ t_{1,2},
\ t_2,\ t_1)
\end{align*}
\end{lemma}
}
\add{Intuitively, 
the two situations that \Cref{lemma-state-relation-symmetry}
asserts as symmetric are those where the
\trg{main}
function
of the recomposed program \trg{P_1 \cup C_2}
is (a) implemented by \trg{P_1} and (b) by \trg{C_2}.
The symmetry lemmas consequently allow us to
apply our simulation lemmas to \emph{both} of these
cases even though these simulation
lemmas are proved for just one of the cases.
}

In \texttt{RecompositionRel.v},
the reader can find the top-level proof of recomposition
(\Cref{lemma-robustptrs-recomp})
that uses these symmetry lemmas in addition to
strengthening (\Cref{lemma-strengthening}),
option simulation (\Cref{lemma-option-sim}), 
and lock-step simulation (\Cref{lemma-lockstep-sim}).

To summarize, the new idea of turn-taking simulations 
helped us
complete the recomposition proof with memory sharing, which is fully mechanized in Coq.

%% \aek{Is there space to show how the lemmas are put together into the main recomp theorem?}

%% \ch{What's more important is to quickly reiterate what's cool about this proof:
%%   it's mechanized and it deals with memory sharing. It's important since the reader
%%   might be lost in details by now.}

\subsection{Axioms}
\label{sec:axioms}
% !TEX root = ./paper.tex

\add{Our three novel proof steps 
(data-flow back-translation, 
\Cref{lemma-robustptrs-backtrans}, recomposition
using the turn-taking simulation, 
\Cref{lemma-robustptrs-recomp}, and
enrichment, \Cref{lemma-enrichment}) are fully mechanized and rely only
on standard logical axioms: proof irrelevance, functional
extensionality, and classical excluded middle.

The proof of \Cref{theorem-rsp}
%, however, %% make a bit less of a counterpoint; this part is not problematic
relies
(in addition to the fully mechanized lemmas above) 
intuitively on \cref{assm-bwdsim,assm-fwdsim}
about (separate) compilation of
\emph{whole programs} from \sourcelanguage to
\targetlanguage. These kind of assumptions are axiomatized in our Coq development
in a similar way to that of \citet{Abate:2018:GCG:3243734.3243745}.
In detail:
(a) we have four axioms stating that the result of compilation is
syntactically well-formed if the source
program is well-formed;
(b) one separate compilation axiom stating that compilation and linking commute;\ifsooner\ch{
  As far as I remember the way we state this in Coq is too strong: our proof
  doesn't require equality of the commutation, just that the behavior is the same}\fi{}
(c) two axioms stating the existence of a forward and a backward simulation
for whole-program compilation;
% (these are \cref{assm-bwdsim,assm-fwdsim} that appear in
% \cref{fig:rsc-proof});
and (d) one axiom ensuring that our compiler 
preserves the privacy of the local buffer. 
We expect this last axiom to hold because 
our compiler pass does not merge
memory blocks.\footnote{If the compiler did merge blocks, then satisfying the
axiom would require ensuring that it never merges a private block with a shared one.}
%are renamed consistently by the compiler (our compiler pass applies
%only the identity renaming, which is trivially consistent).\ch{This explanation
%  is a bit confusing: why are we talking about a renaming at all, if it's the
%  identity renaming? I assume there is lower level stuff that makes this have
%  more sense, but as explained here it seems funny. To be honest, even if it
%  was not the (funny) identity renaming, I still don't know what's the connection
%  between privacy of the local buffer and renaming of block ids.}
To prove it, we expect one can use
fine-grained simulation invariants very similar
to the ones one would use for a compiler correctness proof.
The precise statements of these axioms are given in our {\tt README.md} file.

One key motivation for building on the strategy
of \citet{Abate:2018:GCG:3243734.3243745}
(\autoref{sec:background-proof-strategy})
is to benefit from 
separation of concerns between secure compilation
concern and whole-program compiler correctness concern).
\emph{Axiomatizing}
whole-program compiler correctness,
however, is only
reasonable for the purposes of a 
methodology-oriented case study like ours, but is
not reasonable if the goal were to
provide \rsp-style \emph{assurance} for a real system.
In that case though, our methodology will enable
\emph{reuse} of the compiler correctness proof.
}

\section{Related Work}
\label{sec:related-work}
% !TEX root = ./paper.tex
\paragraph*{Memory relations similar to turn-taking simulations ($\mathtt{mem\_rel\_tt}$)}
\citet{capableptrs} and \citet{10.1145/2676726.2676985} use memory
relations that are similar to $\mathtt{mem\_rel\_tt}$ in that the
shared memories of two related executions may mismatch and the memory
relation guarantees that the context does not modify the
\emph{private} memory of the compiled program. However, there are
notable differences. First, their relations
are \emph{binary}---between two runs that differ in one
component---unlike ours, which is \emph{ternary}. This allows their
relations to be strengthened whenever the compiled program is
executing, while our relation can be strengthened
(\Cref{def-mem-rel-border}) only for single steps right after
interaction events. Second, the applications are quite
different. \citet{10.1145/2676726.2676985}'s relation is used in a
non-security proof about compositional compiler correctness where
guarantees come from assumptions about the target context
(\textbf{Setting 1}, \Cref{sec:background-example}), while our
guarantees come from memory protection features of \targetlanguage
(\textbf{Setting 2}). \citet{capableptrs}'s memory relation is used to
establish a different security criterion, full
abstraction~\cite{abadiprotectiontranslations,Patrignani:2019:FAS:3303862.3280984}.

\paragraph*{Reuse of standard compiler correctness for secure compilation}
We are aware of only two works that reuse compiler correctness
lemmas in a secure compilation
proof. \citet{Abate:2018:GCG:3243734.3243745}, which we directly build
on, have goals similar to ours, but without memory sharing, which is
really the focus of our paper. \citet{capableptrs} support memory
sharing and proof reuse using a different proof technique they
call \tricl. As explained in the paragraph on memory relations above,
their memory relation (which is part of \tricl) is technically very
different from our turn-taking simulations.
%\ch{Where above? Did this move to appendix?
%  In any case give precise reference.}
Additionally, unlike our technique, their proof is not mechanically
verified and, as explained in \Cref{sec:key-ideas-data-flow-traces},
mechanizing their proof is very difficult due to their use of complex
bookkeeping.

\paragraph*{Other kinds of informative traces}
Using inspiration from fully abstract trace
semantics~\cite{9249fdc52dd3414e83b2e3f9a89cb117},
\citet{10.1145/3436809} perform back-translation (with shared memory)
for a compiler pass using traces that record the \emph{whole} 
memory but still only emit it at
just interaction events. 
Although more informative than traces that record
only shared memory at interaction 
events~\cite{9249fdc52dd3414e83b2e3f9a89cb117,capableptrs}, 
these traces
still do not eliminate the need for bookkeeping, unlike
our data-flow traces that selectively expose 
\emph{non-interaction events} to
simplify back-translation.
% by allowing a simpler lock-step simulation.

\paragraph*{Handling memory sharing as message passing}
\citet{patrignani2015secure} describe a completely different secure
compilation of \add{object-oriented}
programs with memory sharing: Their compiled code
implements shared memory in a trusted third party (realized as a
hardware-protected module), and all reads and writes become explicit
RPCs to this third party. Under the hood, the third party relies on
dynamic sealing to hide memory addresses~\cite{protectioninpl}. This
effectively reduces memory sharing to message passing and elides most
of the complications in proofs with true memory sharing, but also
results in extremely inefficient code that requires heavyweight calls
at every read/write to shared memory, thus largely defeating the
purpose of sharing memory in the first place.

\add{It would be interesting to study whether 
enforcing encapsulation while also allowing more
direct memory sharing is feasible, and if so, 
whether the same challenges we faced still arise, 
and hence whether our proof techniques still apply.}

\paragraph*{Secure linking}
To ensure safe interaction with low-level code, typed assembly
languages\iffull~\cite{type-safe-linking}\fi\
and multi-language
semantics\iffull~\cite{multi-language-semantics}\fi\  
have been used
by \citet{FunTAL}. Their technique restricts the low-level language
not with runtime enforcement of memory isolation like in some
architectures~\cite{pma,pump,cheriageofrisk,cheriisatech,micro16,chericompartment}
and in our \targetlanguage model, but instead with a static type
system. The type system and the accompanying logical relation allow
reasoning about the equivalence of a ``mixed-language'' setting, which
is similar to our \textbf{Setting 2} but sometimes requires exposing
low-level abstractions to high-level code. The secure compilation
approach we follow has a chance of avoiding that. For example, by
avoiding the need for directly reasoning in \textbf{Setting 2}, our
secure compilation result beneficially hides from the programmer the
fact that a \targetlanguage function can jump to non-entry points of
other functions in the same component.\ch{Shouldn't this point be made
earlier?  It used to be one of our motivating examples, but even
without that, can we mention it when introducing Interim, for
instance?}

\paragraph*{Robust safety preservation}
Robust safety preservation (RSP), the secure compilation criterion we
use, was first described by
\citet{10.1145/3436809,abate2019journey} and \citet{Abate:2018:GCG:3243734.3243745}.
However, this initial work uses a trivial relation (equality) between
source and target traces. With a general relation, as in our setting,
RSP was first examined by~\citet{DBLP:conf/esop/AbateBCD0HPTT20}. RSP
further traces lineage to the robust verification of safety properties
of a given program (not a given compiler), which is often called
``robust satisfaction of safety
properties''~\cite{10.1007/3-540-48320-9_27}. Robust satisfaction is a
well-developed concept, used in model
checking~\cite{grumberg1991model}, type
systems~\cite{gordon2003authenticity}, and program
logics~\cite{swasey2017robust,10.1145/3371100,7243751}.

\paragraph{Secure compilation of information-flow-like properties}

A long line of work
\cite{10.1145/3371075,barthe2018secure,sison_et_al:LIPIcs:2019:11082,besson2018securing,Namjoshi-model-checking-ct-compilation,constant-resource-secure-compilation,sison_murray_2021,10.1007/978-3-540-24622-0_2}
develops proof techniques and verified compilers to ensure that
information flow properties like non-interference,
the constant-time policy,
or side-channel resistance are preserved
by compilation.  These techniques, however,
are all concerned with whole programs, unlike our work, which starts
with the premise that partial programs will interact with untrusted code. 

%\ch{Can we generalize this to preserving information flow properties?
%  One of the PC members has work in that space~\cite{sison-itp19}.
%  Another work is this~\cite{BessonDJ18}}
%\aek{It turns out we do cite both of these
%	already! Also, please use keys
%	only from
%	paper.bib. There were many duplicate citations. :-)}

\section{Scope and Limitations}
\label{sec:limitations}
% !TEX root = ./paper.tex

\add{We emphasize again that the key benefit of our
data-flow back-translation lies in simplifying secure compilation
proofs when memory is shared via pointers and
the source and target languages are syntactically
dissimilar. However, our technique is useful 
only if the target language
has fine-grained memory protection, which is available in some
capability architectures like CHERI~\cite{cheriageofrisk} or tagged
architectures~\cite{micropolicies}, but not in mainstream architectures
such as x86. Nonetheless, this is not a limitation of our
{\em proof} technique, but rather a fundamental {\em enforcement} problem.
Even leaving aside the proof,
we believe it is not known how to efficiently compile a
memory-safe source language with fine-grained memory sharing to an
architecture without support for fine-grained memory protection in a
way that maintains security against arbitrary target-level contexts.
}

While the presentation in the paper kept the renaming relation abstract,
our RSP theorem in Coq is stated only for a concrete subclass of renaming
relations, which was sufficient for our particular back-translation
function and our particular compiler pass.
Our compiler satisfies compiler correctness for
the identity renaming, meaning that it does not rename pointers.
To simplify our formalization, we exploit this fact by only considering renamings
that are constant shifts.
We leave for future work the generalization of this subclass. Such a
generalization would be needed for applying our proof technique to
o a more interesting compiler
that needs a more complex renaming relation~\cite{Leroy2009, 10.1007/s10817-008-9099-0}.
For instance, instead of storing the whole stack in a single block, the compiler
could implement the stack by allocating a new block for each stack frame.
In this case the renaming relation needed for compiler correctness would relate
blocks in a more subtle way than the simple identity or increment-by-1 relation.
While the proof of back-translation would not be affected by such a change, with
a generalized renaming relation, we will have to think more explicitly about the
properties needed for the recomposition proof and the top-level proof to go through.
We expect that consistency of the renaming is one such property, but there may
be other properties on which we relied implicitly for our special subclass.

\section{Future Work}
\label{sec:conclusion}

In the future, we would like to apply our
proof techniques to more realistic compilers and also to lower-level
compiler passes that implement enforcement mechanisms, for instance based on capability
machines~\cite{cheriageofrisk,chericompartment} or programmable tagged
architectures~\cite{pump,micropolicies}.
We also think our techniques can be extended to stronger secure compilation
criteria, building on work by \citet{abate2019journey}, who illustrate that the
robust preservation of a large class of {\em relational} safety properties can
be proved by trace-directed back-translation.
This is stronger than both \rsp and a full abstraction variant, but their
back-translation technique does not yet cope with mutable state.

\add{The languages we studied
are both dynamically typed. It would be interesting
to study how our proof techniques apply to
secure compilers from
\emph{statically-typed} source languages too.

Another line for extending our work
is to study a more realistic calling convention
involving a single stack for data and control
(our target language uses just a control 
stack and passes arguments only in registers). 
We expect data-flow back-translation to still be applicable,
but to build such a secure compiler
one will need to specify the interface of the low-level
context and dynamically enforce
that the low-level context's
use of the stack adheres to its interface.%

Finally, allowing undefined behavior, as done by
\citet{Abate:2018:GCG:3243734.3243745}, should be compatible with our techniques,
as long as the {\em cross-component} memory operations of the source language
are compiled safely, not left completely undefined---\EG out-of-bounds accesses to
pointers shared by other components.
}

% \ch{Future work:

% - In this work we only do RSP, but we think with some work on
% back-translation, our proof technique can be strengthened to RFrXP, a criterion
% stronger than both RSP and a variant of full abstraction~\cite{abate2019journey}.
% %
% \citet{abate2019journey} has already shown proofs of RFrXP, but they are for
% languages with no state at all. Jeremy's unpublished work is extending that to
% languages with state, but also doing memory sharing will mean merging two
% difficult proofs :)}

%\section{Background on Hardware Capabilities}
%\label{sec:capabilities-intro}
%\input{capabilities}

%\section{\capableptrs Related work}\label{sec:related-work}
%\input{related_work}

%% Acknowledgments
\ifanon\else\ifappendix
% \begin{acks}
{\small
\paragraph{Acknowledgements}
% We are grateful \ldots
%
\add{
We thank the reviewers for their careful reviews and helpful comments.
This work was in part supported
by the
\ifieee
\href{https://erc.europa.eu}{European Research Council}
\else
\grantsponsor{1}{European Research Council}{https://erc.europa.eu/}
\fi
under ERC Starting Grant SECOMP (\ifieee 715753\else\grantnum{1}{715753}\fi),
and by the Deutsche Forschungsgemeinschaft (DFG\iffull, German Research Foundation\fi)
as part of the Excellence Strategy of the German Federal and State Governments
-- EXC 2092 CASA - 390781972.
\ch{TODO: add more acknowledgements here}
}
}
% \end{acks}
\fi\fi

\ifanon\clearpage\fi

\ifappendix
\onecolumn
\appendices

\setcounter{theorem}{0}
\renewcommand{\thetheorem}{\thesection.\arabic{theorem}}

% !TEX root = ./paper.tex

\onecolumn

In the following appendices, we provide more detailed high-level overviews of notions discussed in the paper.
We first describe how one can define \textbf{nowrite} as a predicate on traces (\autoref{app:safety-example}).
We then give a high-level overview of the memory model used by our languages (\autoref{app:sec:memory-model}),
as well as the syntax and semantics of the two languages \sourcelanguage (\autoref{app:sec:source}) and \targetlanguage (\autoref{app:sec:target}).
Finally, in \autoref{app:dataflow-backtrans-example}, we describe how each data-flow event is back-translated
into a source expression.
For more details about our definitions and proofs we refer the reader to the Coq development.

\section{Expressing \textbf{nowrite} using traces}
\label{app:safety-example}
% !TEX root = ./paper.tex

The safety property \textbf{nowrite} can be defined formally as a predicate on traces
(\Cref{def-mem-sharing-trace-event}) as follows.

\begin{example}[The safety spec \textbf{nowrite}]
\label{example-NO-WRITE-spec} 
Suppose $l_\mi{balance}$ is the memory location allocated
for the variable \src{user\_balance\_usd}, and suppose
$c_\mi{main}$ denotes the component that calls the function 
\src{set\_ads\_image}, which is implemented by $c_\mi{setter}$.
\begin{align*}
&\textbf{nowrite}(t) \defeq\\
&\forall t_1\ e_\mi{call}\ t_2\ \mi{Mem_1}\ 
\mi{Mem_2}\ l.\\ 
&t = t_1\ \texttt{++}\ [e_\mi{call}]\ \texttt{++}\ t_2 \implies\\
&e_\mi{call} = \mathtt{Call}~\mathit{Mem_1~c_{main}~c_{setter}.set\_ads\_image()} \implies\\
&\mathtt{find\_matching\_ret}(t_2, e_\mi{call}) =\\ 
&\ \ \ \ \ \ \  \mathtt{Ret}~\mathit{Mem_2~c_{setter}~c_{main}.\mi{void}} \implies\\
&\mi{Mem_1}(l_\mi{balance}) = \mi{Mem_2}(l_\mi{balance})
\end{align*}
\end{example}
The spec above makes sure that the value of the variable
\src{user\_balance\_usd} before the call to the function
\src{set\_ads\_image} is the same as its value at the point
when the function returns.

\section{Memory model}\label{app:sec:memory-model}

\sourcelanguage and \targetlanguage use the same compartmentalized, block-based memory model
in the style of the CompCert's memory model.
Memory is subdivided into a data and a code section; each component has its own section in
memory, which consists in several blocks.

Memory is accessed using \emph{pointers} pointing to locations in memory. A pointer is
a 4-tuple \(\com{ptr} = (\ii{perm}, c, b, o)\) where \(\ii{perm}\) might be \(\mathrm{data}\)
or \(\mathrm{code}\), \(c\) is a component's identifier, \(b\) is a block identifier, and \(o\)
is a positive offset.

One can perform arithmetic operations on these pointers (for instance, incrementing an offset),
but the pointers are safe: one cannot use a pointer \((\ii{perm}, c, b, o)\) to access a location
\((\ii{perm}, c', b', o')\) with \(c' \neq c\) or \(b' \neq b\), for instance by abusing
a buffer overflow bug.

Memory cannot be accessed using pointers with permission \(\mathrm{code}\). These pointers are
used to represent either function pointers (in the source) or the program counter (PC) in the
target language.

To access and update a memory \(\ii{mem}\)three functions are provided:
\begin{itemize}
\item \(\ii{mem}[\ii{ptr}]\) returns the content of the memory \(\ii{mem}\) at the location
      pointed to by \(\ii{ptr}\), if and only if \(\ii{ptr}\)'s permission is
      \(\mathrm{data}\), and the location is already allocated; it is not defined otherwise;
\item \(\ii{mem}[\ii{ptr} \mapsto v]\) updates the content of the memory \(\ii{mem}\) at the
      location pointed to by \(\ii{ptr}\) to value \(v\), and returns the updated memory;
      it does not perform allocation, so it is not defined if
      \(\ii{ptr}\) points to a location not yet allocated (in our languages, it is impossible to
      forge such a pointer);
\item \(\ii{mem\_alloc}(\ii{mem}, c, \ii{size})\) allocates a new block of size \(\ii{size}\)
      in \(c\)'s data section, and returns a pair \((\ii{mem'}, \ii{ptr})\) where
      \(\ii{mem'}\) is the updated memory and \(\ii{ptr}\) points to the newly allocated block.
\end{itemize}

In the Coq development, the memory model can be found in file \texttt{Common/Memory.v}.

\afterpage{\FloatBarrier}
\section{Description of the source language \sourcelanguage}\label{app:sec:source}

In this appendix, we give a high-level overview of the syntax and semantics of
\sourcelanguage, our source language. The source language is implemented in folder
\texttt{Source/} in the Coq development.

A source program \(\src{P} = (\src{intf}, \src{procs},\src{buffers})\) is a triplet where
\begin{itemize}
\item \src{intf} is the program's interface;
\item \src{procs} is a partial map from compartment identifiers and procedure identifiers to
procedure code (an expression);
\item \src{buffers} is a partial map from compartment identifiers to initial static buffers, \IE a
list of values the compartment's data is initialized to.
\end{itemize}

The syntax of the source expressions is reproduced in \autoref{fig:app-source-syntax}.

\begin{figure}[h]
%% \begin{small}
  \centering
  \begin{tabular}{lll}
  \texttt{\src{exp}}& \texttt{::=} \quad\src{v}               & values            \\
                 & | \src{ arg}                & function argument \\
                 & | \src{ local}              & local static buffer \\
                 & | \src{ exp_1\ \otimes\ exp_2} & binary operations \\
                 & | \src{ exp_1;\ exp_2}       & sequence          \\
                 & | \src{ if\ exp_1\ then\ exp_2\ else\ exp_3} & conditional \\
                 & | \src{ alloc\ exp}            & memory allocation \\
                 & | \src{ !exp}                 & dereferencing     \\
                 & | \src{ exp_1 := exp_2}     & assignment        \\
                 & | \src{ c.func(exp)}             & function call    \\
                 & | \src{ *[exp_1](exp_2)}    & call pointer      \\
                 & | \src{ \&func}              & function pointer  \\
                 & | \src{ exit}               & terminate
\end{tabular}
%% \end{small}
%% \vspace{-0.8em}
\caption{Syntax of source language expressions
  %\cite{AbateABEFHLPST18}
}
\label{fig:app-source-syntax}
%% \vspace{-1em}
\end{figure}

\begin{figure}
\centering
%% \begin{small}
  \begin{tabular}{ll}
  \texttt{\src{k}}& \texttt{::=} \quad\src{Kstop}\\
                 & | \src{ Kbinop1\ \otimes\ exp\ k}\\
                 & | \src{ Kbinop2\ \otimes\ v\ k}\\
                 & | \src{ Kseq\ exp\ k}\\
                 & | \src{ Kif\ exp_1\ exp_2\ k}\\
                 & | \src{ Kalloc\ k}\\
                 & | \src{ Kderef\ }\\
                 & | \src{ Kassign1\ exp\ k}\\
                 & | \src{ Kassign2\ v\ k}\\
                 & | \src{ Kcall\ c\ func\ k}\\
                 & | \src{ Kcallptr1\ exp\ k}\\
                 & | \src{ Kcallptr2\ v\ k}
                  \end{tabular}
%% \end{small}
%% \vspace{-0.8em}
%% \vspace{5pt}
\caption{Syntax of source continuations
  %\cite{AbateABEFHLPST18}
}\label{fig:app-cont-syntax}
%% \vspace{-1em}
\end{figure}

The semantics is written in a \emph{continuation-passing style}. The syntax of continuations
\(\src{k}\) is given in \autoref{fig:app-cont-syntax}.
The semantics of the source language is given as a small-step relation
\[
\src{G \vdash s \xrightarrow{t} s'}
\]
read ``under global environment \src{G}, state \src{s} reduces to state \src{s'} emitting
interaction event \src{t}''.

The global environment \src{G} contains information necessary to execute the program, such
as the code of each procedure, and the interface information.

The states \src{s} and \src{s'} are 6-tuple \(\src{(\ii{c}, \sigma, \ii{mem}, k, e, \ii{arg})}\) where:
\begin{itemize}
\item \src{\ii{c}} is the current compartment's identifier;
\item \src{\sigma} is the protected stack, that is, a list of frame that grows
with (both cross-compartment and intra-compartment) calls and shrinks with (both cross-compartment and intra-compartment) returns; frames contain the continuation and the argument to restore after returns;
\item \src{\ii{mem}} is the memory, a partial map from pointers to values;
\item \src{k} is the current continuation;
\item \src{e} is the current runtime expression;
\item \src{\ii{arg}} is the last call's argument.
\end{itemize}

Standard reduction rules can be found in \autoref{fig:app-source-red}.
Reduction rules for memory operations can be found in \autoref{fig:app-source-red-mem}.
Reduction rules for calls and returns can be found in \autoref{fig:app-source-red-call-ret}.

In particular, note that the memory operations do not impose any conditions on the component
of the pointers they use: this reflects the fact that \sourcelanguage allows sharing and using
pointers to other components.

\begin{figure}
\begin{mathpar}
\inferrule[KS\_Binop1]{\
}{
  \src{G \vdash (\ii{c}, \sigma, \ii{mem}, k, e_1 \otimes e_2, \ii{arg}) \xrightarrow{\epsilon}
                (\ii{c}, \sigma, \ii{mem}, Kbinop1 \otimes e_2\ k, e_1, \ii{arg})
  }
}\\

\inferrule[KS\_Binop2]{\
}{
  \src{G \vdash (\ii{c}, \sigma, \ii{mem}, Kbinop1 \otimes e_2\ k, v, \ii{arg}) \xrightarrow{\epsilon}
                (\ii{c}, \sigma, \ii{mem}, Kbinop2 \otimes v\ k, e_2, \ii{arg})
  }
}\\

\inferrule[KS\_BinopEval]{\
}{
  \src{G \vdash (\ii{c}, \sigma, \ii{mem}, Kbinop2 \otimes v\ k, v', \ii{arg}) \xrightarrow{\epsilon}
                (\ii{c}, \sigma, \ii{mem}, k, v \otimes v', \ii{arg})
  }
}\\

\inferrule[KS\_Seq1]{\
}{
  \src{G \vdash (\ii{c}, \sigma, \ii{mem}, k, e_1; e_2, \ii{arg}) \xrightarrow{\epsilon}
                (\ii{c}, \sigma, \ii{mem}, Kseq\ e_2\ k, e_1, \ii{arg})
  }
}\\

\inferrule[KS\_Seq2]{\
}{
  \src{G \vdash (\ii{c}, \sigma, \ii{mem}, Kseq\ e_2\ k, v, \ii{arg}) \xrightarrow{\epsilon}
                (\ii{c}, \sigma, \ii{mem}, k, e_2, \ii{arg})
  }
}\\

\inferrule[KS\_If1]{\
}{
  \src{G \vdash (\ii{c}, \sigma, \ii{mem}, k, if\ e_1\ then\ e_2\ else\ e_3, \ii{arg}) \xrightarrow{\epsilon}
                (\ii{c}, \sigma, \ii{mem}, Kif\ e_2\ e_3\ k, e_1, \ii{arg})
  }
}\\

\inferrule[KS\_IfTrue]{\
}{
  \src{G \vdash (\ii{c}, \sigma, \ii{mem}, Kif\ e_2\ e_3\ k, true, \ii{arg}) \xrightarrow{\epsilon}
                (\ii{c}, \sigma, \ii{mem}, k, e_2, \ii{arg})
  }
}\\

\inferrule[KS\_IfFalse]{\
}{
  \src{G \vdash (\ii{c}, \sigma, \ii{mem}, Kif\ e_2\ e_3\ k, false, \ii{arg}) \xrightarrow{\epsilon}
                (\ii{c}, \sigma, \ii{mem}, k, e_2, \ii{arg})
  }
}\\

\inferrule[KS\_Arg]{\
}{
  \src{G \vdash (\ii{c}, \sigma, \ii{mem}, k, arg, \ii{arg}) \xrightarrow{\epsilon}
                (\ii{c}, \sigma, \ii{mem}, k, \ii{arg}, \ii{arg})
  }
}\\

\inferrule[KS\_Local]{\
}{
  \src{G \vdash (\ii{c}, \sigma, \ii{mem}, k, local, \ii{arg}) \xrightarrow{\epsilon}
                (\ii{c}, \sigma, \ii{mem}, k, Ptr (data, c, 0, 0), \ii{arg})
  }
}\\

\inferrule[KS\_FunPtr]{\
    \src{\ii{procedure\_id}(\ii{c},func)} = \src{b}
}{
  \src{G \vdash (\ii{c}, \sigma, \ii{mem}, k, \&func, \ii{arg}) \xrightarrow{\epsilon}
                (\ii{c}, \sigma, \ii{mem}, k, Ptr (code, c, b, 0), \ii{arg})
  }
}\\
\end{mathpar}
\caption{Small-step semantics of the \sourcelanguage}\label{fig:app-source-red}
\end{figure}

\begin{figure}
\begin{mathpar}
\inferrule[KS\_Alloc1]{\
}{
  \src{G \vdash (\ii{c}, \sigma, \ii{mem}, k, alloc\ e, \ii{arg}) \xrightarrow{\epsilon}
                (\ii{c}, \sigma, \ii{mem}, Kalloc\ k, e, \ii{arg})
  }
}\\

\inferrule[KS\_AllocEval]{\src{\ii{size}} > 0 \and
\ii{mem\_alloc}(\src{\ii{mem}}, \src{\ii{size}}) = (\src{\ii{mem'}}, \src{\ii{ptr}})
}{
  \src{G \vdash (\ii{c}, \sigma, \ii{mem}, Kalloc\ k, \ii{size}, \ii{arg}) \xrightarrow{\epsilon}
                (\ii{c}, \sigma, \ii{mem'}, k, \ii{ptr}, \ii{arg})
  }
}\\

\inferrule[KS\_Deref1]{\
}{
  \src{G \vdash (\ii{c}, \sigma, \ii{mem}, k, !e, \ii{arg}) \xrightarrow{\epsilon}
                (\ii{c}, \sigma, \ii{mem}, Kderef\ k, e, \ii{arg})
  }
}\\

\inferrule[KS\_DerefEval]{\src{v} = \src{\ii{mem}[ptr]}
}{
  \src{G \vdash (\ii{c}, \sigma, \ii{mem}, Kderef\ k, \ii{ptr}, \ii{arg}) \xrightarrow{\epsilon}
                (\ii{c}, \sigma, \ii{mem}, k, v, \ii{arg})
  }
}\\

\inferrule[KS\_Assign1]{\
}{
  \src{G \vdash (\ii{c}, \sigma, \ii{mem}, k, e_1 := e_2, \ii{arg}) \xrightarrow{\epsilon}
                (\ii{c}, \sigma, \ii{mem}, Kassign1\ e_1\ k, e_2, \ii{arg})
  }
}\\

\inferrule[KS\_Assign2]{\
}{
  \src{G \vdash (\ii{c}, \sigma, \ii{mem}, Kassign1\ e_1\ k, v, \ii{arg}) \xrightarrow{\epsilon}
                (\ii{c}, \sigma, \ii{mem}, Kassign2\ v\ k, e_1, \ii{arg})
  }
}\\

\inferrule[KS\_AssignEval]{\src{\ii{mem'}} = \src{\ii{mem}[\ii{ptr} \mapsto v]}
}{
  \src{G \vdash (\ii{c}, \sigma, \ii{mem}, Kassign2\ v\ k, \ii{ptr}, \ii{arg}) \xrightarrow{\epsilon}
                (\ii{c}, \sigma, \ii{mem'}, k, v, \ii{arg})
  }
}\\
\end{mathpar}
\caption{Small-step semantics of \sourcelanguage: memory operations}\label{fig:app-source-red-mem}
\end{figure}

\begin{figure}
\begin{mathpar}
\inferrule[KS\_InitCall]{
}{
  \src{G \vdash (\ii{c}, \sigma, \ii{mem}, k, c.func(e), \ii{arg}) \xrightarrow{\epsilon}
                (\ii{c}, \sigma, \ii{mem}, Kcall\ c\ func\ k, e, \ii{arg})
  }
}\\

\inferrule[KS\_InitCallPtr1]{
}{
  \src{G \vdash (\ii{c}, \sigma, \ii{mem}, k, *[e_1](e_2), \ii{arg}) \xrightarrow{\epsilon}
                (\ii{c}, \sigma, \ii{mem}, Kcallptr1\ e_1\ k, e_2, \ii{arg})
  }
}\\

\inferrule[KS\_InitCallPtr2]{
}{
  \src{G \vdash (\ii{c}, \sigma, \ii{mem}, Kcallptr1\ e_1\ k, v, \ii{arg}) \xrightarrow{\epsilon}
                (\ii{c}, \sigma, \ii{mem}, Kcallptr2\ v\ k, e_1, \ii{arg})
  }
}\\

\inferrule[KS\_InitCallPtr3]{
\src{\ii{procedure\_id}(\ii{c},func)} = \src{b}
}{
  \src{G \vdash (\ii{c}, \sigma, \ii{mem}, Kcallptr2\ v\ k, Ptr(code, c, b, 0), \ii{arg}) \xrightarrow{\epsilon}
                (\ii{c}, \sigma, \ii{mem}, Kcall\ c\ func\ k, v, \ii{arg})
  }
}\\

\inferrule[KS\_InternalCall]{
  \src{\ii{code\_of}(\ii{c}, func)} = \src{e}
}{
  \src{G \vdash (\ii{c}, \sigma, \ii{mem}, Kcall\ c\ func\ k, v, \ii{arg}) \xrightarrow{\epsilon}
                (\ii{c}, (c, \ii{arg}, k) :: \sigma, \ii{mem}, Kstop, e, v)}
}\\

\inferrule[KS\_ExternalCall]{
  \src{\ii{code\_of}(\ii{c}, func)} = \src{e} \and
  \src{\ii{c}} \neq \src{\ii{c'}} \and
  \src{c'.func} \in \src{c.\ii{import}}
}{
  \src{G \vdash (\ii{c}, \sigma, \ii{mem}, Kcall\ c'\ func\ k, v, \ii{arg})
  \xrightarrow{Call\ \ii{mem}\ \ii{c}\ \ii{c'}\ func\ v}
                (\ii{c}', (c, \ii{arg}, k) :: \sigma, \ii{mem}, Kstop, e, v)}
}\\

\inferrule[KS\_InternalReturn]{
}{
  \src{G \vdash (\ii{c}, (c, \ii{arg}, k) :: \sigma, \ii{mem}, Kstop, v, \ii{arg'}) \xrightarrow{\epsilon}
                (\ii{c}, \sigma, \ii{mem}, K, v, \ii{arg})}
}\\

\inferrule[KS\_ExternalReturn]{
  \src{\ii{c}} \neq \src{\ii{c'}}
}{
  \src{G \vdash (\ii{c}', (c, \ii{arg}, k) :: \sigma, \ii{mem}, Kstop, v, \ii{arg'})
   \xrightarrow{Ret\ \ii{mem}\ c'\ c\ v}
                (\ii{c}, \sigma, \ii{mem}, Kstop, v, \ii{arg})}
}
\end{mathpar}
\caption{Small-step semantics of \sourcelanguage: calls and returns}\label{fig:app-source-red-call-ret}
\end{figure}

\clearpage
\section{Description of the target language \targetlanguage}\label{app:sec:target}

In this appendix, we give a high-level overview of the target language \targetlanguage.
The target language is implemented in folder \texttt{Intermediate/} in the Coq development.

The target language \targetlanguage is an instruction-based language. The available instructions are reproduced
in \autoref{fig:app-target-syntax}.

A target program
\(\trg{P} = (\trg{intf}, \trg{procs}, \trg{buffers}, \trg{main})\) is a 4-tuple where
\begin{itemize}
\item \(\trg{intf}\) is the program's interface;
\item \(\trg{procs}%% : \mi{Component} \hookrightarrow \mi{Procedure} \hookrightarrow \trg{code}
\) is a partial map from compartment identifiers and procedure identifiers to code, \IE a list of instructions;
\item \(\trg{buffers}\) is a partial map from compartment identifiers to initial static buffers, \IE
a list of value the compartment's data is initialized to;
\item \(\trg{main}: \mathbb{B}\) is \(\mathtt{true}\) if the program contains the main procedure,
  \(\mathtt{false}\) otherwise.
\end{itemize}

\begin{figure}[h]
\centering
%% \begin{small}
\begin{tabular}{l@{\hskip0pt}l@{\hskip20pt}l}
  \trg{instr ::=} & ~\trg{ Const\ i\ \texttt{->}\  r}                   & | \trg{ Bnz\ r\ L}         \\
                     & | \trg{ Mov\ r_s\ \texttt{->}\  r_d}                 & | \trg{ Jump\  r}             \\
                     & | \trg{ BinOp\ r_1\ {\otimes}\ r_2\ \texttt{->}\ r_d}   & | \trg{ JumpFunPtr\  r}       \\
                     & | \trg{ Label\  L}                         &  | \trg{ Jal\ L}             \\
                     & | \trg{ PtrOfLabel\ L\ \texttt{->}\  r_d}           & | \trg{ Call\ c\ func}         \\
                     & | \trg{ Load\ *r_p\ \texttt{->}\  r_d}              & | \trg{ Return}               \\
                     & | \trg{ Store\ *r_p\ \texttt{<-}\  r_s}             & | \trg{ Nop}                  \\
                     & | \trg{ Alloc\ r_1\  r_2}                 & | \trg{ Halt}                 \\
  %% \trg{code ::= } & ~\trg{\overline{instr}} & \\
  \end{tabular}
%% \end{small}
%% \vspace{-0.8em}
\caption{Instructions of the target language
  %% \cite{AbateABEFHLPST18}
}
\label{fig:app-target-syntax}
\end{figure}

The semantics of the target language is given as a small-step relation
\[ \trg{E \vdash s \xrightarrow{\alpha} s'} \]
read ``under global environment \(\trg{E}\), state \(\trg{s}\) reduces to state \(\trg{s'}\)
emitting data-flow event \(\trg{\alpha}\).''

The global environment \(\trg{E}\) contains information necessary to execute the program,
such as the code of each procedure, label information, and interface information.

The states \(\trg{s}\) and \(\trg{s'}\) are 5-tuple \(\trg{(\sigma, \ii{mem}, \ii{reg}, \ii{pc})}\) where:
\begin{itemize}
\item \(\trg{\sigma}\) is the protected cross-compartment call stack, that is, a list of PCs that grows with cross-compartment
calls and shrinks with cross-compartments returns; \item \(\trg{\ii{mem}}\) is the memory, a partial map from pointers to values;
\item \(\trg{\ii{reg}}\) is the register file, a map from register names to values;
\item \(\trg{\ii{pc}}\) is the current PC, \IE\ a pointer to the current instruction.
\end{itemize}

Note that the stack doesn't contain any data. Instead, each compartment is
responsible with saving their own data in their private memory before giving control to
another compartment.\jt{Does this make sense?}% JT: does this make sense?

Reduction rules can be found in \autoref{fig:app-target-red}, \autoref{fig:app-target-red-mem},
and \autoref{fig:app-target-red-call-ret}, and are mostly standard.

We highlight the following particularities of the semantics:
\begin{itemize}
\item Cross-compartment control exchange is only possible via the two instructions \(\trg{Call\ c\ func}\) and \(\trg{Return}\)
  (rules \textsc{Call} and \textsc{Return}). Conversely, the jump and branching instructions do not allow a change in the
  current compartment (rules \textsc{Jal}, \textsc{Jump}, \textsc{BnzNZ}, \textsc{BnzZ}, \textsc{JumpFunPtr}).
\item Compared to the previous work this work is based on, rules \textsc{Load} and \textsc{Store} do not restrict access
  to memory. Any compartment can read and write to any location in memory, provided it has access to a pointer to this
  location.
\end{itemize}
\begin{figure}
\begin{mathpar}

\inferrule[Nop]{
  \trg{\ii{fetch}(E, \ii{pc})} = \trg{Nop}
}{
  \trg{E \vdash (\sigma, \ii{mem}, \ii{reg}, \ii{pc}) \xrightarrow{\epsilon}
      (\sigma, \ii{mem}, \ii{reg}, \ii{pc} + 1)
}}
\\

\inferrule[Label]{
  \trg{\ii{fetch}(E, \ii{pc})} = \trg{Label \ L}
}{
  \trg{E \vdash (\sigma, \ii{mem}, \ii{reg}, \ii{pc}) \xrightarrow{\epsilon}
      (\sigma, \ii{mem}, \ii{reg}, \ii{pc} + 1)
}}
\\

\inferrule[Const]{
  \trg{\ii{fetch}(E, \ii{pc})} = \trg{ Const\ i\ \texttt{->}\  r}\\
  \trg{v} = \trg{\ii{imm\_to\_val}(i)}\\
  \trg{\ii{reg'}} = \trg{\ii{reg}[r \mapsto i]}\\
  \trg{\alpha} = \trg{\mathtt{Const}~\ii{mem~reg'~\ii{comp}(pc)~v}~r}
}{
  \trg{E \vdash (\sigma, \ii{mem}, \ii{reg}, \ii{pc}) \xrightarrow{\alpha}
      (\sigma, \ii{mem}, \ii{reg'}, \ii{pc} + 1)
}}
\\

\inferrule[Mov]{
  \trg{\ii{fetch}(E, \ii{pc})} = \trg{Mov\ r_s\ \texttt{->}\  r_d} \\
  \trg{\ii{reg}[r_s]} = \trg{v}
  \trg{\ii{reg'}} = \trg{\ii{reg}[r_d \mapsto v]}\\
  \trg{\alpha} = \trg{\mathtt{Mov}~\mathit{mem~reg'~\ii{comp}(pc)~r_s~r_d}}
}{
  \trg{E \vdash (\sigma, \ii{mem}, \ii{reg}, \ii{pc}) \xrightarrow{\alpha}
      (\sigma, \ii{mem}, reg', \ii{pc} + 1)
}}
\\

\inferrule[Binop]{
  \trg{\ii{fetch}(E, \ii{pc})} = \trg{BinOp\ r_1\ {\otimes}\ r_2\ \texttt{->}\ r_d}\\
  \trg{v} = \trg{\ii{reg}[r_1] \otimes \ii{reg}[r_2]}\\
  \trg{\ii{reg'}} = \trg{\ii{reg}[r_d \mapsto v]}\\
  \trg{\alpha} = \trg{\mathtt{Binop}~\mathit{mem~reg~\ii{comp}(pc)~\otimes~r_1~r_2~r_d}}
}{
  \trg{E \vdash (\sigma, \ii{mem}, \ii{reg}, \ii{pc}) \xrightarrow{\alpha}
      (\sigma, \ii{mem}, \ii{reg}[r_d \mapsto v], \ii{pc} + 1)
}}
\\

\inferrule[Jal]{
  \trg{\ii{fetch}(E, \ii{pc})} = \trg{Jal\ L}\\
  \trg{\ii{find\_label}(E, \ii{pc}, L)} = \trg{\ii{pc'}}\\
  \trg{\ii{reg'}} = \trg{\ii{reg}[RA \mapsto v]}\\
  \trg{\alpha} = \trg{\mathtt{Const}~\ii{mem~reg'~\ii{comp}(pc)~v}~RA}
}{
  \trg{E \vdash (\sigma, \ii{mem}, \ii{reg}, \ii{pc}) \xrightarrow{\alpha}
      (\sigma, \ii{mem}, \ii{reg'}, \ii{pc'})
}}
\\

\inferrule[Jump]{
  \trg{\ii{fetch}(E, \ii{pc})} = \trg{Jump\ r}\\
  \trg{\ii{pc'}} = \trg{\ii{reg}[r]} \\
  \trg{\ii{is\_code\_pointer(pc')}} \\
  \trg{\ii{comp}(\ii{pc})} = \trg{\ii{comp}(\ii{pc'})}
}{
  \trg{E \vdash (\sigma, \ii{mem}, \ii{reg}, \ii{pc}) \xrightarrow{\epsilon}
      (\sigma, \ii{mem}, \ii{reg}, \ii{pc'})
}}
\\

\inferrule[BnzNZ]{
  \trg{\ii{fetch}(E, \ii{pc})} = \trg{ Bnz\ r\ L}\\
  \trg{\ii{reg}[r]} = \trg{Int\ z}\\
  \trg{z} \neq 0\\
  \trg{\ii{find\_label}(E, \ii{pc}, L)} = \trg{\ii{pc'}}\\
}{
  \trg{E \vdash (\sigma, \ii{mem}, \ii{reg}, \ii{pc}) \xrightarrow{\epsilon}
      (\sigma, \ii{mem}, \ii{reg}, \ii{pc'})
}}
\\

\inferrule[BnzZ]{
  \trg{\ii{fetch}(E, \ii{pc})} = \trg{ Bnz\ r\ L}\\
  \trg{\ii{reg}[r]} = \trg{Int\ z}\\
  \trg{z} = 0\\
}{
  \trg{E \vdash (\sigma, \ii{mem}, \ii{reg}, \ii{pc}) \xrightarrow{\epsilon}
      (\sigma, \ii{mem}, \ii{reg}, \ii{pc} + 1)
}}
\\

\inferrule[PtrOfLabel]{
  \trg{\ii{fetch}(E, \ii{pc})} = \trg{ PtrOfLabel\ L\ \texttt{->}\  r_d}\\
  \trg{\ii{find\_label}(E, \ii{pc}, L)} = \trg{\ii{ptr}}\\
  \trg{\ii{reg'}} = \trg{\ii{reg}[r_d \mapsto ptr]}\\
  \trg{\alpha} = \trg{\mathtt{Const}~\ii{mem~reg'~\ii{comp}(pc)~ptr}~r_d}
}{
  \trg{E \vdash (\sigma, \ii{mem}, \ii{reg}, \ii{pc}) \xrightarrow{\alpha}
      (\sigma, \ii{mem}, \ii{reg'}, \ii{pc} + 1)
}}
\\

\inferrule[JumpFunPtr]{
  \trg{\ii{fetch}(E, \ii{pc})} = \trg{JumpFunPtr\ r}\\
  \trg{\ii{pc'}} = \trg{\ii{reg}[r]} \\
  \trg{\ii{is\_code\_pointer(pc')}} \\
  \trg{\ii{comp}(\ii{pc})} = \trg{\ii{comp}(\ii{pc'})} \\
  \trg{\ii{ptr\_offset}(\ii{pc})} = 3 %TODO: explain this?
}{
  \trg{E \vdash (\sigma, \ii{mem}, \ii{reg}, \ii{pc}) \xrightarrow{\epsilon}
      (\sigma, \ii{mem}, \ii{reg}, \ii{pc'})
}}
\end{mathpar}
\caption{Small-step semantics of \targetlanguage}
\label{fig:app-target-red}
\end{figure}

\begin{figure}
\begin{mathpar}

\inferrule[Load]{
  \trg{\ii{fetch}(E, \ii{pc})} = \trg{Load\  *r_p\ \texttt{->}\  r_d}\\
  \trg{\ii{ptr}} = \trg{\ii{reg}[r_p]}\\
  \trg{v} = \trg{\ii{mem}[ptr]}\\
  \trg{\ii{reg'}} = \trg{\ii{reg}[r_d \mapsto v]}\\
  \trg{\alpha} = \trg{\mathtt{Load}~\mathit{mem~reg'~\ii{comp}(pc)}~r_p~r_d}
}{
  \trg{E \vdash (\sigma, \ii{mem}, \ii{reg}, \ii{pc}) \xrightarrow{\alpha}
      (\sigma, \ii{mem}, \ii{reg'}, \ii{pc} + 1)
}}
\\

\inferrule[Store]{
  \trg{\ii{fetch}(E, \ii{pc})} = \trg{Store\  *r_p\ \texttt{<-}\  r_s}\\
  \trg{\ii{ptr}} = \trg{\ii{reg}[r_p]}\\
  \trg{v} = \trg{\ii{reg}[r_s]}\\
  \trg{\ii{mem'}} = \trg{\ii{mem}[\ii{ptr} \mapsto v]}\\
  \trg{\alpha} = \trg{\mathtt{Store}~\mathit{mem'~reg~\ii{comp}(pc)}~r_p~r_s}
}{
  \trg{E \vdash (\sigma, \ii{mem}, \ii{reg}, \ii{pc}) \xrightarrow{\alpha}
      (\sigma, \ii{mem'}, \ii{reg}, \ii{pc} + 1)
}}
\\

\inferrule[Alloc]{
  \trg{\ii{fetch}(E, \ii{pc})} = \trg{Alloc\ r_1\ r_2}\\
  \trg{Int\ z} = \trg{\ii{reg}[r_2]}\\
  \trg{z} > 0\\
  \trg{(\ii{mem'}, \ii{ptr})} = \com{mem\_alloc}(\trg{\ii{mem}}, \ii{comp}(\trg{\ii{pc}}), z)\\
  \trg{\ii{reg'}} = \trg{\ii{reg}[r_1 \mapsto \ii{ptr}]}\\
  \trg{\alpha} = \trg{\mathtt{Alloc}~\mathit{mem'~reg'~\ii{comp}(pc)}~r_1~r_2}
}{
  \trg{E \vdash (\sigma, \ii{mem}, \ii{reg}, \ii{pc}) \xrightarrow{\alpha}
      (\sigma, \ii{mem'}, \ii{reg'}, \ii{pc} + 1)
}}
\\
\end{mathpar}
\caption{Small-step semantics of \targetlanguage: memory operations}
\label{fig:app-target-red-mem}
\end{figure}

\begin{figure}
\begin{mathpar}

\inferrule[Call]{
  \trg{\ii{fetch}(E, \ii{pc})} = \trg{Call\ c\ func}\\
  \trg{c} \neq \trg{\ii{comp}(\ii{pc})}\\
  \trg{func} \in \trg{c.\ii{import}}\\
  \trg{\ii{entry}(E,c,func)} = \trg{\ii{pc}'} \\
  \trg{\ii{reg'}} = \trg{\ii{invalidate}(\ii{reg})}\\
  \trg{\alpha} = \trg{\mathtt{dfCall}~\mathit{mem~reg'~\ii{comp}(pc)~\ii{comp}(pc')}.func(\ii{reg}[COM])}
}{
  \trg{E \vdash (\sigma, \ii{mem}, \ii{reg}, \ii{pc}) \xrightarrow{\alpha}
      ((\ii{pc}+1) :: \sigma, \ii{mem}, \ii{reg'}, \ii{pc'})
}}
\\

\inferrule[Return]{
  \trg{\ii{fetch}(E, \ii{pc})} = \trg{Return}\\
  \trg{\ii{comp}(\ii{pc'})} \neq \trg{\ii{comp}(\ii{pc})}\\
  \trg{\ii{reg'}} = \trg{\ii{invalidate}(\ii{reg})}\\
  \trg{\alpha} = \trg{\mathtt{dfRet}~\mathit{mem~reg'~\ii{comp}(pc)~\ii{comp}(pc')}~\ii{reg}[COM])}
}{
  \trg{E \vdash (\ii{pc'} :: \sigma, \ii{mem}, \ii{reg}, \ii{pc}) \xrightarrow{\alpha}
      (\sigma, \ii{mem}, \ii{reg'}, \ii{pc'})
}}
\\
\end{mathpar}
\caption{Small-step semantics of \targetlanguage: calls and returns}
\label{fig:app-target-red-call-ret}
\end{figure}

\clearpage
\section{Output of the data-flow back-translation per event}
\label{app:dataflow-backtrans-example}
% !TEX root = ./paper.tex

\Cref{fig:backtrans-table} shows the 
back-translation of each data-flow event
$\mathcal{E}$.

\begin{figure}[h]
\caption{Data-flow back-translation}
\label{fig:backtrans-table}
\begin{calstable}
\colwidths{{6cm}{12cm}}
\brow
\cell{Data-flow event $\mathcal{E}$}
\cell{Corresponding
	\sourcelanguage expression (within currently
	executing procedure \src{P})}
\erow

\brow
\cell{$\mathtt{dfCall}~\mathit{Mem~Reg~c_{caller}~c_{callee}.proc(v)}$}
\cell{
\begin{align*}
&\mathsf{EXTCALL}\ \src{:= 1;}\\
&\mathsf{(loc\_of\_reg\ r_{COM})}\ \src{:= (c_{callee}.proc(!(\bl{\mathsf{loc\_of\_reg\  r_{COM}}})));}\\
&\mathsf{invalidate\_metadata}\src{;}\\
&\mathsf{EXTCALL}\ \src{:= 0;}\\
&\src{c_{caller}.P(0)}
\end{align*}
}
\erow

\brow
\cell{$\mathtt{dfRet}~\mathit{Mem~Reg~c_{prev}~c_{next}~v}$}
\cell{
\begin{align*}
&\mathsf{EXTCALL}\ \src{:= 1;}\\
&\src{!(\bl{\mathsf{loc\_of\_reg\  r_{COM}}})))}
\end{align*}
}
\erow

\brow
\cell{$\mathtt{Const}~\mathit{Mem~Reg~c_{cur}~v~r_{dest}}$}
\cell{
\begin{align*}
&\mathsf{(loc\_of\_reg\ \mathit{r_{dest}})}\ \src{:= \bl{\mathsf{(expr\_of\_constval\ \mathit{v})}};}\\
&\src{c_{cur}.P(0)}
\end{align*}
}

\erow

\brow
\cell{$\mathtt{Mov}~\mathit{Mem~Reg~c_{cur}~r_{src}~r_{dest}}$}
\cell{
\begin{align*}
&\mathsf{(loc\_of\_reg\ \mathit{r_{dest}})}\  
\src{:=\ !}\ \mathsf{(loc\_of\_reg\ \mathit{r_{src}})}\src{;}\\
&\src{c_{cur}.P(0)}
\end{align*}
}
\erow

\brow
\cell{$\mathtt{BinOp}~\mathit{Mem~Reg~c_{cur}~\otimes~r_{src1}~r_{src2}~r_{dest}}$}
\cell{
\begin{align*}
&\mathsf{(loc\_of\_reg\ \mathit{r_{dest}})}\  
\src{:=}\ (\src{!}\ \mathsf{(loc\_of\_reg\ \mathit{r_{src1}})})\ \src{\otimes}\ 
(\src{!}\ \mathsf{(loc\_of\_reg\ \mathit{r_{src2}})})\src{;}\\
&\src{c_{cur}.P(0)}
\end{align*}
}
\erow

\brow
\cell{$\mathtt{Load}~\mathit{Mem~Reg~c_{cur}~r_{addr}~r_{dest}}$}
\cell{
\begin{align*}
&\mathsf{(loc\_of\_reg\ \mathit{r_{dest}})}\  
\src{:=}\ \src{!\ !}\ \mathsf{(loc\_of\_reg\ \mathit{r_{addr}})}\src{;}\\
&\src{c_{cur}.P(0)}
\end{align*}
}
\erow

\brow
\cell{$\mathtt{Store}~\mathit{Mem~Reg~c_{cur}~r_{addr}~r_{src}}$}
\cell{
\begin{align*}
&\src{!}\ \mathsf{(loc\_of\_reg\ \mathit{r_{addr}})}\  
\src{:=}\ \src{!}\ \mathsf{(loc\_of\_reg\ \mathit{r_{src}})}\src{;}\\
&\src{c_{cur}.P(0)}
\end{align*}
}
\erow

\brow
\cell{$\mathtt{Alloc}~\mathit{Mem~Reg~c_{cur}~r_{ptr}~r_{size}}$}
\cell{
\begin{align*}
&\mathsf{(loc\_of\_reg\ \mathit{r_{ptr}})}\  
\src{:=}\ \src{alloc}\ (\src{!}\ \mathsf{(loc\_of\_reg\ \mathit{r_{size}})})\src{;}\\
&\src{c_{cur}.P(0)}
\end{align*}
}
\erow

\end{calstable}
\end{figure}

Here are some definitions of meta-level 
functions that appear in the right column
of \Cref{fig:backtrans-table}.

\begin{itemize}
\item $\mathsf{EXTCALL} \defeq \src{local\ +\ 1}$\\ (This is a fixed location in the local
buffer in which we store a flag that keeps
track of every time control is exiting the
component.)
\item $\mathsf{loc\_of\_reg\ r}
\defeq \src{local\ +\ \bl{offset(r)}}$
\\ (These are fixed locations in the
local buffer reserved for registers we 
simulate. There are 7 registers in total.)
\item
\begin{align*}
\mathsf{invalidate\_metadata} \defeq&\ \\ &\mathsf{loc\_of\_reg\ r_1}\ \src{:=}\ \mathsf{dummy\_value}\src{;}\\
&\mathsf{loc\_of\_reg\ r_{AUX1}}\ \src{:=}\ \mathsf{dummy\_value}\src{;}\\
&\mathsf{loc\_of\_reg\ r_{AUX2}}\ \src{:=}\ \mathsf{dummy\_value}\src{;}\\
&\mathsf{loc\_of\_reg\ r_{SP}}\ \src{:=}\ \mathsf{dummy\_value}\src{;}\\
&\mathsf{loc\_of\_reg\ r_{RA}}\ \src{:=}\ \mathsf{dummy\_value}\src{;}\\
&\mathsf{loc\_of\_reg\ r_{ARG}}\ \src{:=}\ \mathsf{dummy\_value}
\end{align*}
(These assignments clear the simulated
registers to mimic
the secure \trg{Call} semantics of
\targetlanguage, which \trg{\mathit{invalidate}}s the register
file upon cross-component calls and returns.
Notice that the simulated location of
register $\mathsf{r_{COM}}$ does \emph{not} appear in the list of clearing assignments; 
when executing $\mathsf{invalidate\_metadata}$,
our the \sourcelanguage program will have 
already filled $\mathsf{r_{COM}}$ with the
return value of the call.)
\end{itemize}

\twocolumn
\else
\fi % appendices

% BIBLIOGRAPHY ALWAYS AT THE END!
\ifieee
\bibliographystyle{abbrvnaturl}
%\footnotesize
\else %acm

\ifcamera
\bibliographystyle{ACM-Reference-Format}
\citestyle{acmauthoryear}   %% For author/year citations
\else
\bibliographystyle{abbrvnaturl}
\fi

\fi %ieee/acm

%% \bibliography{local}
\bibliography{paper}

\begin{thebibliography}{52}
\providecommand{\natexlab}[1]{#1}

\bibitem[Abadi(1998)]{abadiprotectiontranslations}
M.~Abadi.
\newblock Protection in programming-language translations.
\newblock \iflongrefs{In \emph{International Colloquium on Automata, Languages,
  and Programming}}\else{\emph{ICALP}}\fi{}. Springer, 1998.

\bibitem[Abate et~al.(2019{\natexlab{a}})Abate, {Azevedo de Amorim}, Blanco,
  Evans, Fachini, Hritcu, Laurent, Pierce, Stronati, Thibault, and
  Tolmach]{Abate:2018:GCG:3243734.3243745}
C.~Abate, A.~{Azevedo de Amorim}, R.~Blanco, A.~N. Evans, G.~Fachini,
  C.~Hritcu, T.~Laurent, B.~C. Pierce, M.~Stronati, J.~Thibault, and
  A.~Tolmach.
\newblock \href {http://arxiv.org/abs/1802.00588} {When good components go bad:
  Formally secure compilation despite dynamic compromise}.
\newblock arXiv:1802.00588 v5 (previous version appeared at CCS'18),
  2019{\natexlab{a}}.

\bibitem[Abate et~al.(2019{\natexlab{b}})Abate, Blanco, Garg, Hritcu,
  Patrignani, and Thibault]{abate2019journey}
C.~Abate, R.~Blanco, D.~Garg, C.~Hritcu, M.~Patrignani, and J.~Thibault.
\newblock Journey beyond full abstraction: Exploring robust property
  preservation for secure compilation.
\newblock \iflongrefs{In \emph{32nd {IEEE} Computer Security Foundations
  Symposium, {CSF} 2019, Hoboken, NJ, USA, June 25-28,
  2019}}\else{\emph{CSF}}\fi{}, 2019{\natexlab{b}}.

\bibitem[Abate et~al.(2020)Abate, Blanco, Ciob{\^{a}}c{\u{a}}, Durier, Garg,
  Hritcu, Patrignani, Tanter, and Thibault]{DBLP:conf/esop/AbateBCD0HPTT20}
C.~Abate, R.~Blanco, {\c{S}}.~Ciob{\^{a}}c{\u{a}}, A.~Durier, D.~Garg,
  C.~Hritcu, M.~Patrignani, {\'{E}}.~Tanter, and J.~Thibault.
\newblock Trace-relating compiler correctness and secure compilation.
\newblock \iflongrefs{In P.~M{\"{u}}ller, editor, \emph{29th European Symposium
  on Programming, {ESOP}}}\else{\emph{ESOP}}\fi{}. 2020.

\bibitem[Ahmed and Blume()]{AhmedFa}
A.~Ahmed and M.~Blume.
\newblock \href {http://dx.doi.org/10.1145/1411203.1411227} {Typed closure
  conversion preserves observational equivalence}.
\newblock \emph{SIGPLAN Not.'08}.

\bibitem[Ahmed and Blume(2011)]{ahmedCPS}
A.~Ahmed and M.~Blume.
\newblock \href {http://dx.doi.org/10.1145/2034574.2034830} {An
  equivalence-preserving {CPS} translation via multi-language semantics}.
\newblock \emph{SIGPLAN Not.}, 46\penalty0 (9)\iflongrefs{:\penalty0
  431--444}\fi{}, \iflongrefs{Sept. }\fi{}2011.

\bibitem[Barthe et~al.(2004)Barthe, Basu, and
  Rezk]{10.1007/978-3-540-24622-0_2}
G.~Barthe, A.~Basu, and T.~Rezk.
\newblock Security types preserving compilation.
\newblock \iflongrefs{In B.~Steffen and G.~Levi, editors, \emph{Verification,
  Model Checking, and Abstract Interpretation}}\else{\emph{VMCAI}}\fi{}. 2004.

\bibitem[Barthe et~al.(2018)Barthe, Gr{\'e}goire, and
  Laporte]{barthe2018secure}
G.~Barthe, B.~Gr{\'e}goire, and V.~Laporte.
\newblock Secure compilation of side-channel countermeasures: the case of
  cryptographic “constant-time”.
\newblock \iflongrefs{In \emph{2018 IEEE 31st Computer Security Foundations
  Symposium (CSF)}}\else{\emph{CSF}}\fi{}. IEEE, 2018.

\bibitem[Barthe et~al.(2019)Barthe, Blazy, Gr\'{e}goire, Hutin, Laporte,
  Pichardie, and Trieu]{10.1145/3371075}
G.~Barthe, S.~Blazy, B.~Gr\'{e}goire, R.~Hutin, V.~Laporte, D.~Pichardie, and
  A.~Trieu.
\newblock \href {http://dx.doi.org/10.1145/3371075} {Formal verification of a
  constant-time preserving {C} compiler}.
\newblock \iflongrefs{\emph{Proc. ACM Program.
  Lang.}}\else{\emph{PACMPL}}\fi{}, 4\penalty0 (POPL), \iflongrefs{Dec.
  }\fi{}2019.

\bibitem[Barthe et~al.(2021)Barthe, Blazy, Hutin, and
  Pichardie]{constant-resource-secure-compilation}
G.~Barthe, S.~Blazy, R.~Hutin, and D.~Pichardie.
\newblock \href {http://dx.doi.org/10.1109/CSF51468.2021.00020} {Secure
  compilation of constant-resource programs}.
\newblock \iflongrefs{In \emph{2021 2021 IEEE 34th Computer Security
  Foundations Symposium (CSF)}}\else{\emph{CSF}}\fi{}. \iflongrefs{jun
  }\fi{}2021.

\bibitem[Besson et~al.(2018)Besson, Dang, and Jensen]{besson2018securing}
F.~Besson, A.~Dang, and T.~Jensen.
\newblock Securing compilation against memory probing.
\newblock \iflongrefs{In \emph{Proceedings of the 13th Workshop on Programming
  Languages and Analysis for Security}}\else{\emph{PLAS}}\fi{}. ACM, 2018.

\bibitem[Busi et~al.(2020)Busi, Noorman, Van~Bulck, Galletta, Degano,
  M{\"u}hlberg, and Piessens]{busi2020provably}
M.~Busi, J.~Noorman, J.~Van~Bulck, L.~Galletta, P.~Degano, J.~T. M{\"u}hlberg,
  and F.~Piessens.
\newblock Provably secure isolation for interruptible enclaved execution on
  small microprocessors.
\newblock \iflongrefs{In \emph{2020 IEEE 33rd Computer Security Foundations
  Symposium (CSF)}}\else{\emph{CSF}}\fi{}. IEEE, 2020.

\bibitem[De~Amorim et~al.(2015)De~Amorim, D{\'e}nes, Giannarakis, Hritcu,
  Pierce, Spector-Zabusky, and Tolmach]{micropolicies}
A.~A. De~Amorim, M.~D{\'e}nes, N.~Giannarakis, C.~Hritcu, B.~C. Pierce,
  A.~Spector-Zabusky, and A.~Tolmach.
\newblock Micro-policies: Formally verified, tag-based security monitors.
\newblock \iflongrefs{In \emph{Security and Privacy (SP), 2015 IEEE Symposium
  on}}\else{\emph{S\&P}}\fi{}. IEEE, 2015.

\bibitem[Devriese et~al.(2016)Devriese, Patrignani, and
  Piessens]{popl-backtrans}
D.~Devriese, M.~Patrignani, and F.~Piessens.
\newblock \href {http://dx.doi.org/10.1145/2837614.2837618} {Fully-abstract
  compilation by approximate back-translation}.
\newblock \iflongrefs{In \emph{Proceedings of the 43rd Annual {ACM}
  {SIGPLAN-SIGACT} Symposium on Principles of Programming Languages, {POPL}
  2016, St. Petersburg, FL, USA, January 20 - 22,
  2016}}\else{\emph{POPL}}\fi{}, 2016.

\bibitem[Devriese et~al.(2017)Devriese, Patrignani, Piessens, and
  Keuchel]{DBLP:journals/lmcs/DevriesePPK17}
D.~Devriese, M.~Patrignani, F.~Piessens, and S.~Keuchel.
\newblock \href {http://dx.doi.org/10.23638/LMCS-13(4:2)2017} {Modular,
  fully-abstract compilation by approximate back-translation}.
\newblock \iflongrefs{\emph{Log. Methods Comput. Sci.}}\else{\emph{LMCS}}\fi{},
  13\penalty0 (4), 2017.

\bibitem[Dhawan et~al.(2015)Dhawan, Hritcu, Rubin, Vasilakis, Chiricescu,
  Smith, Knight, Pierce, and DeHon]{pump}
U.~Dhawan, C.~Hritcu, R.~Rubin, N.~Vasilakis, S.~Chiricescu, J.~M. Smith, T.~F.
  Knight, Jr., B.~C. Pierce, and A.~DeHon.
\newblock \href {http://dx.doi.org/10.1145/2786763.2694383} {Architectural
  support for software-defined metadata processing}.
\newblock \emph{SIGARCH Comput. Archit. News}, 43\penalty0
  (1)\iflongrefs{:\penalty0 487--502}\fi{}, \iflongrefs{Mar. }\fi{}2015.

\bibitem[El-Korashy et~al.(2021)El-Korashy, Tsampas, Patrignani, Devriese,
  Garg, and Piessens]{capableptrs}
A.~El-Korashy, S.~Tsampas, M.~Patrignani, D.~Devriese, D.~Garg, and
  F.~Piessens.
\newblock \href {http://dx.doi.org/10.1109/CSF51468.2021.00036} {{CapablePtrs}:
  Securely compiling partial programs using the pointers-as-capabilities
  principle}.
\newblock \iflongrefs{In \emph{2021 IEEE 34th Computer Security Foundations
  Symposium (CSF)}}\else{\emph{CSF}}\fi{}, 2021.

\bibitem[Gordon and Jeffrey(2003)]{gordon2003authenticity}
A.~D. Gordon and A.~Jeffrey.
\newblock Authenticity by typing for security protocols 1.
\newblock \iflongrefs{\emph{Journal of computer
  security}}\else{\emph{JCS}}\fi{}, 11\penalty0 (4)\iflongrefs{:\penalty0
  451--519}\fi{}, 2003.

\bibitem[Grumberg and Long(1991)]{grumberg1991model}
O.~Grumberg and D.~E. Long.
\newblock Model checking and modular verification.
\newblock \iflongrefs{In \emph{International Conference on Concurrency
  Theory}}\else{\emph{CONCUR}}\fi{}. Springer, 1991.

\bibitem[Jia et~al.(2015)Jia, Sen, Garg, and Datta]{7243751}
L.~Jia, S.~Sen, D.~Garg, and A.~Datta.
\newblock \href {http://dx.doi.org/10.1109/CSF.2015.38} {A logic of programs
  with interface-confined code}.
\newblock \iflongrefs{In \emph{2015 IEEE 28th Computer Security Foundations
  Symposium}}\else{\emph{CSF}}\fi{}, 2015.

\bibitem[Kang et~al.(2016)Kang, Kim, Hur, Dreyer, and Vafeiadis]{KangKHDV15}
J.~Kang, Y.~Kim, C.-K. Hur, D.~Dreyer, and V.~Vafeiadis.
\newblock \href {http://dx.doi.org/10.1145/2914770.2837642} {Lightweight
  verification of separate compilation}.
\newblock \emph{SIGPLAN Not.}, 51\penalty0 (1)\iflongrefs{:\penalty0
  178–190}\fi{}, \iflongrefs{Jan. }\fi{}2016.

\bibitem[Kumar et~al.(2014)Kumar, Myreen, Norrish, and Owens]{KumarMNO14}
R.~Kumar, M.~O. Myreen, M.~Norrish, and S.~Owens.
\newblock \href {http://dx.doi.org/10.1145/2535838.2535841} {{CakeML}: a
  verified implementation of {ML}}.
\newblock \iflongrefs{In \emph{41st Annual {ACM} {SIGPLAN-SIGACT} Symposium on
  Principles of Programming Languages}}\else{\emph{POPL}}\fi{}. 2014.

\bibitem[Kupferman and Vardi(1999)]{10.1007/3-540-48320-9_27}
O.~Kupferman and M.~Y. Vardi.
\newblock Robust satisfaction.
\newblock \iflongrefs{In J.~C.~M. Baeten and S.~Mauw, editors, \emph{CONCUR'99
  Concurrency Theory}}\else{\emph{CONCUR}}\fi{}. 1999.

\bibitem[Laird(2007)]{9249fdc52dd3414e83b2e3f9a89cb117}
J.~Laird.
\newblock \href {http://dx.doi.org/10.1007/978-3-540-73420-8_58} {\emph{A fully
  abstract trace semantics for general references}}, pages 667--679.
\newblock Springer Verlag, \iflongrefs{7 }\fi{}2007.

\bibitem[Leroy(2009)]{Leroy2009}
X.~Leroy.
\newblock \href {http://dx.doi.org/10.1007/s10817-009-9155-4} {A formally
  verified compiler back-end}.
\newblock \emph{Journal of Automated Reasoning}, 43\penalty0
  (4)\iflongrefs{:\penalty0 363}\fi{}, \iflongrefs{Nov }\fi{}2009.

\bibitem[Leroy and Blazy(2008)]{10.1007/s10817-008-9099-0}
X.~Leroy and S.~Blazy.
\newblock \href {http://dx.doi.org/10.1007/s10817-008-9099-0} {Formal
  verification of a {C}-like memory model and its uses for verifying program
  transformations}.
\newblock \emph{J. Autom. Reason.}, 41\penalty0 (1)\iflongrefs{:\penalty0
  1–31}\fi{}, \iflongrefs{July }\fi{}2008.

\bibitem[Milner and Weyhrauch(1972)]{milner1972proving}
R.~Milner and R.~Weyhrauch.
\newblock Proving compiler correctness in a mechanized logic.
\newblock \emph{Machine Intelligence}, 1972.

\bibitem[Morris(1973)]{protectioninpl}
J.~H. Morris, Jr.
\newblock \href {http://dx.doi.org/10.1145/361932.361937} {Protection in
  programming languages}.
\newblock \emph{Commun. ACM}, 16\penalty0 (1)\iflongrefs{:\penalty0
  15--21}\fi{}, \iflongrefs{Jan. }\fi{}1973.

\bibitem[Myreen et~al.(2008)Myreen, Gordon, and
  Slind]{Myreen:2008:MVM:1517424.1517444}
M.~O. Myreen, M.~J.~C. Gordon, and K.~Slind.
\newblock \href {http://dl.acm.org/citation.cfm?id=1517424.1517444}
  {Machine-code verification for multiple architectures: An application of
  decompilation into logic}.
\newblock \iflongrefs{In \emph{Proceedings of the 2008 International Conference
  on Formal Methods in Computer-Aided Design}}\else{\emph{FMCAD}}\fi{}. 2008.

\bibitem[Namjoshi and Tabajara(2020)]{Namjoshi-model-checking-ct-compilation}
K.~S. Namjoshi and L.~M. Tabajara.
\newblock Witnessing secure compilation.
\newblock \iflongrefs{In D.~Beyer and D.~Zufferey, editors, \emph{Verification,
  Model Checking, and Abstract Interpretation}}\else{\emph{VMCAI}}\fi{}. 2020.

\bibitem[Neis et~al.(2015)Neis, Hur, Kaiser, McLaughlin, Dreyer, and
  Vafeiadis]{10.1145/2784731.2784764}
G.~Neis, C.-K. Hur, J.-O. Kaiser, C.~McLaughlin, D.~Dreyer, and V.~Vafeiadis.
\newblock \href {http://dx.doi.org/10.1145/2784731.2784764} {Pilsner: A
  compositionally verified compiler for a higher-order imperative language}.
\newblock \iflongrefs{In \emph{Proceedings of the 20th ACM SIGPLAN
  International Conference on Functional Programming}}\else{\emph{ICFP}}\fi{}.
  2015.

\bibitem[New et~al.(2016)New, Bowman, and Ahmed]{New}
M.~S. New, W.~J. Bowman, and A.~Ahmed.
\newblock \href {http://dx.doi.org/10.1145/2951913.2951941} {Fully abstract
  compilation via universal embedding}.
\newblock \iflongrefs{In \emph{Proceedings of the 21st ACM SIGPLAN
  International Conference on Functional Programming}}\else{\emph{ICFP}}\fi{}.
  2016.

\bibitem[Patrignani and Garg()]{10.1145/3436809}
M.~Patrignani and D.~Garg.
\newblock \href {http://dx.doi.org/10.1145/3436809} {Robustly safe compilation,
  an efficient form of secure compilation}.
\newblock \iflongrefs{\emph{ACM Trans. Program. Lang.
  Syst.}}\else{\emph{TOPLAS'21}}\fi{}.

\bibitem[Patrignani et~al.()Patrignani, Ahmed, and
  Clarke]{Patrignani:2019:FAS:3303862.3280984}
M.~Patrignani, A.~Ahmed, and D.~Clarke.
\newblock \href {http://dx.doi.org/10.1145/3280984} {Formal approaches to
  secure compilation: A survey of fully abstract compilation and related work}.
\newblock \emph{ACM Comput. Surv.'19}.

\bibitem[Patrignani et~al.(2015)Patrignani, Agten, Strackx, Jacobs, Clarke, and
  Piessens]{patrignani2015secure}
M.~Patrignani, P.~Agten, R.~Strackx, B.~Jacobs, D.~Clarke, and F.~Piessens.
\newblock Secure compilation to protected module architectures.
\newblock \iflongrefs{\emph{ACM Transactions on Programming Languages and
  Systems (TOPLAS)}}\else{\emph{TOPLAS}}\fi{}, 37\penalty0
  (2)\iflongrefs{:\penalty0 6}\fi{}, 2015.

\bibitem[Patterson and Ahmed(2019)]{patterson2019next}
D.~Patterson and A.~Ahmed.
\newblock \href {http://dx.doi.org/10.1145/3341689} {The next 700 compiler
  correctness theorems (functional pearl)}.
\newblock \iflongrefs{\emph{Proc. ACM Program.
  Lang.}}\else{\emph{PACMPL}}\fi{}, 3\penalty0 (ICFP), \iflongrefs{July
  }\fi{}2019.

\bibitem[Patterson et~al.(2017)Patterson, Perconti, Dimoulas, and
  Ahmed]{FunTAL}
D.~Patterson, J.~Perconti, C.~Dimoulas, and A.~Ahmed.
\newblock \href {http://dx.doi.org/10.1145/3140587.3062347} {Funtal: Reasonably
  mixing a functional language with assembly}.
\newblock \emph{SIGPLAN Not.}, 52\penalty0 (6)\iflongrefs{:\penalty0
  495–509}\fi{}, \iflongrefs{June }\fi{}2017.

\bibitem[Sammler et~al.()Sammler, Garg, Dreyer, and Litak]{10.1145/3371100}
M.~Sammler, D.~Garg, D.~Dreyer, and T.~Litak.
\newblock \href {http://dx.doi.org/10.1145/3371100} {The high-level benefits of
  low-level sandboxing}.
\newblock \iflongrefs{\emph{Proc. ACM Program.
  Lang.}}\else{\emph{POPL'19}}\fi{}.

\bibitem[Sison and Murray(2019)]{sison_et_al:LIPIcs:2019:11082}
R.~Sison and T.~Murray.
\newblock \href {http://dx.doi.org/10.4230/LIPIcs.ITP.2019.27} {{Verifying That
  a Compiler Preserves Concurrent Value-Dependent Information-Flow Security}}.
\newblock \iflongrefs{In J.~Harrison, J.~O'Leary, and A.~Tolmach, editors,
  \emph{10th International Conference on Interactive Theorem Proving (ITP
  2019)}}\else{\emph{ITP}}\fi{}. 2019.

\bibitem[Sison and Murray(2021)]{sison_murray_2021}
R.~Sison and T.~Murray.
\newblock \href {http://dx.doi.org/10.1017/S0956796821000162} {Verified secure
  compilation for mixed-sensitivity concurrent programs}.
\newblock \iflongrefs{\emph{Journal of Functional
  Programming}}\else{\emph{JFP}}\fi{}, 2021.

\bibitem[Skorstengaard et~al.(2019)Skorstengaard, Devriese, and
  Birkedal]{Skorstengaard:2019:SEW:3302515.3290332}
L.~Skorstengaard, D.~Devriese, and L.~Birkedal.
\newblock \href {http://dx.doi.org/10.1145/3290332} {{StkTokens}: Enforcing
  well-bracketed control flow and stack encapsulation using linear
  capabilities}.
\newblock \iflongrefs{\emph{Proc. ACM Program.
  Lang.}}\else{\emph{PACMPL}}\fi{}, 3\penalty0 (POPL)\iflongrefs{:\penalty0
  19:1--19:28}\fi{}, \iflongrefs{Jan. }\fi{}2019.

\bibitem[Song et~al.(2019)Song, Cho, Kim, Kim, Kang, and Hur]{10.1145/3371091}
Y.~Song, M.~Cho, D.~Kim, Y.~Kim, J.~Kang, and C.-K. Hur.
\newblock \href {http://dx.doi.org/10.1145/3371091} {{CompCertM}: {CompCert}
  with {C}-assembly linking and lightweight modular verification}.
\newblock \iflongrefs{\emph{Proc. ACM Program.
  Lang.}}\else{\emph{PACMPL}}\fi{}, 4\penalty0 (POPL), \iflongrefs{Dec.
  }\fi{}2019.

\bibitem[Stewart et~al.(2015)Stewart, Beringer, Cuellar, and
  Appel]{10.1145/2676726.2676985}
G.~Stewart, L.~Beringer, S.~Cuellar, and A.~W. Appel.
\newblock \href {http://dx.doi.org/10.1145/2676726.2676985} {Compositional
  compcert}.
\newblock \iflongrefs{In \emph{Proceedings of the 42nd Annual ACM
  SIGPLAN-SIGACT Symposium on Principles of Programming
  Languages}}\else{\emph{POPL}}\fi{}. 2015.

\bibitem[Strackx et~al.(2013)Strackx, Noorman, Verbauwhede, Preneel, and
  Piessens]{pma}
R.~Strackx, J.~Noorman, I.~Verbauwhede, B.~Preneel, and F.~Piessens.
\newblock Protected software module architectures.
\newblock pages 241--251. Springer, 2013.

\bibitem[Strydonck et~al.(2019)Strydonck, Piessens, and
  Devriese]{DBLP:journals/pacmpl/StrydonckPD19}
T.~V. Strydonck, F.~Piessens, and D.~Devriese.
\newblock \href {http://dx.doi.org/10.1145/3341688} {Linear capabilities for
  fully abstract compilation of separation-logic-verified code}.
\newblock \iflongrefs{\emph{Proc. {ACM} Program.
  Lang.}}\else{\emph{PACMLP}}\fi{}, 3\penalty0 ({ICFP})\iflongrefs{:\penalty0
  84:1--84:29}\fi{}, 2019.

\bibitem[Swasey et~al.(2017)Swasey, Garg, and Dreyer]{swasey2017robust}
D.~Swasey, D.~Garg, and D.~Dreyer.
\newblock Robust and compositional verification of object capability patterns.
\newblock \emph{Proc. ACM Program. Lang.}, 1\penalty0
  (OOPSLA)\iflongrefs{:\penalty0 89--1}\fi{}, 2017.

\bibitem[Tsampas et~al.(2019)Tsampas, Devriese, and
  Piessens]{TsampasStelios2019Tsfs}
S.~Tsampas, D.~Devriese, and F.~Piessens.
\newblock \href {http://dx.doi.org/10.1109/CSF.2019.00024} {Temporal safety for
  stack allocated memory on capability machines}.
\newblock \iflongrefs{In \emph{2019 IEEE 32nd Computer Security Foundations
  Symposium (CSF)}}\else{\emph{CSF}}\fi{}, 2019.

\bibitem[Watson et~al.(2015)Watson, Woodruff, Neumann, Moore, Anderson,
  Chisnall, Dave, Davis, Gudka, Laurie, et~al.]{chericompartment}
R.~N. Watson, J.~Woodruff, P.~G. Neumann, S.~W. Moore, J.~Anderson,
  D.~Chisnall, N.~Dave, B.~Davis, K.~Gudka, B.~Laurie, et~al.
\newblock {CHERI: A hybrid capability-system architecture for scalable software
  compartmentalization}.
\newblock \iflongrefs{In \emph{Security and Privacy (SP), 2015 IEEE Symposium
  on}}\else{\emph{S\&P}}\fi{}. IEEE, 2015.

\bibitem[Watson et~al.(2016)Watson, Norton, Woodruff, Moore, Neumann, Anderson,
  Chisnall, Davis, Laurie, Roe, Dave, Gudka, Joannou, Markettos, Maste,
  Murdoch, Rothwell, Son, and Vadera]{micro16}
R.~N.~M. Watson, R.~M. Norton, J.~Woodruff, S.~W. Moore, P.~G. Neumann,
  J.~Anderson, D.~Chisnall, B.~Davis, B.~Laurie, M.~Roe, N.~H. Dave, K.~Gudka,
  A.~Joannou, A.~T. Markettos, E.~Maste, S.~J. Murdoch, C.~Rothwell, S.~D. Son,
  and M.~Vadera.
\newblock \href {http://dx.doi.org/10.1109/MM.2016.84} {Fast protection-domain
  crossing in the {CHERI} capability-system architecture}.
\newblock \emph{IEEE Micro}, 36\penalty0 (5)\iflongrefs{:\penalty0
  38--49}\fi{}, \iflongrefs{Sept. }\fi{}2016.

\bibitem[Watson et~al.(2019)Watson, Neumann, Woodruff, Roe, Almatary, Anderson,
  Baldwin, Chisnall, Davis, Filardo, Joannou, Laurie, Markettos, Moore,
  Murdoch, Nienhuis, Norton, Richardson, Rugg, Sewell, Son, and
  Xia]{cheriisatech}
R.~N.~M. Watson, P.~G. Neumann, J.~Woodruff, M.~Roe, H.~Almatary, J.~Anderson,
  J.~Baldwin, D.~Chisnall, B.~Davis, N.~W. Filardo, A.~Joannou, B.~Laurie,
  A.~T. Markettos, S.~W. Moore, S.~J. Murdoch, K.~Nienhuis, R.~Norton,
  A.~Richardson, P.~Rugg, P.~Sewell, S.~Son, and H.~Xia.
\newblock \href {https://www.cl.cam.ac.uk/techreports/UCAM-CL-TR-927.pdf}
  {{Capability Hardware Enhanced RISC Instructions: CHERI Instruction-Set
  Architecture (Version 7)}}.
\newblock Technical Report UCAM-CL-TR-927, University of Cambridge, Computer
  Laboratory, \iflongrefs{June }\fi{}2019.

\bibitem[Woodruff et~al.(2014)Woodruff, Watson, Chisnall, Moore, Anderson,
  Davis, Laurie, Neumann, Norton, and Roe]{cheriageofrisk}
J.~Woodruff, R.~N. Watson, D.~Chisnall, S.~W. Moore, J.~Anderson, B.~Davis,
  B.~Laurie, P.~G. Neumann, R.~Norton, and M.~Roe.
\newblock \href {http://dx.doi.org/10.1145/2678373.2665740} {{The CHERI
  Capability Model: Revisiting RISC in an Age of Risk}}.
\newblock \emph{SIGARCH Comput. Archit. News}, 42\penalty0
  (3)\iflongrefs{:\penalty0 457--468}\fi{}, \iflongrefs{June }\fi{}2014.

\bibitem[Zhang and D'Hollander(2004)]{unstructured-to-structured}
F.~Zhang and E.~D'Hollander.
\newblock \href {http://dx.doi.org/10.1109/TSE.2004.1274043} {Using hammock
  graphs to structure programs}.
\newblock \iflongrefs{\emph{IEEE Transactions on Software
  Engineering}}\else{\emph{IEEE Trans. Softw. Eng.}}\fi{}, 2004.

\end{thebibliography}
% BIBLIOGRAPHY ALWAYS AT THE END!

\end{document}